\documentclass[lettersize,journal]{IEEEtran}
\usepackage{amsmath,amsfonts,amssymb}
\usepackage{algorithmic}
\usepackage{algorithm}
\usepackage{array}
\usepackage[caption=false,font=normalsize,labelfont=sf,textfont=sf]{subfig}
\usepackage{textcomp}
\usepackage{stfloats}
\usepackage{url}
\usepackage{verbatim}
\usepackage{graphicx}
\usepackage{cite}
\usepackage{color}
\newtheorem{remark}{Remark} 
\newenvironment*{proof}{{\it Proof.}}{\hfill $\square$\par} 
\newtheorem{lemma}{Lemma} 
\newtheorem{my_claim}{Claim} 
\newtheorem{thm}{Theorem} 
\newtheorem{corollary}{Corollary}
\usepackage{parskip}
\usepackage{makecell}
\usepackage{float}

\let\appendices\relax
\usepackage[title]{appendix} 

\allowdisplaybreaks[1]

\begin{document}

\title{Rapid Boundary Stabilization of Two-Dimensional Elastic Plates with In-Domain Aeroelastic Instabilities}

\author{Xingzhi Huang, Ji wang,
\thanks{Xingzhi Huang is with the School of Aerospace Engineering, Xiamen University, Xiamen, Fujian, China,
34520241151610@xmu.edu.cn}
\thanks{Ji wang is with the School of Aerospace Engineering, Xiamen University, Xiamen, Fujian, China,
jiwang9024@gmail.com}
}

\markboth{IEEE Transactions on Control Systems Technology}%
{Shell \MakeLowercase{\textit{et al.}}: A Sample Article Using IEEEtran.cls for IEEE Journals}

\IEEEpubid{0000--0000/00\$00.00~\copyright~2021 IEEE}

\maketitle

\begin{abstract}
Motivated by active wing flutter suppression in high-Mach-number flight, this paper presents a rapid boundary stabilization strategy for a two-dimensional PDE-modeled elastic plate with in-domain instabilities, where the exponential stability is achieved with a decay rate that can be arbitrarily assigned by the users. First, the aeroelastic system is modeled as two-dimensional coupled wave PDEs with internal anti-damping terms, derived by Piston theory and Hamilton's principle. Using Fourier series expansion, the 2-D problem is decomposed into a large-scale 1-D system, based on which full-state boundary feedback control is designed via PDE backstepping transformation. To enable output-feedback implementation, a state observer is further designed to estimate the distributed states over the two-dimensional spatial domain using the available boundary measurements. Through Lyapunov analysis, the exponential stability of the 2-D elastic plate PDE under the proposed boundary control is established with a designer-tunable decay rate. Numerical simulations verify the effectiveness of the control strategy in suppressing flow-induced vibrations in a 2-D elastic plate. 
\end{abstract}

\begin{IEEEkeywords}
Boundary control; Two-dimensional hyperbolic PDEs; Backstepping; Aeroelastic flutter suppression.
\end{IEEEkeywords}

\section{Introduction}

\subsection{Motivation}
\IEEEPARstart{M}{odern} flying-wing aircraft are characterized by low mass and low wing bending natural frequencies. In high-Mach-number flight regimes, unsteady aerodynamic loads strongly interact with the elastic wing dynamics, resulting in aeroelastic flutter. This instability severely constrains the flight envelope and degrades the mission capability of the aircraft. The underlying physics of the flexible wing under aerodynamic loading can be accurately captured by a two-dimensional elastic plate model featuring spatially destabilizing flow-induced terms, i.e., two-dimensional PDEs with in-domain instabilities. Most existing results on active wing vibration suppression approximate the two-dimensional PDE by a one-dimensional model \cite{flexible_wing,Nonhomogeneous,wing_unsteady,membrane,Adaptive}. In this paper, we treat the wing as a two-dimensional PDE system in both the dynamic modeling and the control system design.

\subsection{Boundary control of two-dimensional elastic plates}
Early work by Lagnese in 1989 proposed boundary actuation to stabilize elastic plates \cite{1989Plates}. Subsequently, Rao demonstrated that boundary control can effectively suppress transverse vibrations of elastic plates under certain boundary conditions \cite{BoundaryDamping}. Using Rao’s approach, Liu et al. \cite{LIU2005353} arrived at the same conclusions as Lagnese, but they only achieved asymptotic stability rather than exponential stability. Recent research by Bouhamed et al. \cite{BOUHAMED2024} presents the problem of optimal control of a nonlinear Kirchhoff plate equation by a bilinear control on the boundary. He and Zhang \cite{flexible_wing} studied vibration control of a nonlinear flexible wing. Kar et al. \cite{torsional} addressed bending-torsional vibration control of plates using $H_\infty$ methods. Robu et al. \cite{simultaneousH} presented active control of a plane wing's vibrations induced by the sloshing of
large masses of fuel inside a partly full tank. He et al. \cite{3dflexible} later developed a trajectory-tracking control for 3-D flexible wings. Heining et al. \cite{optimalQ} investigated optimal actuator placement in control of quasi-static elastic plates. However, most existing results focus on systems with inherent internal damping and therefore do not address the more challenging case of plates with in-domain instabilities that are unmatched with boundary control inputs. Boundary control of flexible or compliant structures exhibiting instabilities has been studied using the backstepping method. However, the available results are largely restricted to one-dimensional structures, such as compliant cables and flexible beams, as shown in Table \ref{tab:1}. Boundary control designs for 2-D plate structures with in-domain instability are rare.
\begin{table}[h!]
	\begin{center}
		\caption{Backstepping boundary control for compliant or flexible structures.}
        \label{tab:1}
		\begin{tabular}{lc} 
			\hline
			Category & References\\
			\hline
			Compliant Cables/Strings & \cite{remote_sensing},\cite{drilling},\cite{wave_drilling},\cite{Cable_elevator},\cite{deep_sea},\cite{Cable}\\
			Flexible Beams & \cite{Rapid},\cite{shear_beam},\cite{beam_design},\cite{uncontrolled_boundary},\cite{2025arXivW}\\
			Flexible Plates & This paper\\
			\hline
		\end{tabular}
	\end{center}
\end{table}

\IEEEpubidadjcol

\subsection{Higher-dimensional backstepping control of PDEs}

In recent decades, significant progress has been made in PDE control. Among various established methodologies, the backstepping approach has demonstrated notable effectiveness, offering a systematic framework for boundary feedback design, primarily for one-dimensional parabolic and hyperbolic systems. Extending backstepping to higher spatial dimensions presents considerable challenges, largely due to the increased complexity of the resulting kernel equations. Advances in this direction have been achieved by utilizing specific geometric symmetries and boundary conditions to simplify these equations. Initial breakthroughs in multidimensional control, particularly in fluid flow applications, made use of spatial invariance \cite{2002Distributed}. This approach transformed the original system into families of parameterized one-dimensional PDEs via Fourier transforms \cite{2007A}, a technique later employed in the control of convection loops \cite{2010Boundaryobserver} and magnetohydrodynamic flows \cite{2008Stabilization}. For domains exhibiting radial symmetry, such as disks \cite{2016Disk}, spheres \cite{2019Sphere}, and n-dimensional balls \cite{2015Ball}, explicit backstepping controllers have been derived using spherical harmonics and Bessel functions. More recent developments address systems with spatially varying coefficients on arbitrary-dimensional balls \cite{2016Explicit} and extend the methodology to three-dimensional multi-agent systems \cite{2015Multi,2024Multi,2024Delay}, as well as to PDEs coupled with lower-dimensional boundary dynamics \cite{2019Stabilization}. Besides, \cite{GUAN2023111242,GUAN2025} developed a bilateral delay-compensation control strategy for an unstable 2-D reaction-diffusion system with distinct input delays.
Despite these advances, existing theoretical results are largely confined to parabolic PDEs. In contrast, the active suppression of wing flutter in high-Mach-number flight regimes calls for boundary control frameworks of high-dimensional coupled wave PDEs with in-domain unstable sources, which remain largely unexplored.

\subsection{Main Contribution}

1) The work \cite{Rapid} addressed state-feedback boundary control of one-dimensional beam systems by rewriting it into a class of coupled hyperbolic PDEs \cite{local_exponential,Heterodirectional,n+1,2plus2,finite_2plus2}. Here, we address output-feedback boundary control of two-dimensional plate models, accounting for more complex couplings among three-directional vibrations.

2) Different from the boundary control of high-dimensional PDEs investigated in \cite{2015Multi,2024Multi,2024Delay,2007A,2016Explicit,2019Stabilization,GUAN2023111242,GUAN2025} for parabolic PDEs, this work deals with a group of coupled wave PDEs with instability sources in two dimensions.

3) To the best of our knowledge, this is the first result of rapid boundary stabilizing a two-dimensional PDE-modeled elastic plate with anti-stable sources in the spatial domain, where the decay rate can be arbitrarily assigned by users. 

\subsection{Notation}

\begin{itemize}
	\item[$\bullet$] The notation $|\cdot|$ denotes Euclidean norm. The notation $\dot{z}(t)$ denotes the time derivative of $z$. The notation $c^{(i)}(x)$ denote the $i$ times derivatives of $c$.
	\item[$\bullet$] Let $\Omega \subset \mathbb{R}^n$ be an open set, and let $x = (x_1,\dots,x_n) \in \Omega$ denote the spatial variable.
	The space $L^2(\Omega)$ consists of all measurable functions 
	$f : \Omega \to \mathbb{R}$ such that
	$\int_\Omega |f(x)|^2 \, dx < \infty$,
	with norm
	$
	\|f\|_{L^2(\Omega)}
	=\left( \int_\Omega |f(x)|^2 dx \right)^{1/2}$. The Sobolev space $H^1(\Omega)$ is defined as
	$H^1(\Omega)=\left\{f \in L^2(\Omega) :
	\partial_{x_i} f \in L^2(\Omega), i=1,\dots,n\right\}$, where $\partial_{x_i} f$ denotes the weak partial derivative with respect to $x_i$.
	The norm in $H^1(\Omega)$ is given by
	$\|f\|_{H^1(\Omega)}=\left(\|f\|_{L^2(\Omega)}^2+\sum_{i=1}^n
	\|\partial_{x_i} f\|_{L^2(\Omega)}^2
	\right)^{1/2}$.
\end{itemize}

\section{Modeling}
\subsection{Modeling of aerodynamic forces}
We consider an aeroelastic system shown in Fig.\ref{fig:model}, with the domain $0 < x^* < L_1^*, 0 < y^* < L_2^*, z^* = 0$, where $L_1^*$ and $L_2^*$ are the lengths of the plate in $x^*$, $y^*$ direction, respectively, and its thickness is $h^*$, so that it represents a two-dimensional elastic plate with flow in the $x^*$ direction.

The physical models used in treating fluid-structure interaction phenomena vary enormously in their complexity and range of applicability. The simplest model is the very popular "Piston theory". Based on Piston theory \cite{Machnumber}, the local aerodynamic pressure exerted by local fluid velocity normal to the elastic plate is given by
\begin{align}
	p^*(t^*,x^*,y^*) = \tfrac{\rho_f^* U^*}{M^*}\left(	w^*_{t^*} + U^*w^*_{x^*}\right),
\end{align}
where $w^*(t^*,x^*,y^*)$ is the displacement of the plate in $z^*$ direction at the axial location $(x^*, y^*)$ at time $t^*$, and ${\rho_f}^*$, $U^*$ and $M^*$ are the free stream density, velocity, and Mach number \cite{Machnumber}, respectively.
\begin{figure}
	\begin{center}
		\includegraphics[height=10.5cm, width = 8.5cm]{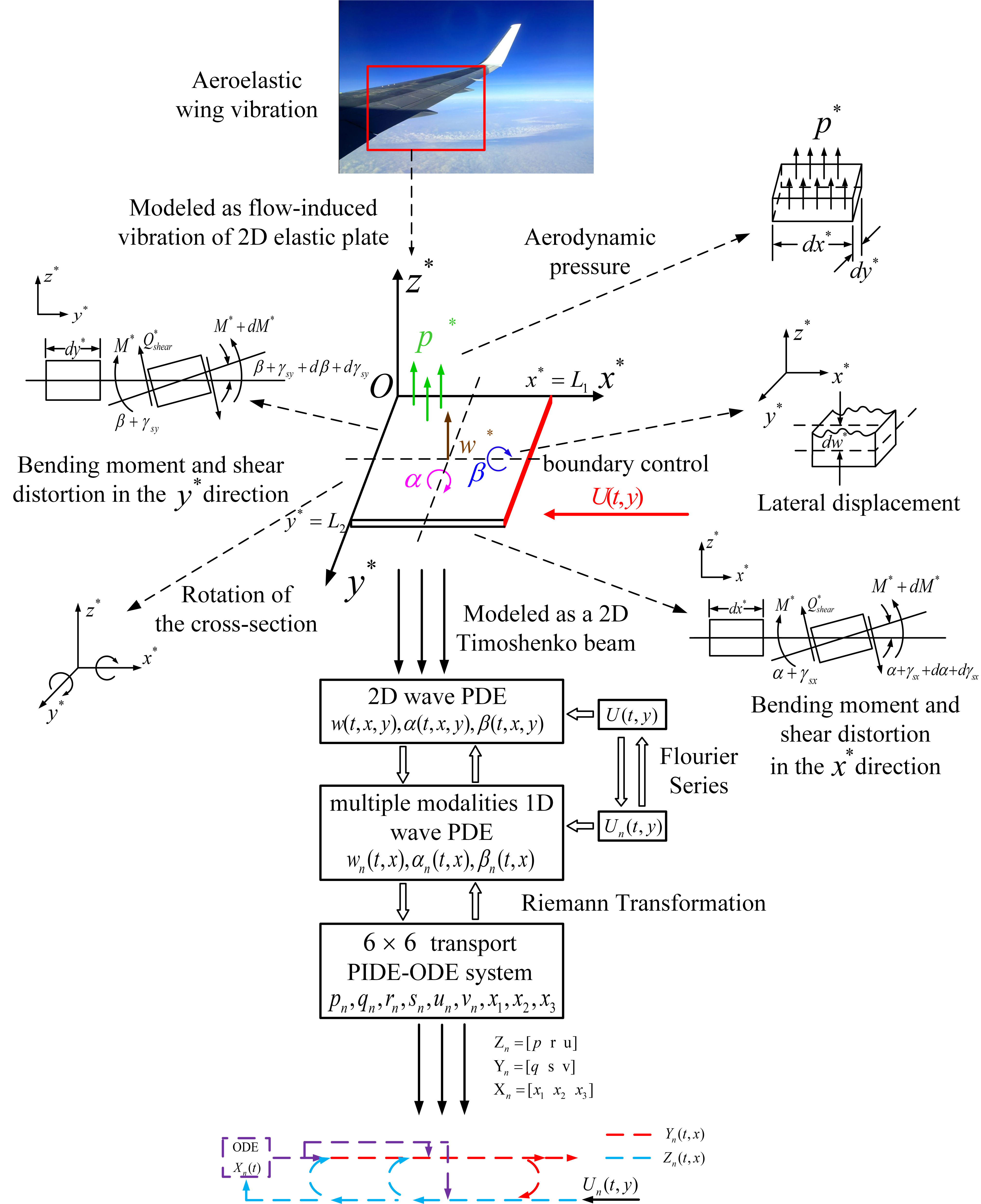}    
		\caption{Flow-induced vibration wing: from the physical model to the mathematical plant.}  
		\label{fig:model}                                 
	\end{center}                                 
\end{figure}

\subsection{Equations of motion and boundary conditions via Hamilton's Principle}

The modeling process follows the approach in \cite{dynamics}. The equation of motion is obtained using Hamilton's variation principle. According to Kirchhoff Plate Theory in \cite{plates} and ignore the influence of Poisson's ratio, the potential energy due to bending is given by
\begin{align}\label{eq_bend}
	\nonumber
	PE^{*}_{bending} &= \tfrac{D^*}{2}\smallint\nolimits_{0}^{L_1^*} \smallint\nolimits_{0}^{L_2^*}  [(w^*_{x^*x^*})^2 + (w^*_{y^*y^*})^2 \\
	&\quad + (w^*_{x^*y^*} + w^*_{y^*x^*})^2] dx^*dy^*,
\end{align}
where $D^* = \tfrac{E^*{h^*}^3}{12}$ presents the stiffness of a homogeneous plate of orthotropic material which can be defined as $E^*I^*$, where $E^*$ is the modulus of elasticity and $I^* = \tfrac{{h^*}^3}{12}$ is the equivalent cross-sectional moment of inertia about the neutral axis.
This model adds the effect of shear distortion (but not rotary inertia). We introduce variables $\alpha$ and $\beta$, representing the angle of rotation of the cross-section due to the bending moment in $x^*$ direction and $y^*$ direction respectively, and $\gamma_{sx}, \gamma_{sy}$, the angle of distortion due to shear in $x^*$ and $y^*$ direction.
The total angle  of rotation in $x^*$ direction is the sum of $\alpha$ and $\gamma_{sx}$ and the sum of $\beta$ and $\gamma_{sy}$ in $y^*$ direction, and is approximately the first derivative of the deflection:
\begin{align}
	\alpha + \gamma_{sx} = w^*_{x^*},\quad \beta + \gamma_{sy} = w^*_{y^*}.
\end{align}
Therefore, the potential energy due to bending given in equation \eqref{eq_bend} is slightly modified in this case such that
\begin{align}
	\nonumber
	PE^{*}_{bending} &= \tfrac{E^*I^*}{2}\smallint\nolimits_{0}^{L_1^*}\smallint\nolimits_{0}^{L_2^*}  [(\alpha_{x^*})^2 + (\beta_{y^*})^2\\
	&\quad + (\alpha_{y^*} + \beta_{x^*})^2] dx^*dy^*.
\end{align}
The potential energy due to shear is given by
\begin{align}
	\nonumber
	PE^*_{shear} &= \tfrac{k^{\prime}G^*h^*}{2}\smallint\nolimits_{0}^{L_1^*}\smallint\nolimits_{0}^{L_2^*} [(w^*_{x^*} - \alpha)^2\\
	&\quad + (w^*_{y^*} - \beta)^2] dx^*dy^*.
\end{align}  
where $k^{\prime}$ is the shape factor, $G^*$ is the shear modulus and $h^*$ is the thickness of the plate.
The kinetic energy due to displacement is given by
\begin{align}
	KE^*_{trans} = \tfrac{\rho^*h^*}{2}\smallint\nolimits_{0}^{L_1^*}\smallint\nolimits_{0}^{L_2^*} (w^*_{t^*})^2 dx^*dy^*,
\end{align}
where $\rho^*$ is the density of the elastic plate.
The kinetic energy due to the rotation of the cross-section is given by
\begin{align}
	\scalebox{0.9}{${KE}^*_{rot} = \tfrac{{\rho}^*}{2} \smallint\nolimits_{0}^{L_1^*}\smallint\nolimits_{0}^{L_2^*}[I_1^*(w^*_{x^*t^*})^2 + I_2^* (w^*_{y^*t^*})^2] dx^*dy^*$},
\end{align}
where ${I}_1^*$ and ${I}_2^*$ are the area moments of inertia of the plate in the $x$ and $y$ directions, respectively.

In this model we assume that there is no rotational kinetic energy associated with shear distortion,but only with rotation due to bending. Therefore, the kinetic energy term due to rotation is modified to include only the angle of rotation due to bending by replacing $w^*_{x^*}$ with $\alpha$ and $w^*_{y^*}$ with $\beta$:
\begin{align}
	{KE}^*_{rot} = \tfrac{{\rho}^*}{2} \smallint\nolimits_{0}^{L_1^*}\smallint\nolimits_{0}^{L_2^*}[I_1^*(\alpha_{t^*})^2 + I_2^*(\beta_{t^*})^2]dx^*dy^*.
\end{align}
The Lagrangian, defined by kinetic energy-potential energy, is obtained as follows
\begin{align}
	L_{lagr}^* = KE_{trans}^* + KE_{rot}^* - PE_{beding} - PE_{shear}^*.
\end{align}
The virtual work due to the non-conservative transverse pressure is given by
\begin{align}
	\delta {W_{nc}^*} = \smallint\nolimits_{0}^{L_1^*}\smallint\nolimits_{0}^{L_2^*} p^*\delta w^* dx^*dy^*.
\end{align}
Introducing the following dimensionless parameters
\begin{align}\label{eq_dimensionless}
	\nonumber
	x &= \tfrac{x^*}{L_1^*}, y = \tfrac{y^*}{L_1^*}, L = \tfrac{L_2^*}{L_1^*}, h = \tfrac{h^*}{L_1^*}, t = \tfrac{t^*}{1}, G = \tfrac{G^*{L_1^*}^3}{E^*I^*},\\
	\nonumber
	\rho &= \tfrac{\rho^*{L_1^*}^5}{E^*I^*}, I = \tfrac{I^*}{{L_1^*}^3}, I_1 = \tfrac{I_1^*}{{L_1^*}^3}, I_2= \tfrac{I_2^*}{{L_1^*}^3},\\
	\rho_f &= \tfrac{\rho_f^*{L_1^*}^5}{E^*I^*}, U = \tfrac{U^*}{L_1^*}, M = \tfrac{M^*}{1}, p = \tfrac{p^*{L_1^*}^3}{E^*I^*}
\end{align} 
In terms of these dimensionless length scales, the dimensionless $L_{lagr}$ defined by $\tfrac{L_{lagr}^*}{E^*I^*}$ is given by
\begin{align}\label{eq_lagrangian}
	\nonumber
	&L_{lagr} = \tfrac{1}{2}\smallint\nolimits_{0}^{1}\smallint\nolimits_{0}^{L} [\rho h w_{t}^{2} + \rho I_1 \alpha_{t}^{2} + \rho I_2 \beta_{t}^{2} - \alpha_{x}^{2} - \beta_{y}^{2}\\
	& - (\alpha_y + \beta_x)^2 - k^{\prime}Gh(w_x - \alpha + w_y - \beta)^2] dxdy,
\end{align}
the dimensionless non-conservative work $W_{nc}$ defined by $\tfrac{W_{nc}^*}{E^*I^*}$ is given by
\begin{align}
	\delta W_{nc} = \smallint\nolimits_{0}^{1}\smallint\nolimits_{0}^{L} \tfrac{\rho_f U}{M}(w_t + U w_x)\delta w dxdy.
\end{align}	
Using the extended Hamilton's principle, by including the non-conservative forcing, the governing differential equation of motion is given by
 \begin{align}
	 	&\scalebox{0.9}{$\rho h w_{tt} = k^{\prime}Gh(w_{xx} - \alpha_x + w_{yy} - \beta_y) + \tfrac{\rho_f U}{M}(w_t + U w_x)$},\\
	 	&\rho I_1 \alpha_{tt} = \alpha_{xx} + k^{\prime}Gh(w_x - \alpha) + \alpha_{yy} + \beta_{xy},\\
	 	&\rho I_2 \beta_{tt} = \beta_{yy} + k^{\prime}Gh(w_y - \beta) + \beta_{xx} + \alpha_{yx},
	 \end{align}
and the boundary conditions are given by
 \begin{align}
	 	&\scalebox{0.9}{$k^{\prime}Gh(w_x-\alpha)\delta w \big|_{x=0}^{x=1} = 0, k^{\prime}Gh(w_y-\beta)\delta w \big|_{y=0}^{y=L} = 0$},\\
	 	&\alpha_x \delta\alpha \big|_{x=0}^{x=1} = 0, \quad \left(\alpha_y + \beta_x \right)(\delta \alpha)\big|_{y=0}^{y=L} = 0,\\
	 	&\left(\beta_x + \alpha_y\right) \delta \beta \big|_{x=0}^{x=1},\quad \beta_y \delta \beta \big|_{y=0}^{y=L} = 0.
	 \end{align}
Set $\epsilon=\frac{\rho h}{k^{\prime}Gh}$, $\mu_1=\rho I_1$, $\mu_2=\rho I_2$, $a=\rho h$, $\theta=\frac{\rho_f U}{Mk^{\prime}Gh}$, $\xi=\frac{\rho_f U^2}{Mk^{\prime}Gh}$, the equations of motion are turned into
\begin{align}\label{eq_2dwave}
	\epsilon w_{tt} &= w_{xx} - \alpha_x + w_{yy} - \beta_y + \theta w_t + \xi w_x,\\
	\mu_1 \alpha_{tt} &= \alpha_{xx} + \frac{a}{\epsilon}(w_x - \alpha) + \alpha_{yy} + \beta_{xy},\\
	\mu_2 \beta_{tt} &= \beta_{yy} + \frac{a}{\epsilon}(w_y - \beta) + \beta_{xx} + \alpha_{yx},
\end{align}
For the free end $x=0$ and $y=1$, the displacement $w$ and the rotation angle $\alpha$, $\beta$ maintain their natural boundary conditions
\begin{align}
	w_x(t,0,y) &= \alpha(t,0,y),\quad \alpha_x(t,0,y) = 0,\\
	\beta_x(t,0,y) &= -\alpha_y(t,0,y),
\end{align}
At the edge $y=0$, a clamping boundary is considered, and thus the lateral displacement $w(t,x,0)=0$ and the rotational displacement about the $y-$axis $\alpha(t,x,0)=0$. Between the supports and the elastic plate, cylindrical rollers aligned along the $x-$direction are installed. These rollers allow the plate's edge to rotate freely about the $x-$axis via rolling motion, thereby satisfying the zero-moment condition $\beta_y(t,x,0)=0$, i.e.,
\begin{align}
	w(t,x,0) = 0,\quad \alpha(t,x,0) = 0,\quad \beta_y(t,x,0) = 0.
\end{align}
At $x=1$, active control inputs are implemented, modifying the boundary conditions to 
\begin{align}
	w_x(t,1,y) &= \alpha(t,1,y) + \tfrac{1}{k^{\prime}Gh}U_1(t,y),\\
	\alpha_x(t,1,y) &= U_2(t,y),\\
	\beta_x(t,1,y) &= -\alpha_y(t,1,y) + U_3(t,y), \label{eq_2dw}
\end{align}
where $U_1(t,y)$, and $U_2(t,y)$, $U_3(t,y)$ are boundary control force and bending moments, respectively. Next, we conduct the control based on the 2D elastic plate model \eqref{eq_2dwave}--\eqref{eq_2dw}.
\subsection{Decompose the 2D problem into infinite many 1D problems through Fourier Series}
Specifically, the boundary conditions in $y$ motivate expanding the solution and control inputs in sine and cosine series, respectively:
\begin{align}
	\label{eq_w} w(t,x,y) &= \scalebox{0.9}{$\sum_{n=0}^{\infty}$} w_n(t,x) \sin(n\pi\tfrac{y}{L}),\\
	\label{eq_alpha} \alpha(t,x,y) &= \scalebox{0.9}{$\sum_{n=0}^{\infty}$} \alpha_n(t,x) \sin(n\pi\tfrac{y}{L}),\\
	\label{eq_beta} \beta(t,x,y) &= \scalebox{0.9}{$\sum_{n=0}^{\infty}$} \beta_n(t,x) \cos(n\pi\tfrac{y}{L}),\\
	\label{eq_U1} U_1(t,y) &= \scalebox{0.9}{$\sum_{n=0}^{\infty}$} U_{1,n}(t) \sin(n\pi\tfrac{y}{L}),\\
	\label{eq_U2} U_2(t,y) &= \scalebox{0.9}{$\sum_{n=0}^{\infty}$} U_{2,n}(t) \sin(n\pi\tfrac{y}{L}),\\
	\label{eq_U3} U_3(t,y) &= \scalebox{0.9}{$\sum_{n=0}^{\infty}$} U_{3,n}(t) \cos(n\pi\tfrac{y}{L}).
\end{align}
This expansion effectively decomposes the 2D problem into infinite many 1D problems, each corresponding to a different sine or cosine mode. For each mode $n$:
\begin{align}\label{eq_1dwave}
	&\scalebox{0.9}{$\epsilon w_{n,tt}$} = \scalebox{0.9}{$w_{n,xx} - \alpha_{n,x} - \tfrac{n^2\pi^2}{L^2} w_n + \tfrac{n\pi}{L} \beta_n + \theta w_{n,t} + \xi w_{n,x}$},\\
	&\scalebox{0.9}{$\mu_1 \alpha_{n,tt}$} = \scalebox{0.9}{$\alpha_{n,xx} + \frac{a}{\epsilon}(w_{n,x} - \alpha_n) - \tfrac{n^2\pi^2}{L^2} \alpha_n - \tfrac{n\pi}{L} \beta_{n,x}$},\\
	&\scalebox{0.9}{$\mu_2 \beta_{n,tt}$} = \scalebox{0.9}{$\beta_{n,xx} - \tfrac{n^2\pi^2}{L^2} \beta_n + \frac{a}{\epsilon}(\tfrac{n\pi}{L} w_n - \beta_n) + \tfrac{n\pi}{L} \alpha_{n,x}$},
\end{align}
with boundary conditions:
\begin{align}
	w_{n,x}(t,0) &= \alpha_n(t,0),\quad w_{n,x}(t,1) = U_{4,n}(t),\\
	\alpha_{n,x}(t,0) &= 0,\quad \alpha_{n,x}(t,1) = U_{2,n}(t),\\
	\label{eq_1dw} \beta_{n,x}(t,0) &= -\tfrac{n\pi}{L} \alpha_n(t,0),\quad \beta_{n,x}(t,1) = U_{5,n}(t),
\end{align}
where $U_{4,n}(t) = U_{1,n}(t)/(k^{\prime}Gh) + \alpha_n(t,1)$ and $U_{5,n}(t) = U_{3,n}(t) - \tfrac{n\pi}{L} \alpha_n(t,1)$.
Therefore, we obtain the 1D Timoshenko beam model represented by the PDE system \eqref{eq_1dwave}--\eqref{eq_1dw} for each mode $n$.
\begin{remark}\label{re:1}
\normalfont In practical modeling and control of flexible plates under specific operating conditions, modal truncation is well justified from an engineering standpoint. We assume a truncation number $N$ for mode $n$ in the following control design, where $N$ is a known and arbitrarily positive integer. From a theoretical perspective, a rigorously justified choice of $N$ for general cases can also be obtained when internal structural damping, which is not included in the present plant model to emphasize the treatment of in-domain instability, is incorporated into the model, under which the uncontrolled high-frequency modes are inherently stable.
\end{remark}
\section{Controller Design}
We present the control design for the large-scale wave PDEs \eqref{eq_1dwave}--\eqref{eq_1dw} with $n=0,1,\cdots,N$ where $N$ is an arbitrarily positive integer.

\subsection{Transformation to coupled transport PIDEs}
Following the classical Riemann transformation, the Timoshenko beam for each mode $n$ can be mapped into a first-order hyperbolic integro-differential system coupled with ODEs. Furthermore, in order to remove the diagonal coupling terms, we use a change of coordinates as presented in \cite{Adaptive}. The system becomes a $6\times6$ system of 1D hyperbolic PDEs coupled with three ODEs without diagonal coupling terms by using the following transformation:
\begin{align}
	\label{eq_riemann} p_n &= \exp(\sqrt{\epsilon}\bar{c}_1 x)(w_{n,x} + \sqrt{\epsilon}w_{n,t}),\\
	q_n &= \exp(\sqrt{\epsilon}\bar{c}_2 x)(w_{n,x} - \sqrt{\epsilon}w_{n,t}),\\
	r_n &= \alpha_{n,x} + \sqrt{\mu_1}\alpha_{n,t},\quad s_n = \alpha_{n,x} - \sqrt{\mu_1}\alpha_{n,t},\\
	u_n &= \beta_{n,x} + \sqrt{\mu_2}\beta_{n,t},\quad v_n = \beta_{n,x} - \sqrt{\mu_2}\beta_{n,t},\\
	\label{eq_rie} x_{1,n} &= w_n(t,0),\quad x_{2,n} = \alpha_n(t,0),\quad x_{3,n} = \beta_n(t,0).
\end{align}
where $\bar{c}_1 = \frac{\xi}{2\sqrt{\epsilon}} + \frac{\theta}{2\epsilon}, \bar{c}_2 = \frac{\xi}{2\sqrt{\epsilon}} - \frac{\theta}{2\epsilon}$. Define
\begin{align}
	\nonumber
	Z_n &= \begin{bmatrix}
		p_n&r_n&u_n
	\end{bmatrix}^\top,
	Y_n = \begin{bmatrix}
		q_n&s_n&v_n
	\end{bmatrix}^\top,\\
	X_n &= \begin{bmatrix}
		x_{1,n}&x_{2,n}&x_{3,n}
	\end{bmatrix}^\top,\\
	U_{\mathrm{in},n} &= \begin{bmatrix}
		\exp(\sqrt{\epsilon}\bar{c}_1)U_{p,n}&U_{r,n}&U_{u,n}
	\end{bmatrix}^\top,
\end{align}
where $U_{p,n}(t) = U_{4,n}(t) + \sqrt{\epsilon}w_{n,t}(t,1)$, $U_{r,n}(t) = U_{2,n}(t) + \sqrt{\mu_1}\alpha_{n,t}(t,1)$ and $U_{u,n}(t) = U_{5,n}(t) + \sqrt{\mu_2}\beta_{n,t}(t,1)$ are redefined control variables for this plant. Then, \eqref{eq_1dwave}--\eqref{eq_1dw} is equivalent to the PDE-ODE system in the matrix form:
\begin{align}
	\nonumber
	Z_{n,t} &= \Sigma Z_{n,x} + F_{11,n}(x)(Z_n+Y_n) + F_{12}(x)(Z_n-Y_n)\\
	\nonumber
	&\quad + F_{13,n}(x) X_n + \smallint\nolimits_{0}^{x} F_{14,n}(x,y)Z_n(t,y)dy\\
	&\quad + \smallint\nolimits_{0}^{x} F_{15,n}(x,y)Y_n(t,y) dy,\label{eq_matrix}\\
	Y_{n,t} &= -\Sigma Y_{n,x} + F_{21,n}(x)(Z_n+Y_n) + F_{22}(x)(Z_n-Y_n)\nonumber\\
	\nonumber
	&\quad + F_{23,n}(x) X_n + \smallint\nolimits_{0}^{x} F_{24,n}(x,y)Z_n(t,y) dy\\
	&\quad + \smallint\nolimits_{0}^{x} F_{25,n}(x,y)Y_n(t,y) dy,\label{eq_matrix2}\\
	\dot{X}_n &= AX_n + \Sigma Z_n(t,0),\label{eq_matrix3}
\end{align}
with boundary conditions
\begin{align}
	Z_n(t,1) = U_{\mathrm{in},n}, \quad Y_n(t,0) = CZ_n(t,0) + DX_n,\label{eq_matrix5}
\end{align}	
where the definition of $\Sigma$, $F_{11,n}(x)$, $F_{12}(x)$, $F_{13,n}(x)$, $F_{14,n}(x,y)$, $F_{15,n}(x,y)$, $F_{21,n}(x)$, $F_{22}(x)$, $F_{23,n}(x)$, $F_{24,n}(x,y)$, $F_{25,n}(x,y)$, $A$, $C$ and $D$ are shown in Appendix-\ref{adx:mat_coeff}. The system \eqref{eq_matrix}--\eqref{eq_matrix5} contains integral coupling terms and the states of ODEs appearing inside the domain of the PDEs. In what follows, without loss of generality, we assume $\tfrac{1}{\sqrt{\epsilon}} < \tfrac{1}{\sqrt{\mu_1}} < \tfrac{1}{\sqrt{\mu_2}}$, the other cases can be treated analogously by switching the order of the states $p_n$, $r_n$, $u_n$ in all subsequent steps.

\subsection{Backstepping Transformation and Target System}
We introduce the following backstepping transformation:
\begin{align}
	\label{eq_bactr} \nonumber
	\sigma_n &= Z_n - \smallint\nolimits_{0}^{x} K_n(x,y)Z_n(t,y) dy - \smallint\nolimits_{0}^{x} L_n(x,y)Y_n(t,y) dy\\
	&\quad - \Phi_n(x)X_n(t),\\
	\label{eq_bactr2} \psi_n &= Y_n.
\end{align}
The gain kernels are $3 \times 3$ matrices, i.e., $K_n(x,y) = \{k_{ij,n}(x,y)\}_{1\leq i,j \leq 3}$, $L_n(x,y) = \{l_{ij,n}(x,y)\}_{1\leq i,j \leq 3}$ and $\Phi_n(x) = \{\phi_{ij,n}(x)\}_{1\leq i,j \leq 3}$,
where the kernels $K_n(x,y)$ and $L_n(x,y)$ are both defined in the triangle domain $\Gamma:\left\{(x,y) \in \mathbb{R}^2| 0 \leq y \leq x \leq 1\right\}$, and where $\Phi_n(x)$ is defined in $[0,1]$. They satisfy following equations:
\begin{align}\label{eq_kernel}
	\nonumber
	&\Sigma L_{n,x}(x,y) - L_{n,y}(x,y)\Sigma = K_n(x,y)(F_{11,n}(y)-F_{12}(y))\\
	\nonumber
	&\quad + L_n(x,y)(F_{21,n}(y)-F_{22}(y))\\
    &\quad - \Omega(x)L_n(x,y) - F_{15,n}(x,y)\\
	&\quad + \smallint\nolimits_{y}^{x} \left[K_n(x,s)F_{15,n}(s,y) + L_n(x,s)F_{25,n}(s,y)\right] ds,\\
	\nonumber
	&\Sigma K_{n,x}(x,y) + K_{n,y}(x,y)\Sigma = K_n(x,y)(F_{11,n}(y)+F_{12}(y))\\
	\nonumber
	&\quad + L_n(x,y)(F_{21,n}(y)+F_{22}(y))\\
    &\quad - \Omega(x)K_n(x,y) - F_{14,n}(x,y)\\
	&\quad + \smallint\nolimits_{y}^{x} \left[K_n(x,s)F_{14,n}(s,y) + L_n(x,s)F_{24,n}(s,y)\right] ds,\\
	\nonumber
	&\Phi_{n,x}(x) = {\Sigma}^{-1}\Phi_n(x) A - \Sigma^{-1}F_{13,n}(x)\\
	\nonumber
	&\quad - \Sigma^{-1}\Omega_n(x)\Phi_n(x) + \Sigma^{-1}L_n(x,0)\Sigma D\\
	&\quad + \smallint\nolimits_{0}^{x} \Sigma^{-1} (K_n(x,y)F_{13,n}(y) + L_n(x,y)F_{23,n}(y)) dy
\end{align}
with boundary conditions for $K$ and $L$:
\begin{align}
	&\Sigma L_n(x,x) + L_n(x,x)\Sigma = -(F_{11,n}(x) - F_{12}(x)),\\
	\label{eq_bcomega} &\Sigma K_n(x,x) - K_n(x,x)\Sigma = -(F_{11,n}(x) + F_{12}(x)) + \Omega_n(x),\\
	&K_n(x,0)\Sigma - L_n(x,0)\Sigma C = \Phi_n(x) \Sigma,
\end{align}
and with initial conditions for $\Phi_n(x)$:
\begin{align}\label{eq_ker}
	\Phi_n(0) = 
	\begin{bmatrix}
            - \delta_1 \sqrt{\epsilon} & 1 & 0\\
			0 & -\delta_2\sqrt{\mu_1} & 0\\
			0 & -n\pi & -\delta_3\sqrt{\mu_2}
	\end{bmatrix}
\end{align} 
where
\begin{align}
	\Omega_n(x) = 
	\begin{bmatrix}
            0 & \omega_{12} & \omega_{13}\\
			0 & 0 & \omega_{23,n}\\
			0 & 0 & 0
	\end{bmatrix}
\end{align}
with $\omega_{12}(x) = (\frac{1}{\sqrt{\epsilon}} - \frac{1}{\sqrt{\mu_1}})k_{12,n}(x,x) + c_2(x)$, $\omega_{13}(x) = (\frac{1}{\sqrt{\epsilon}}-\frac{1}{\sqrt{\mu_2}})k_{13,n}(x,x)$, $\omega_{23,n}(x) = (\frac{1}{\sqrt{\mu_1}} - \frac{1}{\sqrt{\mu_2}})k_{23,n}(x,x) - \frac{n\pi}{2L\sqrt{\mu_1}}$. 
As will be seen in Appendix-\ref{adx:lyapunov}, the parameters $\delta_1, \delta_2$ and $\delta_3$ in \eqref{eq_ker} are positive design parameters which determine the decay rate of the closed-loop controlled Timoshenko beam.

Applying the above transformation \eqref{eq_bactr}, \eqref{eq_bactr2}, choosing the control law in the boundary \eqref{eq_matrix5} as
\begin{align}\label{eq_Uin}
	\nonumber
	U_{\mathrm{in},n} &= \smallint\nolimits_{0}^{1} K_n(1,y)Z_n(t,y) dy + \smallint\nolimits_{0}^{1} L_n(1,y)Y_n(t,y) dy\\
	&\quad + \Phi_n(1)X_n(t),
\end{align}
we convert \eqref{eq_matrix}--\eqref{eq_matrix5} to the target system:
\begin{align}\label{eq_target}
	\sigma_{n,t} &= \Sigma\sigma_{n,x} + \Omega_n(x)\sigma_n,\\
	\nonumber
	\psi_{n,t} &= -\Sigma\psi_{n,x} + (F_{21,n}(x) + F_{22}(x))\sigma_n\\
    \nonumber
    &\quad + (F_{21,n}(x) - F_{22}(x))\psi_n + \smallint\nolimits_{0}^{x} \Xi_{2,n}(x,y)\sigma_n(t,y) dy\\
    &\quad + \smallint\nolimits_{0}^{x} \Xi_{3,n}(x,y)\psi_n(t,y) dy + \Xi_{1,n}(x) X_n,\label{eq_target2}\\
	\dot X_n &= E_{1,n} X_n + \Sigma\sigma_n(t,0),
\end{align}
with boundary conditions
\begin{align}\label{eq_bc}
	\sigma_n(t,1) = 0,\quad \psi_n(t,0) = E_{2,n}X_n + C\sigma_n(t,0)
\end{align}
where
\begin{align}\label{eq_E}
	E_{1,n} = \Sigma\Phi_n(0) + A,\quad E_{2,n} = C\Phi_n(0) + D.
\end{align}
The well-posedness of the kernel equations is given in the following theorem.
\begin{thm}\label{thm:1}
	There exist unique bounded solutions $k_{ij,n}(x,y), l_{ij,n}(x,y)$ and $\phi_{ij,n}(x,y), i = 1,2,3; j = 1,2,3$ to the kernel equations \eqref{eq_kernel}--\eqref{eq_ker}; in particular, there exists a positive number $M$ such that for $i,j = 1,2,3$
	\begin{align}
		|k_{ij,n}(x,y)|, |l_{ij,n}(x,y)|,|\phi_{ij,n}(x,y)| \leq Me^{Mx}.
	\end{align}
\end{thm}
\begin{proof}
	\normalfont
	The proof of well-posedness of the kernel equations essentially follows the line in \cite{minimum}, but with the differences that our kernel equations incorporate additional integral terms and ODE to be solved. For the treatment of the ODEs, we draw inspiration from \cite{coupled}. The complete proof is presented in Appendix-\ref{adx:wellposed}. 
\end{proof}
Since the kernels in \eqref{eq_bactr} are bounded, the transformation is invertible from the theory of Volterra integral equation, and the inverse backstepping transformation is denoted as:
\begin{align}
	\label{eq_invtr} \nonumber
	Z_n &= \sigma_n + \smallint\nolimits_{0}^{x} \breve{K}_n(x,y)\sigma_n(t,y) dy + \smallint\nolimits_{0}^{x} \breve{L}_n(x,y)\psi_n(t,y) dy\\
	&\quad + \breve{\Phi}_n(x)X_n,\\
	\label{eq_invtr2} Y_n &= \psi_n
\end{align}
where the kernels $\breve{K}_n(x,y)$ $\breve{L}_n(x,y)$ and $\breve{\Phi}_n(x)$ are also $3 \times 3$ matrices,  defined in the triangle domain $\Gamma:\left\{(x,y) \in \mathbb{R}^2| 0 \leq y \leq x \leq 1\right\}$, and in $[0,1]$, respectively.
The functions $\Xi_{1,n}(x)$, $\Xi_{2,n}(x,y)$ and $\Xi_{3,n}(x,y)$ in \eqref{eq_target2} are given by
\begin{align}
	\nonumber
	\Xi_{1,n}(x) &= (F_{21,n}(x) + F_{22}(x))\breve{\Phi}_n(x) + F_{23,n}(x)\\
	&\quad + \smallint\nolimits_{0}^{x} F_{24,n}(x,y)\breve{\Phi}_n(y) dy,\\
	\nonumber
	\Xi_{2,n}(x,y) &= (F_{21,n}(x) + F_{22}(x))\breve{K}_n(x,y) + F_{24,n}(x,y)\\
	&\quad + \smallint\nolimits_{y}^{x} F_{24,n}(x,s)\breve{K}_n(s,y) ds,\\
	\nonumber
	\Xi_{3,n}(x,y) &= (F_{21,n}(x) + F_{22}(x))\breve{L}_n(x,y) + F_{25,n}(x,y)\\
	&\quad + \smallint\nolimits_{y}^{x} F_{24,n}(x,s)\breve{L}_n(s,y) ds,
\end{align}
where $\Xi_{2,n}(x,y)$ and $\Xi_{3,n}(x,y)$ are both defined in the triangle domain $\Gamma$, and $\Xi_{1,n}(x)$ is defined in $[0,1]$.

\subsection{Stabilizing control law and main result}
Expressing \eqref{eq_Uin} in terms of the Timoshenko beam variables for each mode $n$ and recalling that $U_{4,n}(t) = \frac{1}{k^{\prime}Gh}U_{1,n}(t) + \alpha_n(t,1)$ and $U_{5,n}(t) = U_{3,n}(t) - \frac{n\pi}{L}\alpha_n(t,1)$, we have
\vspace{-0.2cm}
\begin{align}
	\label{eq_U1n} \nonumber
	&U_{1,n}(t) = k^{\prime}Gh[(\smallint\nolimits_{0}^{1} (\mathcal{F}_{11}(\xi) w_n(t,\xi) + \mathcal{F}_{12}(\xi) w_{n,t}(t,\xi)\\
	\nonumber
	& - \mathcal{F}_{13}(\xi) \alpha_n(t,\xi) + \mathcal{F}_{14}(\xi) \alpha_{n,t}(t,\xi)\\
	\nonumber
	& - \mathcal{F}_{15}(\xi) \beta_n(t,\xi) + \mathcal{F}_{16}(\xi) \beta_{n,t}(t,\xi)) d\xi\\
	\nonumber
	& + \mathcal{D}_{11} w_n(t,1) - \mathcal{D}_{12} w_n(t,0) + \mathcal{D}_{13}\alpha_n(t,1)\\
	\nonumber
	& - \mathcal{D}_{14}\alpha_n(t,0) + \mathcal{D}_{15}\beta_n(t,1) - \mathcal{D}_{16}\beta_n(t,0))\\
	&\times\exp(-\sqrt{\epsilon}\bar{c}_1) - \sqrt{\epsilon}w_{n,t}(t,1) - \alpha_n(t,1)],\\
	\label{eq_U2n} \nonumber
	&U_{2,n}(t) = \smallint\nolimits_{0}^{1} (\mathcal{F}_{21}(\xi) w_n(t,\xi) + \mathcal{F}_{22}(\xi) w_{n,t}(t,\xi)\\
	\nonumber
	& - \mathcal{F}_{23}(\xi) \alpha_n(t,\xi) + \mathcal{F}_{24}(\xi) \alpha_{n,t}(t,\xi)\\
	\nonumber
	& - \mathcal{F}_{25}(\xi) \beta_n(t,\xi) + \mathcal{F}_{26}(\xi) \beta_{n,t}(t,\xi)) d\xi\\
	\nonumber
	& + \mathcal{D}_{21} w_n(t,1) - \mathcal{D}_{22} w_n(t,0)\\
	\nonumber
	& + \mathcal{D}_{23}\alpha_n(t,1) - \mathcal{D}_{24}\alpha_n(t,0) - \sqrt{\mu_1}\alpha_{n,t}(t,1)\\
	& + \mathcal{D}_{25}\beta_n(t,1) - \mathcal{D}_{26}\beta_n(t,0),\\
	\label{eq_U3n} \nonumber
	&U_{3,n}(t)  = \smallint\nolimits_{0}^{1} (-\mathcal{F}_{31}(\xi) w_n(t,\xi) + \mathcal{F}_{32}(\xi) w_{n,t}(t,\xi)\\
	\nonumber
	& - \mathcal{F}_{33}(\xi) \alpha_n(t,\xi) + \mathcal{F}_{34}(\xi) \alpha_{n,t}(t,\xi)\\
	\nonumber
	& - \mathcal{F}_{35}(\xi) \beta_n(t,\xi) + \mathcal{F}_{36}(\xi) \beta_{n,t}(t,\xi)) d\xi\\
	\nonumber
	& + \mathcal{D}_{31} w_n(t,1) - \mathcal{D}_{32} w_n(t,0)\\
	\nonumber
	& + \mathcal{D}_{33} \alpha_n(t,1) - \mathcal{D}_{34} \alpha_n(t,0)\\
	& + \mathcal{D}_{35}\beta_n(t,1) - \mathcal{D}_{36}\beta_n(t,0) - \sqrt{\mu_2}\beta_{n,t}(t,1),
\end{align}
where the expressions of $\mathcal{F}{ij}(\xi), \mathcal{D}_{ij}(i=1,2,3,\text{ }j=1,2,3,4,5,6)$ are given in Appendix-\ref{adx:input_coeff}. The main result for each mode is stated next:
\begin{thm}\label{thm:2}
	Consider system \eqref{eq_1dwave}--\eqref{eq_1dw} for $n=0,1,\cdots,N$, with initial conditions $w_{n,0}, \alpha_{n,0}, \beta_{n,0} \in H^{1}(0,1)$, $w_{n,0t}, \alpha_{n,0t}, \beta_{n,0t} \in L^2$, under the control law \eqref{eq_U1n}--\eqref{eq_U3n}, for $\delta_1, \delta_2, \delta_3$ satisfying $c_n = 2\min\left\{\delta_1, \delta_2, \delta_3\right\} - 1>0$, the exponential stability is obtained in the sense of
	\begin{align}\label{eq_thm2_n}
		\Omega_n(t) \leq C_ne^{-c_n t}\Omega_n(0),~n=0,1,\cdots,N
	\end{align}
	for some positive $C_n$, where $\Omega_n(t) = {\Vert w_n(t,\cdot)\Vert}_{H^1}^{2} + {\Vert \alpha_n(t,\cdot)\Vert}_{H^1}^{2} + {\Vert \beta_n(t,\cdot)\Vert}_{H^1}^{2} + {\Vert w_{n,t}(t,\cdot)\Vert}_{L^2}^{2} + {\Vert \alpha_{n,t}(t,\cdot)\Vert}_{L^2}^{2} + {\Vert \beta_{n,t}(t,\cdot)\Vert}_{L^2}^{2}$. 
\end{thm}
The proof of Theorem \ref{thm:2} is given in Appendix-\ref{adx:lyapunov}. 

\begin{corollary}\label{cor:1}
Under the assumption in Remark \ref{re:1}, 
 consider system \eqref{eq_2dwave}--\eqref{eq_2dw} with initial conditions $w_0$, $\alpha_0$, $\beta_0$ $\in H^1((0,1)^2)$ and $w_{0,t}$, $\alpha_{0,t}$, $\beta_{0,t}$ $\in L^2((0,1)^2)$ under the control law \begin{align}
		\label{eq_reU1}U_1(t,y) &= \scalebox{0.8}{$\sum_{n=0}^{N}$} U_{1,n}(t) \sin(n\pi\tfrac{y}{L}),\\
		\label{eq_reU2}U_2(t,y) &= \scalebox{0.8}{$\sum_{n=0}^{N}$} U_{2,n}(t) \sin(n\pi\tfrac{y}{L}),\\
		\label{eq_reU3}U_3(t,y) &= \scalebox{0.8}{$\sum_{n=0}^{N}$} U_{3,n}(t) \cos(n\pi\tfrac{y}{L}).
	\end{align} Let $\Omega_a(t)$ denote the total norm of the 2-D system:
	\begin{align}
		\nonumber
		&\Omega_a(t) = {\Vert w(t,\cdot,\cdot)\Vert}_{H^1}^{2} + {\Vert \alpha(t,\cdot,\cdot)\Vert}_{H^1}^{2} + {\Vert \beta(t,\cdot,\cdot)\Vert}_{H^1}^{2}\\
		& + {\Vert w_{t}(t,\cdot,\cdot)\Vert}_{L^2}^{2} + {\Vert \alpha_{t}(t,\cdot,\cdot)\Vert}_{L^2}^{2} + {\Vert \beta_{t}(t,\cdot,\cdot)\Vert}_{L^2}^{2},
	\end{align}
then there exists a constant $D_1>0$ and an arbitrary positive number $D_2$, which only depends on the arbitrarily positive design parameters $\delta_1, \delta_2$, and $\delta_3$, such that
	\begin{align}\label{eq_thm2_totalOmega}
		\Omega_a(t) \leq D_1 e^{-D_2 t}\Omega_a(0).
	\end{align}
\end{corollary}
\begin{proof}
	\normalfont
	Recalling the Fourier series \eqref{eq_w}--\eqref{eq_beta}, which can be truncated by modal number $N$ under the assumption of Remark \ref{re:1}, by the Parseval's identity, there exist constants $M_n>0$ such that $\Omega_a(t) \le  \scalebox{0.8}{$\sum_{n=0}^{N}$} M_n\Omega_n(t)$. Applying \eqref{eq_thm2_n}, we obtain
	\begin{align}
		\scalebox{0.95}{$\Omega_a(t) \leq D_1e^{-D_2t}\scalebox{0.8}{$\sum_{n=0}^{N}$} M_n \Omega_n(0) \leq D_1e^{-D_2t}\Omega_a(0)$}.
	\end{align}
	where $D_1\ge \max\{C_n\}$ and $D_2\le \min\{c_n\}$, for all $n \leq N$.
	
	We know from Theorem \ref{thm:2} that the constant $c_n$ only depends on the design parameters $\delta_1, \delta_2$ and $\delta_3$. 
	From Appendix-\ref{adx:lyapunov}, we know that the controller design parameters $\delta_1, \delta_2, \delta_3$ are chosen independently of the mode index $n$. Therefore, $D_2$ only depends on the arbitrarily positive design parameters $\delta_1, \delta_2$ and $\delta_3$, and it can be set as large as desired by adjusting $\delta_1, \delta_2$ and $\delta_3$.
	The proof is complete.
\end{proof}	

\section{Observer Design}
\subsection{Observer structure}
Next, we present the state observer design for distributed states of the $2-D$ elastic plates \eqref{eq_2dwave}--\eqref{eq_2dw} by using the boundary measurements $w_x(t,0,y)$, $w_t(t,0,y)$, $\alpha_x(t,0,y)$, $\alpha_t(t,0,y)$, $\beta_x(t,0,y)$, $\beta_t(t,0,y)$, $w(t,0,y)$, $\alpha(t,0,y)$ and $\beta(t,0,y)$. This implies that $Z_n(t,0)$ and $X_n(t)$ are accessible in the equivalent model \eqref{eq_matrix}--\eqref{eq_matrix5} via:
\begin{align}
	\scalebox{0.93}{$p_n(t,0)$} &= \scalebox{0.93}{$2\smallint\nolimits_{0}^{L} (w_x(t,0,y) + w_t(t,0,y))\sin(n\pi\tfrac{y}{L}) dy$},\label{eq:pn0}\\
	\scalebox{0.93}{$r_n(t,0)$} &= \scalebox{0.93}{$2\smallint\nolimits_{0}^{L} (\alpha_x(t,0,y) + \alpha_t(t,0,y))\sin(n\pi \tfrac{y}{L}) dy$},\\
	\scalebox{0.93}{$u_n(t,0)$} &= \scalebox{0.93}{$2\smallint\nolimits_{0}^{L} (\beta_x(t,0,y) + \beta_t(t,0,y))\cos(n\pi \tfrac{y}{L}) dy$},\\
	\scalebox{0.93}{$x_{1,n}(t)$} &= \scalebox{0.93}{$2\smallint\nolimits_{0}^{L} w(t,0,y)\sin(n\pi \tfrac{y}{L}) dy$},\\
	\scalebox{0.93}{$x_{2,n}(t)$} &= \scalebox{0.93}{$2\smallint\nolimits_{0}^{L} \alpha(t,0,y)\sin(n\pi \tfrac{y}{L}) dy$},\\
	\scalebox{0.93}{$x_{3,n}(t)$} &= \scalebox{0.93}{$2\smallint\nolimits_{0}^{L} \beta(t,0,y)\cos(n\pi \tfrac{y}{L}) dy$},\label{eq:x3}
\end{align}
by estimating \eqref{eq_riemann}--\eqref{eq_rie} at $x=0$.
In what follows, the state estimates are denoted by a hat. Relying on the measurements \eqref{eq:pn0}--\eqref{eq:x3}, recalling that $U_{p,n}(t) = U_{4,n}(t) + \sqrt{\epsilon}{\hat w}_{n,t}(t,1)$, $U_{r,n}(t) = U_{2,n}(t) + \sqrt{\mu}{\hat \alpha}_{n,t}(t,1)$ and $U_{u,n}(t) = U_{5,n}(t) + \sqrt{\mu}{\hat \beta}_{n,t}(t,1)$, and defining 
\begin{align*}
	\hat{Z}_n &= \begin{bmatrix}
		\hat{p}_n& \hat{r}_n& \hat{u}_n
	\end{bmatrix}^\top,
	\hat{Y}_n = \begin{bmatrix}
		\hat{q}_n& \hat{s}_n& \hat{v}_n
	\end{bmatrix}^\top,\\
	\hat{X}_n &= \begin{bmatrix}
		\hat{x}_{1,n}& \hat{x}_{2,n}& \hat{x}_{3,n}
	\end{bmatrix}^\top,
\end{align*}
we build an observer for the equivalent model \eqref{eq_matrix}--\eqref{eq_matrix5} as:
\begin{align}\label{eq_hat}
	\nonumber
	\hat{Z}_{n,t} &= \Sigma \hat{Z}_{n,x} + F_{11,n}(x)(\hat{Z}_n+\hat{Y}_n) + F_{12}(x)(\hat{Z}_n-\hat{Y}_n)\\
	\nonumber
	& + \smallint\nolimits_{0}^{x} F_n\left[F_{14,n}(x,y)\hat{Z}_n(t,y) + F_{15,n}(x,y)\hat{Y}_n(t,y)\right] dy\\
	& + P_n^{-}(x)(Z_n(t,0)-\hat{Z}_n(t,0)),\\
	\nonumber
	\hat{Y}_{n,t} &= -\Sigma \hat{Y}_{n,x} + F_{21,n}(x)(\hat{Z}_n+\hat{Y}_n) + F_{22}(x)(\hat{Z}_n-\hat{Y}_n)\\
	\nonumber
	& + \smallint\nolimits_{0}^{x} F_n\left[F_{24,n}(x,y)\hat{Z}_n(t,y) + F_{25,n}(x,y)\hat{Y}_n(t,y)\right] dy\\
	& + P_n^{+}(x)(Z_n(t,0)-\hat{Z}_n(t,0)),\\
	\dot{\hat{X}}_n &= AX_n + \Sigma Z_n(t,0) + L_{n,x} (X_n - \hat{X}_n),
\end{align}
with boundary conditions
\begin{align}
    \hat{Z}_n(t,1) &= U_{o,n}(t) + R_1 \hat{Y}_n(t,1),\\
	\label{eq_hatbc} \hat{Y}_n(t,0) &= CZ_n(t,0) + DX_n,
\end{align}
and where
\begin{align}
	U_{o,n}(t) &= 
	\begin{bmatrix}
			h_1 U_{4,n}(t)\\
			2U_{2,n}(t)\\
			2U_{5,n}(t)
	\end{bmatrix},
	R_1 = \begin{bmatrix}
			h_2 & 0 & 0\\
			0 & -1 & 0\\
			0 & 0 & -1
	\end{bmatrix},
\end{align}
where $h_1 = \frac{2\exp(-\sqrt{\epsilon}\bar{c}_1)}{2-\exp(-\sqrt{\epsilon}\bar{c}_1)}$, $h_2 = \frac{-\exp(-\sqrt{\epsilon}\bar{c}_1)}{2-\exp(-\sqrt{\epsilon}\bar{c}_1)}$ and $\Sigma$, $F_{11,n}(x)$, $F_{12}(x)$, $F_{21,n}(x)$, $F_{22}(x)$, $F_{14,n}(x)$, $F_{15,n}(x)$, $F_{24,n}(x)$, $F_{25,n}(x)$, $C$ and $D$ are defined in Appendix-\ref{adx:mat_coeff}, 
and where $P_n^{-}(\cdot), P_n^{+}(\cdot)$ and $L_{n,x}$ are output injection gain matrices yet to be designed.
Recalling the transformation \eqref{eq_riemann}--\eqref{eq_rie}, the estimate of original wave PDEs \eqref{eq_1dwave}--\eqref{eq_1dw} are obtained as 
\begin{align}\label{eq:w1}
	\left\{
	\begin{array}{l}
		\hat{\alpha}_n(t,x) = \left[\smallint\nolimits_{0}^{x} \tfrac{\hat{r}_n(t,y) + \hat{s}_n(t,y)}{2} dy + \hat{x}_2(t)\right],\\
		\hat{\beta}_n(t,x) = \left[\smallint\nolimits_{0}^{x} \tfrac{\hat{u}_n(t,y) + \hat{v}_n(t,y)}{2} dy + \hat{x}_3(t)\right],\\
		\hat{w}_n(t,x) = \smallint\nolimits_{0}^{x} \tfrac{k_2(y)\hat{p}_n(t,y) + k_1(y)\hat{q}_n(t,y)}{2k_1(y)\cdot k_2(y)} dy + \hat{x}_1(t),\\
		{\hat \alpha}_{n,t}(t,x) = \tfrac{\hat{r}_n(t,x) - \hat{s}_n(t,x)}{2\sqrt{\mu_1}},
		{\hat \beta}_{n,t}(t,x) = \tfrac{\hat{u}_n(t,x) - \hat{v}_n(t,x)}{2\sqrt{\mu_2}},\\
		{\hat w}_{n,t}(t,x) = \tfrac{k_2(x)\hat{p}_n(t,x)-k_1(x)\hat{q}_n(t,x)}{2\sqrt{\epsilon}(k_1(x)\cdot k_2(x))},\\
		\hat{\alpha}_{n,x}(t,x) = \tfrac{\hat{r}_n(t,x) + \hat{s}_n(t,x)}{2}, \hat{\beta}_{n,x}(t,x) = \tfrac{\hat{u}_n(t,x) + \hat{v}_n(t,x)}{2},\\
		\hat{w}_{n,x}(t,x) = \tfrac{k_2(x)\hat{p}_n(t,x) + k_1(x)\hat{q}_n(t,x)}{2k_1(x)\cdot k_2(x)},
	\end{array}
	\right.
\end{align}
where functions $k_1(x) = \exp(\sqrt{\epsilon}\bar{c}_1x)$, $k_2(x) = \exp(\sqrt{\epsilon}\bar{c}_2x)$, and where $\hat{p}_n, \hat{q}_n, \hat{r}_n, \hat{s}_n, \hat{u}_n$, $\hat{v}_n$ are computed from \eqref{eq_hat}--\eqref{eq_hatbc}. 
\subsection{Observer gains and error systems}
Defining the observer errors
\begin{align}
	{\tilde Z}_n=Z_n-{\hat Z}_n, \tilde{Y}_n=Y_n-\hat{Y}_n, \tilde{X}_n = X_n - \hat{X}_n.
\end{align}
Subtracting the \eqref{eq_hat}--\eqref{eq_hatbc} from \eqref{eq_matrix}--\eqref{eq_matrix5}, we get the observer error system:
\begin{align}
	\label{eq_tildez} \nonumber
	\tilde{Z}_{n,t} &= \Sigma \tilde{Z}_{n,x} + F_{11,n}(x)(\tilde{Z}_n+\tilde{Y}_n) + F_{12}(x)(\tilde{Z}_n-\tilde{Y}_n)\\
	\nonumber
	&\quad + \smallint\nolimits_{0}^{x} F_n\left[F_{14,n}(x,y)\tilde{Z}_n(t,y) + F_{15,n}(x,y)\tilde{Y}_n(t,y)\right] dy\\
	&\quad - P_n^{-}(x)\tilde{Z}_n(t,0),\\
	\label{eq_tildey} \nonumber
	\tilde{Y}_{n,t} &= -\Sigma \tilde{Y}_{n,x} + F_{21,n}(x)(\tilde{Z}_n+\tilde{Y}_n) + F_{22}(x)(\tilde{Z}_n-\tilde{Y}_n)\\
	\nonumber
	&\quad + \smallint\nolimits_{0}^{x} F_n\left[F_{24,n}(x,y)\tilde{Z}_n(t,y) + F_{25,n}(x,y)\tilde{Y}_n(t,y)\right] dy\\
	&\quad - P_n^{+}(x)\tilde{Z}_n(t,0),\\
	\dot{\tilde{X}}_n &= - L_{n,x} \tilde{X}_n,
\end{align}
with boundary conditions
\begin{align}\label{eq_tildebc}
	\tilde Z_n(t,1) &= R_1 \tilde{Y}_n(t,1),\quad \tilde{Y}_n(t,0) = 0,
\end{align}
where observer gain $L_{n,x}$ is
\begin{align}\label{eq_lx}
	L_{n,x} = 
	\begin{bmatrix}
            \frac{L_1}{\sqrt{\epsilon}} & 0 & 0\\
			0 & \frac{L_2}{\sqrt{\mu_1}} & 0\\
			0 & 0 & \frac{L_3}{\sqrt{\mu_2}}	
	\end{bmatrix}.
\end{align}
The parameters $L_1, L_2$, and $L_3$ are positive design parameters that determine the decay rate of the state $\tilde{X}_n$.
To determine the other two observer gains $P_n^-$ and $P_n^+$, we introduce the following Volterra transformation:
\begin{align}
	\label{eq_tibacz} \tilde{Z}_n(t,x) &= \tilde{\sigma}_n(t,x) + \smallint\nolimits_{0}^{x} N_n(x,y)\tilde{\sigma}_n(t,y) dy,\\
	\label{eq_tibacy} \tilde{Y}_n(t,x) &= \tilde{\psi}_n(t,x) + \smallint\nolimits_{0}^{x} M_n(x,y)\tilde{\sigma}_n(t,y) dy,
\end{align}
where the kernels $N_n$ and $M_n$ defined on $\Gamma=\{(x,y)\in\mathbb{R}^2|0\leq y \leq x \leq 1\}$ satisfy following kernel equations:
\begin{align}\label{eq_obkernel}
	\nonumber
	&\Sigma M_x(x,y) - M_y(x,y) \Sigma = (F_{21}(x) - F_{22}(x))M(x,y)\\
	\nonumber
	& + (F_{21}(x) + F_{22}(x))N(x,y) - M(x,y)\Omega(y) + F_{24}(x,y)\\
	& + \smallint\nolimits_{y}^{x} (F_{24}(x,s)N(s,y) + F_{25}(x,s)M(s,y)) ds,\\
	\nonumber
	&\Sigma N_x(x,y) + N_y(x,y) \Sigma = -(F_{11}(x) + F_{12}(x))N(x,y)\\
	\nonumber
	& - (F_{11}(x) - F_{12}(x))M(x,y) + N(x,y)\Omega(y) - F_{14}(x,y)\\
	& - \smallint\nolimits_{y}^{x} (F_{14}(x,s)N(s,y) + F_{15}(x,s)M(s,y)) ds,
\end{align}
with boundary conditions:
\begin{align}
	&\scalebox{0.9}{$\Sigma N_n(x,x) - N_n(x,x)\Sigma$} = \scalebox{0.9}{$-(F_{11,n}(x) + F_{12}(x)) + \Omega_n(x)$},\\
	&\scalebox{0.9}{$\Sigma M_n(x,x) + M_n(x,x)\Sigma$} = \scalebox{0.9}{$F_{21,n}(x) + F_{22}(x)$},\label{eq_obkernelbc}
\end{align}
Applying the backstepping transformation \eqref{eq_tibacz},\eqref{eq_tibacy}, and choosing the observer gains $P_n^{+}$, $P_n^{-}$ as
\begin{align}\label{eq_P}
	P_n^{+}(x) = M_n(x,0)\Sigma,\quad P_n^{-}(x) = N_n(x,0)\Sigma,
\end{align}
we map the observer error system \eqref{eq_tildez}--\eqref{eq_tildebc} into the following target system:
\begin{align}
	\label{eq_tisigma} \nonumber
	\tilde{\sigma}_{n,t} &= \Sigma\tilde{\sigma}_{n,x} + (F_{11,n}(x) - F_{12}(x))\tilde{\psi}_n + \Omega_n(x)\tilde{\sigma}_n(t,x)\\
	&\quad + \smallint\nolimits_{0}^{x} D_n^{-}(x,y)\tilde{\psi}_n(t,y) dy,\\
	\label{eq_tipsi} \nonumber
	\tilde{\psi}_{n,t} &= -\Sigma\tilde{\psi}_{n,x} + (F_{21,n}(x) - F_{22}(x))\tilde{\psi}_n\\
	&\quad + \smallint\nolimits_{0}^{x} D_n^{+}(x,y)\tilde{\psi}_n(t,y) dy,\\
	\label{eq_x} \dot{\tilde{X}}_n &= -L_{n,x}\tilde{X}_n,
\end{align}
with boundary conditions
\begin{align}\label{eq_tibc}
	\tilde{\sigma}_n(t,1) &= R_1 \tilde{\psi}_n(t,1),\quad \tilde{\psi}_n(t,0) = 0,
\end{align}
where
\begin{align}
	\Omega_n(x) = 
	\begin{bmatrix}
            0 & \omega_{12} & \omega_{13}\\
			0 & 0 & \omega_{23,n}\\
			0 & 0 & 0
	\end{bmatrix},
\end{align}
and where
$D_n^{+}, D_n^{-}$ are given by
\begin{align}
	\nonumber
	D_n^{+}(x,y) &= -M_n(x,y)(F_{11,n}(y)-F_{12}(y))\\
	&\quad - \smallint\nolimits_{y}^{x} M_n(x,s)D_n^{-}(s,y) ds + F_{25,n}(x,y),\\
	\label{eq_obker} \nonumber
	D_n^{-}(x,y) &= -N_n(x,y)(F_{11,n}(y)-F_{12}(y))\\
	&\quad - \smallint\nolimits_{y}^{x} N_n(x,s)D_n^{-}(s,y) ds + F_{15,n}(x,y).
\end{align}
\begin{lemma}\label{lem:2}
	There exists a unique bounded solution $m_{ij,n}(x,y), n_{ij,n}(x,y)$, $i = 1,2,3; j = 1,2,3$ to the kernel equations \eqref{eq_obkernel}--\eqref{eq_obkernelbc}. In particular, there exists some positive number $\mathcal{\phi}$ and $D$ such that for $i,j = 1,2,3$
	\begin{align}
		|m_{ij,n}(x,y)|, |n_{ij,n}(x,y)| \leq \varphi e^{Dx}.
	\end{align}
\end{lemma}
The proof of Lemma \ref{lem:2} is shown in Appendix-\ref{adx:observer}.

\subsection{Result of the observer}
The following lemma assesses the convergence of the target system to zero.
\begin{thm}\label{thm:3}
	Consider system \eqref{eq_tisigma}--\eqref{eq_x} for $n=0,1,\cdots,N$, with initial conditions for $\tilde{X}_n(0)$ and the output injection kernels given by $P_n^{-}(x), P_n^{+}(x)$ and $L_{n,x}$, where $P_n^{-}(x)$ and $P_n^{+}(x)$ are obtained from \eqref{eq_P} for each mode $n$. Choosing the values of the output injection gains $L_1, L_2, L_3$ to be positive, the observer error system is exponentially stable in the sense of:
	\begin{align}\label{eq:thm3}\left\{ 
		\begin{array}{l}
			\tilde{X}_n(t) \leq e^{-c_n^{\prime}t} \tilde{X}_n(0),t\geq 0\\
			\tilde{\sigma}_n \equiv \tilde{\psi}_n \equiv 0,\quad t \geq \tfrac{2}{\lambda_1}.
		\end{array}
		\right.
	\end{align}
	for $n=0,1,\cdots,N$, where $c_n^{\prime} = \min\{L_1, L_2, L_3\}$, and where $\tilde{\sigma}_n$ and $\tilde{\psi}_n$ are bounded in $t\in[0,\frac{2}{\lambda_1}]$.
\end{thm}
\begin{proof}
	\normalfont
	Noting \eqref{eq_tisigma},\eqref{eq_tipsi}, we find that the system consists in a cascade of the $\tilde{\psi}_n$-system into the $\tilde{\sigma}_n$-system. Therefore, by using the method of characteristics, we can easily find that $\tilde{\psi}_n$ is identically zero for $t \geq \frac{1}{\lambda_1}$. When $t \geq \frac{1}{\lambda_1}$, the $\tilde{\sigma}_n$-system  becomes:
	\begin{align}
		\scalebox{0.97}{$\tilde{\sigma}_{n,t}(t,x) - \Sigma \tilde{\sigma}_{n,x}(t,x) = \Omega_n(x)\tilde{\sigma}_n(t,x),  \tilde{\sigma}_n(t,1) = 0$}.
	\end{align}
	Noting the particular structure of $\Omega_n(x)$, the $\tilde{\sigma}_n$-system is in fact a cascade of its fast states into its slow states. So one can obtain that $\tilde{\sigma}_n$ eventually identically vanishes for
	$t \geq \tfrac{1}{\lambda_1} + \tfrac{1}{\lambda_1} = \tfrac{2}{\lambda_1}$.
	This concludes the proof.
\end{proof}
Defining the 2-D variables
\begin{align}\label{eq:waveobe}
	\tilde{w} = w - \hat{w}, \tilde{\alpha} = \alpha - \hat{\alpha}, \tilde{\beta} = \beta - \hat{\beta},
\end{align}
the results about the observer error for the original 2-D PDE are given below.  
\begin{corollary}\label{cor:2}
	Under the assumption of Remark \ref{re:1}, considering the observer \eqref{eq_hat}--\eqref{eq_hatbc} together with \eqref{eq:w1} for each Fourier mode $n$, constructing the 2D state estimates $\hat{w}(t,x,y) =\scalebox{0.75}{$\sum_{n=0}^{N}$}\hat{w}_n(t,x)\sin(\tfrac{n\pi y}{L})$, $\hat{\alpha}(t,x,y) = \scalebox{0.75}{$\sum_{n=0}^{N}$}\hat{\alpha}_n(t,x)\sin(\tfrac{n\pi y}{L})$ and $\hat{\beta}(t,x,y) = \scalebox{0.75}{$\sum_{n=0}^{N}$}\hat{\beta}_n(t,x)\cos(\tfrac{n\pi y}{L})$, for initial conditions of the resulting observer errors $\tilde{w}_0$, $\tilde{\alpha}_0$, $\tilde{\beta}_0$ $\in H^1((0,1)^2)$ and $\tilde{w}_{0,t}$, $\tilde{\alpha}_{0,t}$, $\tilde{\beta}_{0,t}$ $\in L^2((0,1)^2)$, the estimates exponentially track the states in the plant \eqref{eq_2dwave}--\eqref{eq_2dw} in the sense that
	there exist a constant $D_3 > 0$ and an arbitrary positive number $D_4$ such that   \begin{align}\label{eq_thm3_totalOmega}
		\Omega_f(t) \leq D_3 e^{-D_4 t}\Omega_f(0)
	\end{align}
	where 
	\begin{align}
		\nonumber
		&\Omega_f(t) = {\Vert \tilde{w}(t,\cdot,\cdot)\Vert}_{H^1}^{2} + {\Vert \tilde{\alpha}(t,\cdot,\cdot)\Vert}_{H^1}^{2} + {\Vert \tilde{\beta}(t,\cdot,\cdot)\Vert}_{H^1}^{2}\\
		& + {\Vert \tilde{w}_{t}(t,\cdot,\cdot)\Vert}_{L^2}^{2} + {\Vert \tilde{\alpha}_{t}(t,\cdot,\cdot)\Vert}_{L^2}^{2} + {\Vert \tilde{\beta}_{t}(t,\cdot,\cdot)\Vert}_{L^2}^{2}.
	\end{align}
\end{corollary}
\begin{proof}
	According to Theorem \ref{thm:3},  \eqref{eq_tibacz},\eqref{eq_tibacy}, \eqref{eq:w1}, and \eqref{convert}, we have
	\begin{align}\label{eq_thm3_n}
		\Omega_{nf}(t) \leq H_ne^{-c^{\prime}_n t}\Omega_{nf}(0),
	\end{align}
	for some positive $H_n$, where $\Omega_{nf}(t) = {\Vert \tilde{w}_n(t,\cdot)\Vert}_{H^1}^{2} + {\Vert \tilde{\alpha}_n(t,\cdot)\Vert}_{H^1}^{2} + {\Vert \tilde{\beta}_n(t,\cdot)\Vert}_{H^1}^{2} + {\Vert \tilde{w}_{n,t}(t,\cdot)\Vert}_{L^2}^{2} + {\Vert \tilde{\alpha}_{n,t}(t,\cdot)\Vert}_{L^2}^{2} + {\Vert \tilde{\beta}_{n,t}(t,\cdot)\Vert}_{L^2}^{2}$. It is obtained from \eqref{eq_w}--\eqref{eq_beta} and Remark \ref{re:1} that
	$\tilde{w} = \scalebox{0.8}{$\sum_{n=0}^{N}$} \tilde{w}_n(t,x)\sin(n\pi\tfrac{y}{L})$, $\tilde{\alpha} = \scalebox{0.8}{$\sum_{n=0}^{N}$} \tilde{\alpha}_n(t,x)\sin(n\pi\tfrac{y}{L})$, $\tilde{\beta} = \scalebox{0.8}{$\sum_{n=0}^{N}$} \tilde{\beta}_n(t,x)\cos(n\pi\tfrac{y}{L})$, with similar expansions for their time derivatives $\tilde{w}_t$, $\tilde{\alpha}_t$, $\tilde{\beta}_t$. Using Parseval's identity, there exist constants $Q_n>0$ such that $\Omega_f(t) \le  \scalebox{0.8}{$\sum_{n=0}^{N}$} Q_n\Omega_{nf}(t)$. Applying \eqref{eq_thm3_n}, we obtain
	\begin{align}
		\scalebox{0.93}{$\Omega_f(t) \leq D_3e^{-D_4t}\scalebox{0.8}{$\sum_{n=0}^{N}$} Q_n \Omega_{nf}(0) \leq D_3e^{-D_4t}\Omega_f(0)$},
	\end{align}
	where $D_3 \ge \max\{H_n\}$ and $D_4\le \min\{c^{\prime}_n\}$, for all $n \leq N$. The proof is complete.
	We know from Theorem \ref{thm:3} that the constant $c^{\prime}_n$ only depends on the observer gain parameters $L_1, L_2$ and $L_3$ and they can be chosen independently of the mode index $n$. Therefore, $D_4$ only depends on the arbitrarily positive design parameters $L_1, L_2$, $L_3$, and it can be set as large as desired by adjusting $L_1, L_2$, $L_3$.
\end{proof}

\section{Output-Feedback Control}
Combining the full state feedback law with the observer estimates, we obtain an output-feedback law for each mode $n$:
\begin{align}
	\label{eq_U1nob} \nonumber
	&\hat U_{1,n}(t) = k^{\prime}Gh[(\smallint\nolimits_{0}^{1} (\mathcal{F}_{11}(\xi) \hat{w}_n(t,\xi) + \mathcal{F}_{12}(\xi) \hat{w}_{n,t}(t,\xi)\\
	\nonumber
	& - \mathcal{F}_{13}(\xi) \hat{\alpha}_n(t,\xi) + \mathcal{F}_{14}(\xi)\hat{\alpha}_{n,t}(t,\xi)\\
	\nonumber
	& - \mathcal{F}_{15}(\xi) \hat{\beta}_n(t,\xi) + \mathcal{F}_{16}(\xi) \hat{\beta}_{n,t}(t,\xi)) d\xi\\
	\nonumber
	& + \mathcal{D}_{11} \hat{w}_n(t,1) - \mathcal{D}_{12} {w}_n(t,0) + \mathcal{D}_{13} \hat{\alpha}_n(t,1)\\
	\nonumber
	& - \mathcal{D}_{14} {\alpha}_n(t,0) + \mathcal{D}_{15} \hat{\beta}_n(t,1) - \mathcal{D}_{16} {\beta}_n(t,0))\\
	&\times\exp(-\sqrt{\epsilon}\bar{c}_1) - \sqrt{\epsilon}\hat{w}_{n,t}(t,1) - \hat{\alpha}_n(t,1)],\\
	\label{eq_U2nob} \nonumber
	&\hat U_{2,n}(t) = \smallint\nolimits_{0}^{1} (\mathcal{F}_{21}(\xi) \hat{w}_n(t,\xi) + \mathcal{F}_{22}(\xi) \hat{w}_{n,t}(t,\xi)\\
	\nonumber
	& - \mathcal{F}_{23}(\xi) \hat{\alpha}_n(t,\xi) + \mathcal{F}_{24}(\xi) \hat{\alpha}_{n,t}(t,\xi)\\
	\nonumber
	& - \mathcal{F}_{25}(\xi) \hat{\beta}_n(t,\xi) + \mathcal{F}_{26}(\xi) \hat{\beta}_{n,t}(t,\xi)) d\xi\\
	\nonumber
	& + \mathcal{D}_{21} \hat{w}_n(t,1) - \mathcal{D}_{22} {w}_n(t,0)\\
	\nonumber
	& + \mathcal{D}_{23}\hat{\alpha}_n(t,1) - \mathcal{D}_{24} {\alpha}_n(t,0) - \sqrt{\mu_1}\hat{\alpha}_{n,t}(t,1)\\
	& + \mathcal{D}_{25} \hat{\beta}_n(t,1) - \mathcal{D}_{26} {\beta}_n(t,0),\\
	\label{eq_U3nob} \nonumber
	&\hat U_{3,n}(t) = \smallint\nolimits_{0}^{1} (-\mathcal{F}_{31}(\xi) \hat{w}_n(t,\xi) + \mathcal{F}_{32}(\xi) \hat{w}_{n,t}(t,\xi)\\
	\nonumber
	& - \mathcal{F}_{33}(\xi) \hat{\alpha}_n(t,\xi) + \mathcal{F}_{34}(\xi) \hat{\alpha}_{n,t}(t,\xi)\\
	\nonumber
	& - \mathcal{F}_{35}(\xi) \hat{\beta}_n(t,\xi) + \mathcal{F}_{36}(\xi) \hat{\beta}_{n,t}(t,\xi)) d\xi\\
	\nonumber
	& + \mathcal{D}_{31} \hat{w}_n(t,1) - \mathcal{D}_{32} {w}_n(t,0)\\
	\nonumber
	&+ \mathcal{D}_{33} \hat{\alpha}_n(t,1) - \mathcal{D}_{34} {\alpha}_n(t,0)\\
	& + \mathcal{D}_{35} \hat{\beta}_n(t,1) - \mathcal{D}_{36} {\beta}_n(t,0) - \sqrt{\mu_2}\hat{\beta}_{n,t}(t,1),
\end{align}
where the estimated states ${\hat{w}}_n, {\hat{\alpha}}_n, {\hat{\beta}}_n, {\hat{w}}_{n,t}, {\hat{\alpha}}_{n,t}$ and ${\hat{\beta}}_{n,t}$ are obtained from \eqref{eq:w1}.
Based on \eqref{eq_U1nob}--\eqref{eq_U3nob} for each mode $n$ and under the assumption of Remark \ref{re:1}, the boundary controller for the original 2-D PDE \eqref{eq_2dwave}--\eqref{eq_2dw}  is
\begin{align}
	\hat U_1(t,y) &= \scalebox{0.9}{$\sum_{n=1}^{N}$} \hat U_{1,n}(t)\sin(n\pi \tfrac{y}{L}),\label{eq:hatU1}\\
	\hat U_2(t,y) &= \scalebox{0.9}{$\sum_{n=1}^{N}$}\hat  U_{2,n}(t)\sin(n\pi \tfrac{y}{L}),\\
	\hat U_3(t,y) &= \scalebox{0.9}{$\sum_{n=0}^{N}$}\hat  U_{3,n}(t)\cos(n\pi \tfrac{y}{L}).\label{eq:hatU3}
\end{align}
The result of the obtained output-feedback closed-loop system is provided as follows.
\begin{thm}
    Consider system \eqref{eq_1dwave}--\eqref{eq_1dw} for $n=0,1,\cdots, N$, with initial conditions $w_{n,0}$, $\alpha_{n,0}$, $\beta_{n,0} \in H^{1}(0,1)$, $w_{n,0t}$, $\alpha_{n,0t}$, $\beta_{n,0t} \in L^2$, under the control law \eqref{eq_U1nob}--\eqref{eq_U3nob} with the observer designed as \eqref{eq_hat}--\eqref{eq:w1}. For $\delta_1, \delta_2, \delta_3$ satisfying $ 2\min\left\{\delta_1, \delta_2, \delta_3\right\} - 1>0$ and output injection gains $L_1$, $L_2$, $L_3$ to be positive, the exponential stability is obtained in the sense of
	\begin{align}\label{eq_thm4_n}
		\Omega_{nd}(t) \leq \check C_ne^{-\check c_n t}\Omega_{nd}(0),~n=0,1,\cdots,N
	\end{align}
	for some positive $\check C_n$, where $\check c_n=\min\{2\min\left\{\delta_1, \delta_2, \delta_3\right\} - 1,L_1,L_2,L_3\}$ and where $\Omega_{nd}(t) = {\Vert w_n(t,\cdot)\Vert}_{H^1}^{2} + {\Vert \alpha_n(t,\cdot)\Vert}_{H^1}^{2} + {\Vert \beta_n(t,\cdot)\Vert}_{H^1}^{2} + {\Vert w_{n,t}(t,\cdot)\Vert}_{L^2}^{2} + {\Vert \alpha_{n,t}(t,\cdot)\Vert}_{L^2}^{2} + {\Vert \beta_{n,t}(t,\cdot)\Vert}_{L^2}^{2} + {\Vert \hat{w}_n(t,\cdot)\Vert}_{H^1}^{2} + {\Vert \hat{\alpha}_n(t,\cdot)\Vert}_{H^1}^{2} + {\Vert \hat{\beta}_n(t,\cdot)\Vert}_{H^1}^{2} + {\Vert \hat{w}_{n,t}(t,\cdot)\Vert}_{L^2}^{2} + {\Vert \hat{\alpha}_{n,t}(t,\cdot)\Vert}_{L^2}^{2} + {\Vert \hat{\beta}_{n,t}(t,\cdot)\Vert}_{L^2}^{2}$. 
\end{thm}
\begin{proof}
    Comparing the output-feedback control inputs \eqref{eq_U1nob}--\eqref{eq_U3nob} also exponentially converge to the state-feedback control laws \eqref{eq_U1n}--\eqref{eq_U3n}, recalling Theorem \ref{thm:3}, we have $U_{i,n}-\hat U_{i,n}$, are exponentially convergent to zero. Applying Theorem \ref{thm:2},  as well as $\hat{w}_n = w_n - \tilde{w}_n$, $\hat{\alpha}_n=\alpha_n-\tilde{\alpha}_n$, $\hat{\beta}_n=\beta_n-\tilde{\beta}_n$, the proof is complete.
\end{proof}
\begin{corollary}\label{cor:3}
	Under the assumption in Remark \ref{re:1}, consider the closed-loop system composed of the original plant \eqref{eq_2dwave}--\eqref{eq_2dw}, the observer \eqref{eq_hat}--\eqref{eq:w1}, and the control law given by \eqref{eq_U1nob}--\eqref{eq:hatU3}, the exponential stability is achieved in the sense that there exist constants $D_5>0$ and $D_6>0$ such that
	\begin{align}\label{eq_totalOmega}
		\Omega_d(t) \leq D_5 e^{-D_6 t}\Omega_d(0)
	\end{align}
	where $D_6$, which only depends on the arbitrarily positive design parameters  $\delta_1, \delta_2, \delta_3$, $L_1, L_2, L_3$, can be arbitrarily assigned by users, and where
	\begin{align}
		\nonumber
		&\Omega_d(t) = {\Vert w(t,\cdot,\cdot)\Vert}_{H^1}^{2} + {\Vert \alpha(t,\cdot,\cdot)\Vert}_{H^1}^{2} + {\Vert \beta(t,\cdot,\cdot)\Vert}_{H^1}^{2}\\
		\nonumber
		& + {\Vert w_{t}(t,\cdot,\cdot)\Vert}_{L^2}^{2} + {\Vert \alpha_{t}(t,\cdot,\cdot)\Vert}_{L^2}^{2} + {\Vert \beta_{t}(t,\cdot,\cdot)\Vert}_{L^2}^{2}\\
		\nonumber
		& + {\Vert \hat{w}(t,\cdot,\cdot)\Vert}_{H^1}^{2} + {\Vert \hat\alpha(t,\cdot,\cdot)\Vert}_{H^1}^{2} + {\Vert \hat\beta(t,\cdot,\cdot)\Vert}_{H^1}^{2}\\
		& + {\Vert \hat{w}_{t}(t,\cdot,\cdot)\Vert}_{L^2}^{2} + {\Vert \hat{\alpha}_{t}(t,\cdot,\cdot)\Vert}_{L^2}^{2} + {\Vert \hat{\beta}_{t}(t,\cdot,\cdot)\Vert}_{L^2}^{2}.
	\end{align}
\end{corollary}
\begin{proof}
	\normalfont
	The proof is similar to those of Corollary \ref{cor:1} and \ref{cor:2}
\end{proof}

\section{Numerical Simulation}
The 2-D elastic model considered in the simulation is \eqref{eq_2dwave}--\eqref{eq_2dw} with the physical parameters given in Table \ref{tab:2}.  The relationship between physical and dimensionless parameters is given in \eqref{eq_dimensionless}. Implementing Fourier Series \eqref{eq_w}--\eqref{eq_U3} and recalling Remark \ref{re:1} with choosing $N=3$, applying the Riemann transformation \eqref{eq_riemann}--\eqref{eq_rie}, the simulation model is obtained as \eqref{eq_matrix}--\eqref{eq_matrix5}  with $\epsilon = 3$, $\mu_1 = 1.8$, $\mu_2 = 0.2$, $a = 0.2$, $\theta = 0.057$, $\xi = 0.2$, $L = 9$, and the modal number $n = 0,1,2,3$. These parameters are calculated by recalling that $\epsilon=\frac{\rho}{k^{\prime}G}$, $\mu_1=\rho I_1$, $\mu_2=\rho I_2$, $a=\rho h$, $\theta=\frac{\rho_f U}{Mk^{\prime}Gh}$, $\xi=\frac{\rho_f U^2}{Mk^{\prime}Gh}$ and \eqref{eq_dimensionless}. The simulation is conducted by the finite difference method with a time step of $0.001$ and a space step of $0.05$. The initial values are $x_{1,n}(0) = 0.01$, $x_{2,n}(0) = 0.02$, $x_{3,n}(0) = 0.01$ and $p_n(0,x) = q_n(0,x) = r_n(0,x) = s_n(0,x) = u_n(0,x) = v_n(0,x) = 0.01\sin(\pi x)$.
\begin{table}[h!]
    \renewcommand\arraystretch{2}
	\begin{center}
		\caption{Physical parameters of the elastic plate.}
        \label{tab:2}
		\begin{tabular}{lll} 
			\hline
			Name & Value & Unit\\
			\hline
			Plate width $L_1^*$ &  $1$ & $m$\\
            Plate length $L_2^*$ &  $9$ & $m$\\
            Plate thickness $h^*$ &  $0.03$ & $m$\\
            Density $\rho^*$ & $2700$ & $kg/m^3$\\
            Young's Modulus $E^*$ &  $1.8\times 10^8$ & $Pa$\\
            Modulus of rigidity $G^*$ &  $1080$ & $Pa$\\
            Shear factor $k^{\prime}$ &  $0.833$ & $-$\\
            \makecell[l]{Moment of inertia of the \\cross-section per unit width  $I^*$} & $2.25\times 10^{-6}$ & $m^3$\\
            \makecell[l]{Equivalent moment of inertia for\\ bending in the $x$-direction $I_1^*$} & $0.27$ & $m^3$\\
            \makecell[l]{Equivalent moment of inertia for\\ bending in the $y$-direction $I_2^*$} & $0.03$ & $m^3$\\
            Mach number $M^*$ & $3$ & $-$\\
            Free stream density $\rho_f^*$ & $0.00453$ & $kg/m^3$\\
            Free stream velocity $U^*$ & $1020$ & $m/s$\\
			\hline
		\end{tabular}
	\end{center}
\end{table}

We apply the proposed output-feedback controller \eqref{eq_U1nob}--\eqref{eq_U3nob}  with the design parameters chosen as $\delta_1 = \delta_2 = \delta_3 = 5$, $L_1=L_2=L_3=5$, and the gains $K_n(1,y)$, $L_n(1,y)$, $\Phi_n(1)$ are computed using a power series approach as in \cite{powerseries} for corresponding modal number. Using the relationship \eqref{convert} and \eqref{eq:w1}, we can derive the evolution of $w_n(t,x), \alpha_n(t,x), \beta_n(t,x)$, which, by \eqref{eq_w}--\eqref{eq_beta}, are then used to represent the states of the elastic plant $w(t,x,y)$, $\beta(t,x,y)$, $\alpha(t,x,y)$, as shown in the following figures, where the red line highlights the controlled boundary in the 2D domain.

As shown in Fig. \ref{fig:op}, the plant is open-loop unstable, with the states rapidly growing to large magnitudes due to the presence of in-domain instability and the absence of in-domain damping in the simulation model. Figs. \ref{fig:cl_w}-\ref{fig:cl_beta}, where the red line denotes the controlled boundary, show that all 2-D states are fast convergent to zero under the proposed boundary control, as expected in Corollary \ref{cor:3}. Additionally, Fig. \ref{fig:ew}--\ref{fig:ebeta} shows that the observer errors of these 2-D states, i.e., $\tilde w(t,x,y)$, $\tilde\beta(t,x,y)$, $\tilde\alpha(t,x,y)$ are also convergent to zero, which demonstrates that the proposed observer rapidly converges to the actual PDE states. The observer-based output-feedback boundary control inputs $U_1(t,y),~U_2(t,y),~U_3(t,y)$ of the 2-D plant \eqref{eq_2dwave}--\eqref{eq_2dw} are also calculated by summing the modal components \eqref{eq_U1nob}--\eqref{eq_U3nob} via \eqref{eq:hatU1}--\eqref{eq:hatU3}, as shown in Fig. \ref{fig:U}.
\section{Conclusion and Future Work}
In this work, motivated by active wing flutter suppression in high-Mach-number flight regimes, we modeled the flow-induced vibration of a two-dimensional elastic plate as a coupled system of two-dimensional wave PDEs with in-domain instabilities, and designed an observer-based output-feedback boundary controller via backstepping. The two-dimensional control problem is decomposed into a series of one-dimensional modal systems via Fourier series expansion for control law design. Then, an observer is constructed to recover the distributed system states solely from available boundary measurements, thereby achieving output-feedback control. The proposed controller guarantees exponential stability of the closed-loop system, with a tunable convergence rate that can be arbitrarily assigned by users. Simulation results demonstrate that the proposed controller rapidly suppresses three-dimensional vibrations of the 2-D plate, even in the presence of in-domain instability sources.
Future work will extend the design to adaptive control to address the system parameter uncertainties and external disturbances.
\begin{figure}[htpb]
	\begin{center}
		\includegraphics[height=4cm, width = 8cm]{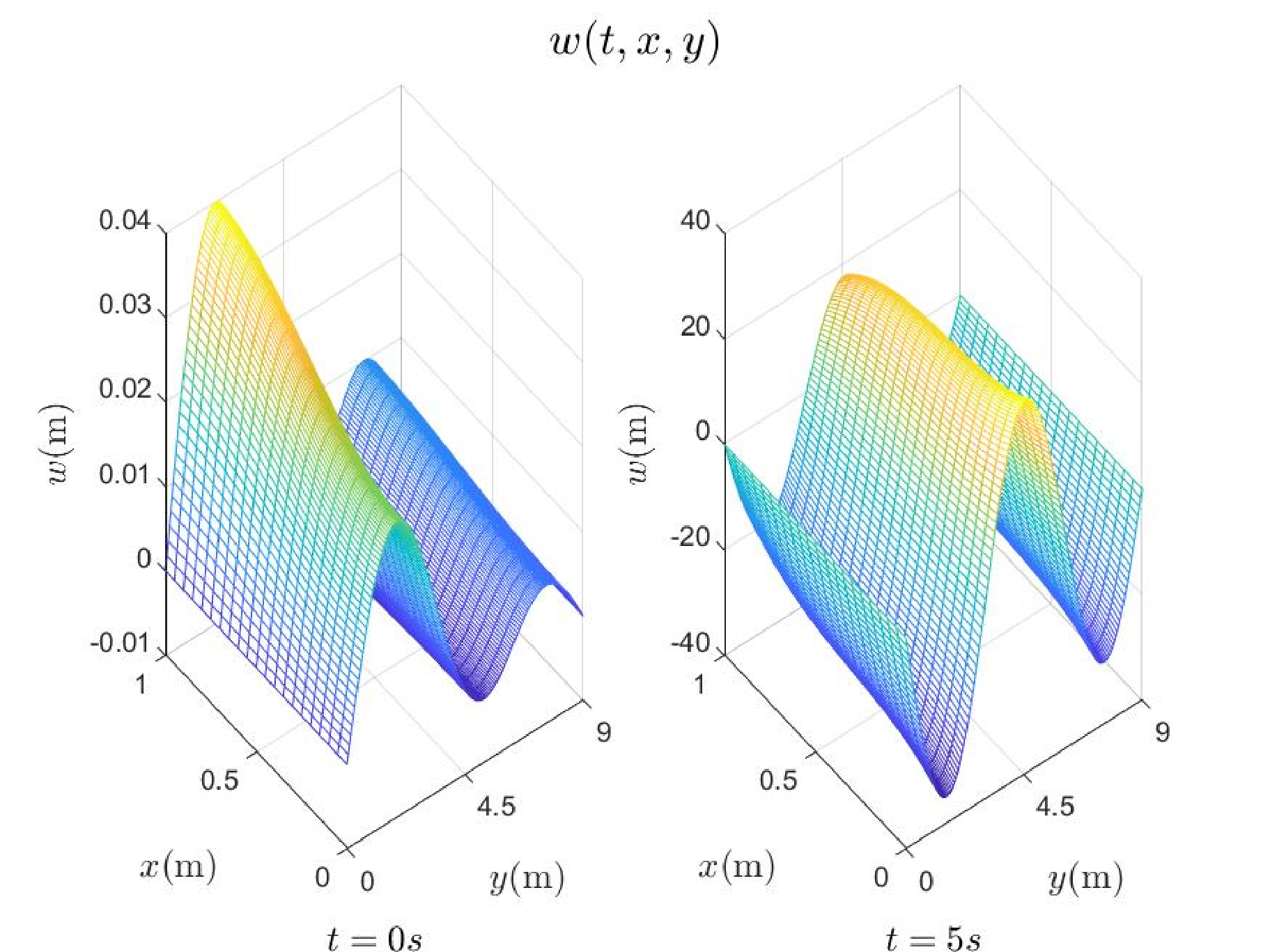}
		\caption{Results in the open loop.}
		\label{fig:op} 
	\end{center}
\end{figure}
\begin{figure}[htpb]
	\begin{center}
		\includegraphics[height=5.2cm, width = 8.5cm]{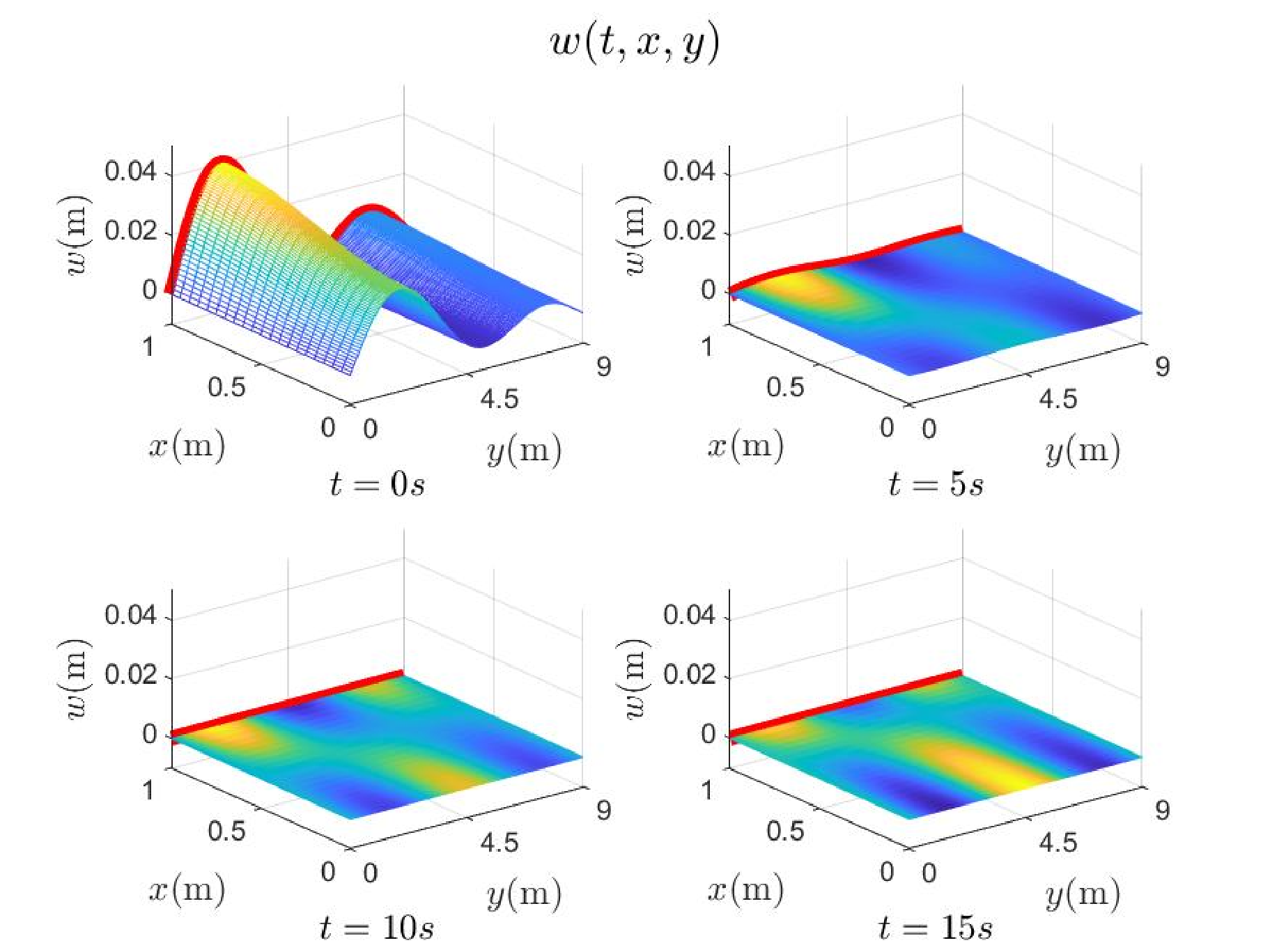}
		\caption{$w(t,x,y)$ under the proposed observer-based output-feedback boundary controller.}
		\label{fig:cl_w} 
	\end{center}
\end{figure}
\begin{figure}[htpb]
	\begin{center}
		\includegraphics[height=5.2cm, width = 8.5cm]{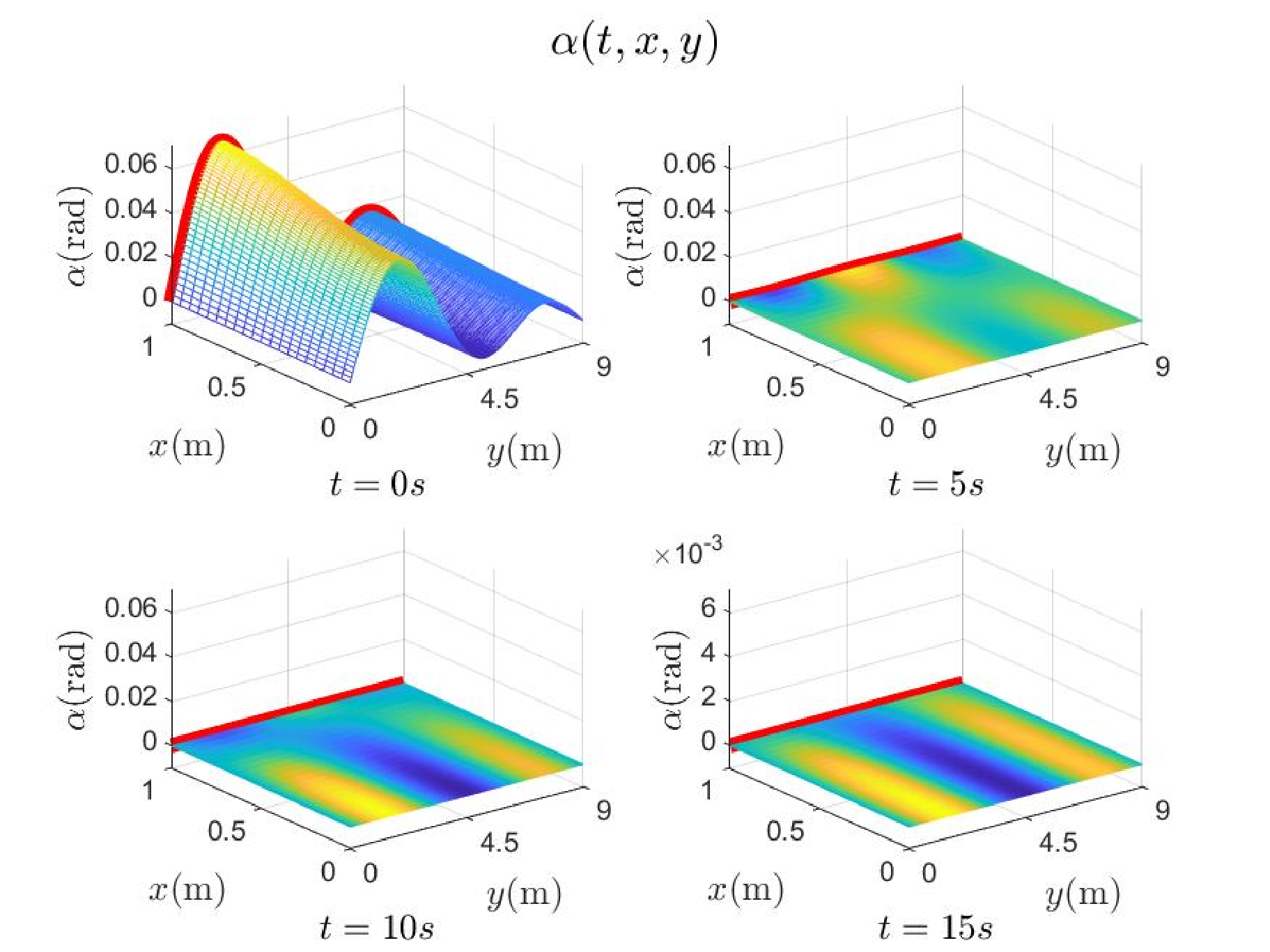}
		\caption{$\alpha(t,x,y)$ under the proposed observer-based output-feedback boundary controller.}
		\label{fig:cl_alpha} 
	\end{center}
\end{figure}
\begin{figure}[htpb]
	\begin{center}
		\includegraphics[height=6cm, width = 8.5cm]{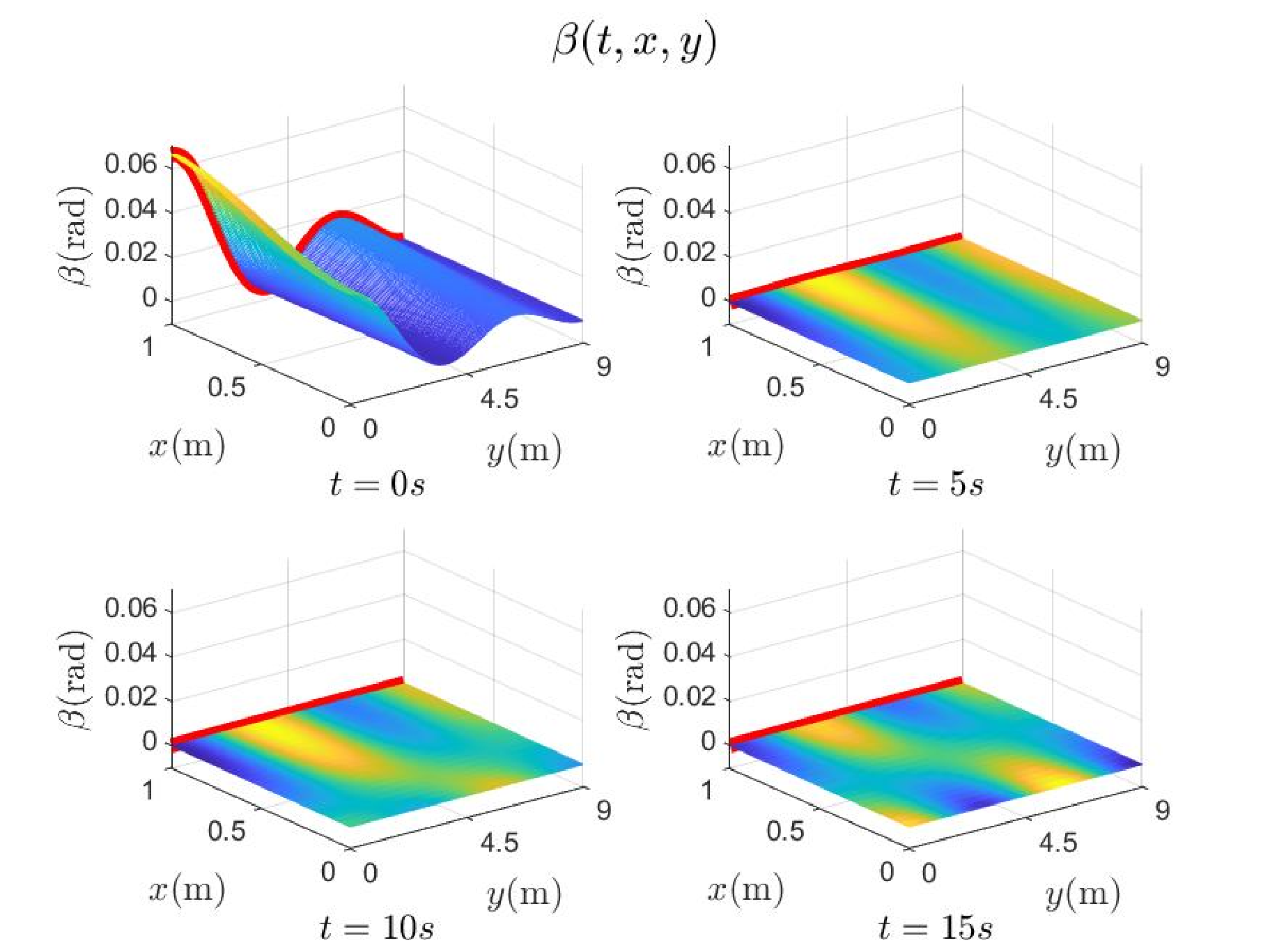}
		\caption{$\beta(t,x,y)$ under the proposed observer-based output-feedback boundary controller.}
		\label{fig:cl_beta} 
	\end{center}
\end{figure}
\begin{figure}[htpb]
	\begin{center}
		\includegraphics[height=6cm, width = 8.5cm]{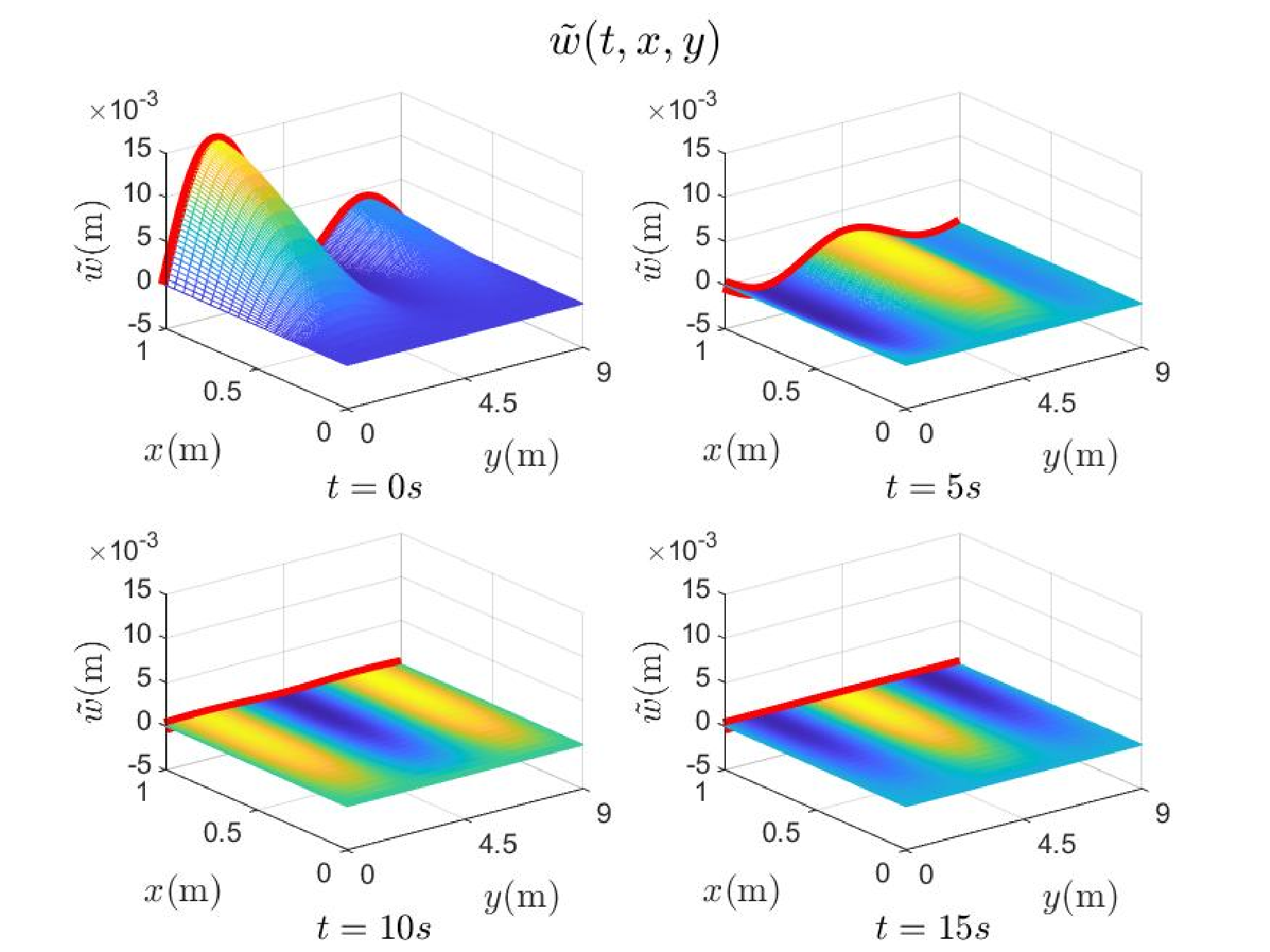}
		\caption{Observer errors $\tilde{w}(t,x,y)$.}
		\label{fig:ew} 
	\end{center}
\end{figure}
\begin{figure}[htpb]
	\begin{center}
		\includegraphics[height=6cm, width = 8.5cm]{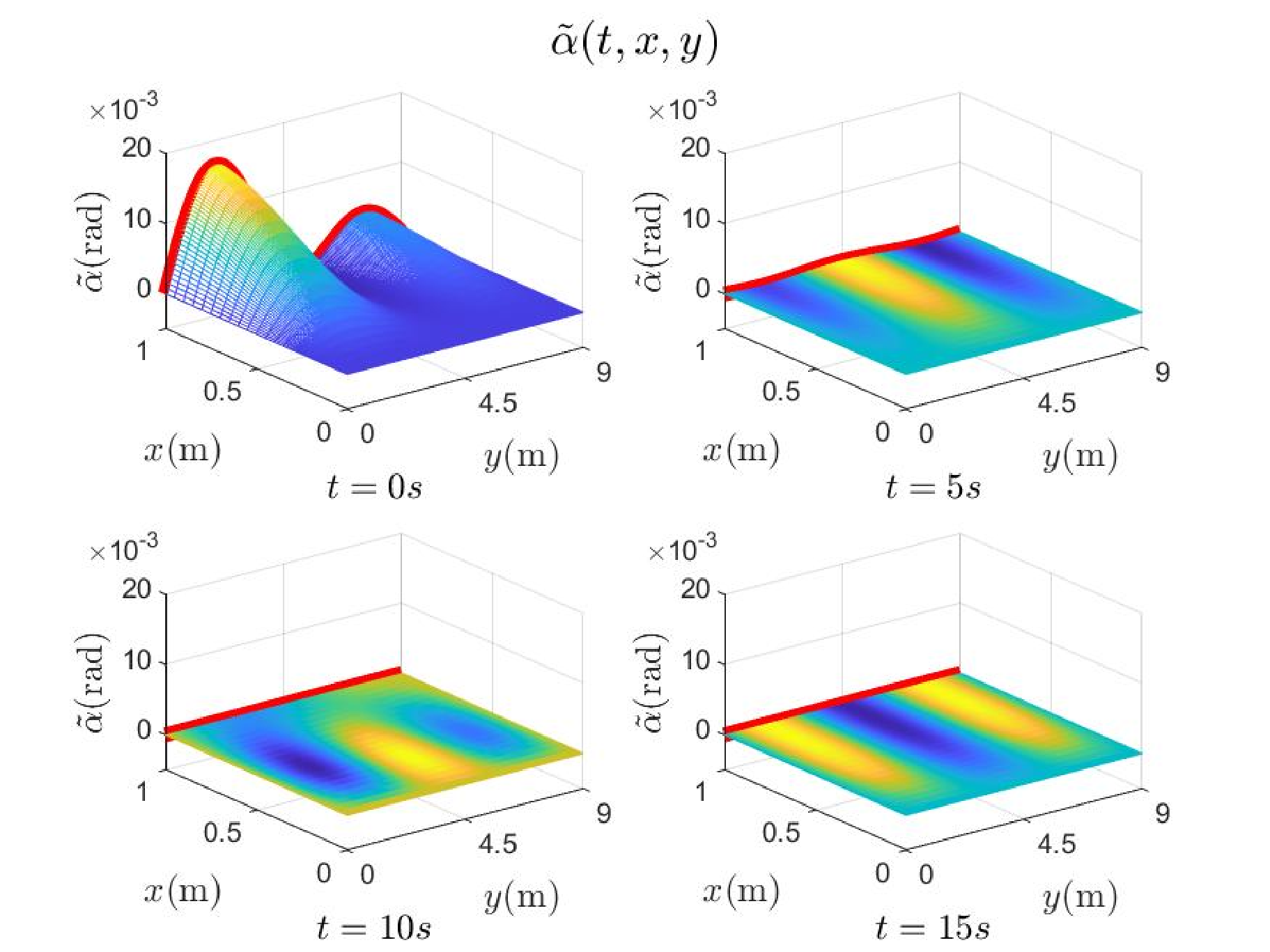}
		\caption{Observer errors $\tilde{\alpha}(t,x,y)$.}
		\label{fig:ealpha} 
	\end{center}
\end{figure}
\begin{figure}[htpb]
	\begin{center}
		\includegraphics[height=6cm, width = 8.5cm]{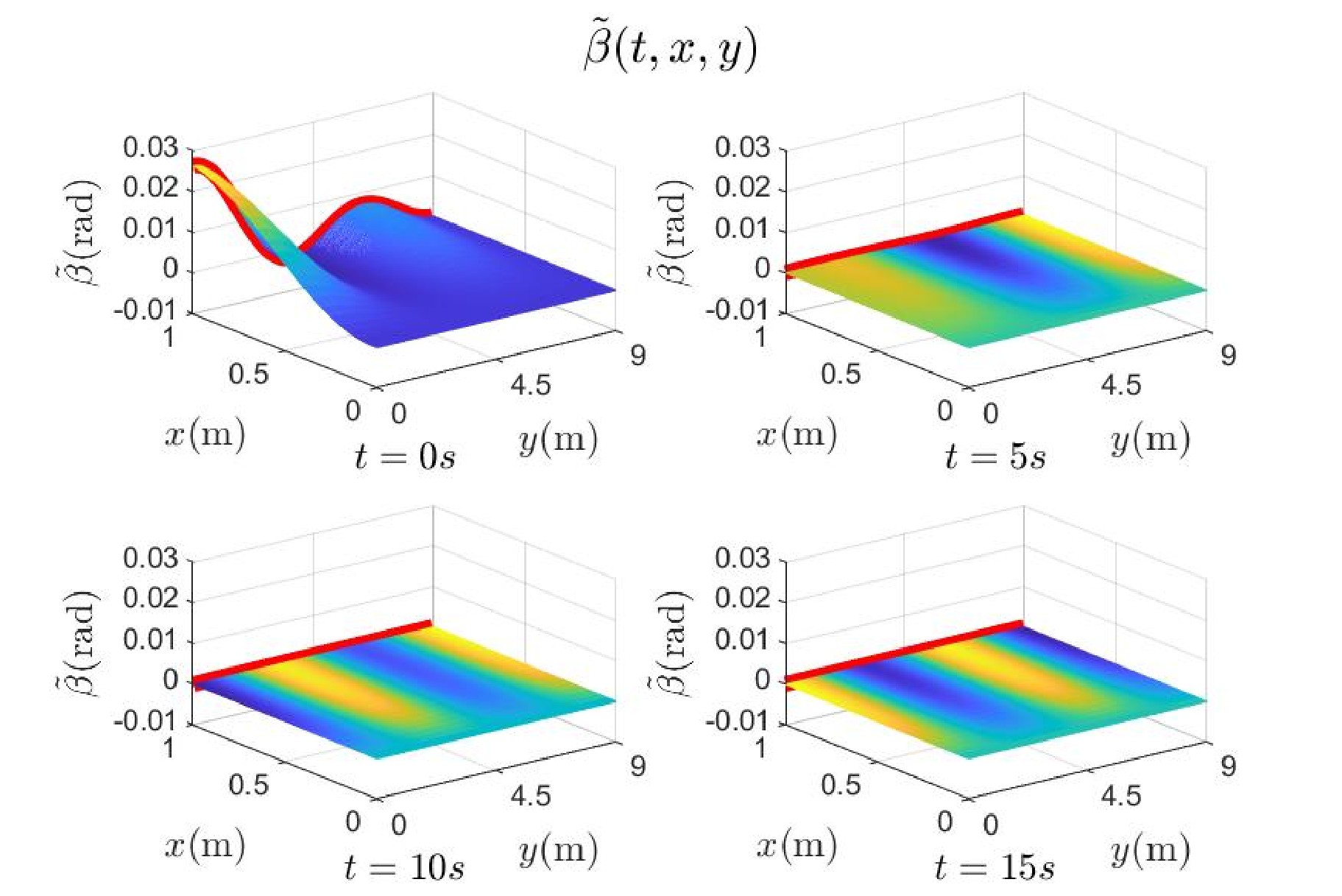}
		\caption{Observer errors $\tilde{\beta}(t,x,y)$.}
		\label{fig:ebeta} 
	\end{center}
\end{figure}
\begin{figure}[htpb]
	\begin{center}
		\includegraphics[height=6cm, width = 9cm]{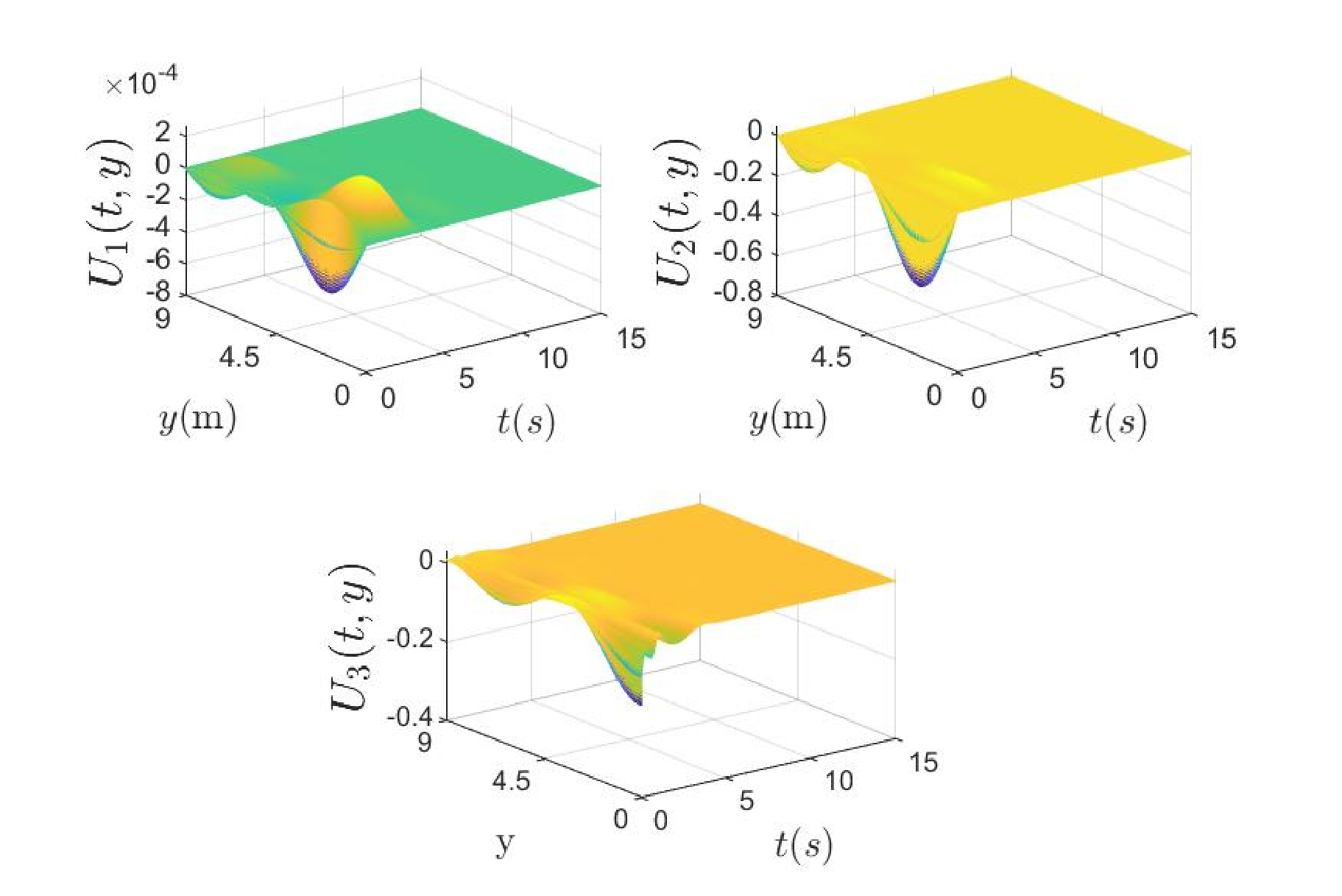}
		\caption{Output-feedback control inputs.}
		\label{fig:U} 
	\end{center}
\end{figure}

\appendices

\section{Expression of the coefficients in \eqref{eq_matrix}--\eqref{eq_matrix5}} \label{adx:mat_coeff}
\setcounter{equation}{0}
\renewcommand{\theequation}{A.\arabic{equation}}
The expressions of the coefficients in \eqref{eq_matrix}--\eqref{eq_matrix5} are shown as follows.
\begin{align}
	\Sigma &= 
	\begin{bmatrix}
		\begin{smallmatrix}
			\frac{1}{\sqrt{\epsilon}} & 0 &0\\
			0 & \frac{1}{\sqrt{\mu_1}} & 0\\
			0 & 0 & \frac{1}{\sqrt{\mu_2}}
		\end{smallmatrix}
	\end{bmatrix},
	A = 
	\begin{bmatrix}
		\begin{smallmatrix}
			0 & -\frac{1}{\sqrt{\epsilon}} & 0\\
			0 & 0 & 0\\
			0 & \frac{n\pi}{L\sqrt{\mu_2}} & 0
		\end{smallmatrix}
	\end{bmatrix},\\
	\nonumber
	C &= 
	\begin{bmatrix}  
		\begin{smallmatrix}
			-1 & 0 & 0\\
			0 & -1 & 0\\
			0 & 0 & -1
		\end{smallmatrix}
	\end{bmatrix},
	D = 
	\begin{bmatrix}
		\begin{smallmatrix}
			0 & 2 & 0\\
			0 & 0 & 0\\
			0 & -\frac{2n\pi}{L} & 0
		\end{smallmatrix}
	\end{bmatrix},\\
	\nonumber
	F_{11}(x) &= 
	\begin{bmatrix}
		\begin{smallmatrix}
			\frac{c_1(x)}{2} & c_2(x) & 0\\
			\frac{a(c_5(x)+c_6(x))}{4\epsilon\sqrt{\mu_1}} & 0 & -\frac{n\pi}{2L\sqrt{\mu_1}}\\
			0 & \frac{n\pi}{2L\sqrt{\mu_2}} & 0
		\end{smallmatrix}
	\end{bmatrix},\\
	\nonumber
	F_{12}(x) &=
	\begin{bmatrix}
		\begin{smallmatrix}
			-\frac{c_1(x)}{2} & 0 & 0\\
			\frac{a(c_5(x)-c_6(x))}{4\epsilon\sqrt{\mu_1}} & 0 & 0\\
			0 & 0 & 0
		\end{smallmatrix}
	\end{bmatrix},\\
	\nonumber
	F_{21}(x) &= 
	\begin{bmatrix}
		\begin{smallmatrix}
			\frac{c_3(x)}{2} & c_4(x) & 0\\
			-\frac{a(c_5(x)+c_6(x))}{4\epsilon\sqrt{\mu_1}} & 0 & \frac{n\pi}{2L\sqrt{\mu_1}}\\
			0 & -\frac{n\pi}{2L\sqrt{\mu_2}} & 0
		\end{smallmatrix}
	\end{bmatrix},\\
	\nonumber
	F_{22}(x) &=
	\begin{bmatrix}
		\begin{smallmatrix}
			\frac{c_3(x)}{2} & 0 & 0\\
			-\frac{a(c_5(x)-c_6(x))}{4\epsilon\sqrt{\mu_1}} & 0 & 0\\
			0 & 0 & 0
		\end{smallmatrix}
	\end{bmatrix},\\
	\nonumber
	F_{13}(x) &= 
	\begin{bmatrix}
		\begin{smallmatrix}
			\frac{2n^2\pi^2}{L^2}c_2(x) & 0 & -\frac{2n\pi}{L} c_2(x)\\
			0 & -\frac{n^2\pi^2\epsilon+aL^2}{\epsilon L^2\sqrt{\mu_1}} & 0\\
			\frac{an\pi}{\epsilon L\sqrt{\mu_2}} & 0 & -\frac{n^2\pi^2\epsilon+aL^2}{\epsilon L^2\sqrt{\mu_2}}
		\end{smallmatrix}
	\end{bmatrix},\\
	\nonumber
	F_{14}(x,y) &= 
	\begin{bmatrix}
		\begin{smallmatrix}
			-\frac{n^2\pi^2}{2\sqrt{\epsilon}L^2}f_{11}(x,y) & 0 & -\frac{n\pi}{L} c_2(x)\\
			0 & -\frac{n^2\pi^2\epsilon+aL^2}{2\epsilon L^2 \sqrt{\mu_1}} & 0\\
			\frac{an\pi}{2\epsilon L\sqrt{\mu_2}}c_5(y) & 0 & -\frac{n^2\pi^2\epsilon+aL^2}{2\epsilon L^2\sqrt{\mu_2}}
		\end{smallmatrix}
	\end{bmatrix},\\
	\nonumber
	F_{15}(x,y) &= 
	\begin{bmatrix}
		\begin{smallmatrix}
			-\frac{n^2\pi^2}{2\sqrt{\epsilon}L^2}f_{12}(x,y) & 0 & -\frac{n\pi}{L} c_2(x)\\
			0 & -\frac{n^2\pi^2\epsilon+aL^2}{2\epsilon L^2\sqrt{\mu_1}} & 0\\
			\frac{an\pi}{2\epsilon L\sqrt{\mu_2}}c_6(y) & 0 & -\frac{n^2\pi^2\epsilon+aL^2}{2\epsilon L^2\sqrt{\mu_2}}
		\end{smallmatrix}
	\end{bmatrix},\\
	\nonumber
	F_{23}(x) &= 
	\begin{bmatrix}
		\begin{smallmatrix}
			\frac{2n^2\pi^2}{L^2}c_4(x) & 0 & -\frac{2n\pi}{L} c_4(x)\\
			0 & \frac{n^2\pi^2\epsilon+aL^2}{\epsilon L^2\sqrt{\mu_1}} & 0\\
			-\frac{an\pi}{\epsilon L\sqrt{\mu_2}} & 0 & \frac{n^2\pi^2\epsilon+aL^2}{\epsilon L^2\sqrt{\mu_2}}
		\end{smallmatrix}
	\end{bmatrix},\\
	\nonumber
	F_{24}(x,y) &= 
	\begin{bmatrix}
		\begin{smallmatrix}
			\frac{n^2\pi^2}{2\sqrt{\epsilon}L^2}f_{21}(x,y) & 0 & -\frac{n\pi}{L} c_4(x)\\
			0 & \frac{n^2\pi^2\epsilon+aL^2}{2\epsilon L^2\sqrt{\mu_1}} & 0\\
			-\frac{an\pi}{2\epsilon L\sqrt{\mu_2}}c_5(y) & 0 & \frac{n^2\pi^2\epsilon+aL^2}{2\epsilon L^2\sqrt{\mu_2}}
		\end{smallmatrix}
	\end{bmatrix},\\
	\nonumber
	F_{25}(x,y) &= 
	\begin{bmatrix}
		\begin{smallmatrix}
			\frac{n^2\pi^2}{2\sqrt{\epsilon}L^2}f_{22}(x,y) & 0 & -\frac{n\pi}{L} c_4(x)\\
			0 & \frac{n^2\pi^2\epsilon+aL^2}{2\epsilon L^2\sqrt{\mu_1}} & 0\\
			-\frac{an\pi}{2\epsilon L\sqrt{\mu_2}}c_6(y) & 0 & \frac{n^2\pi^2\epsilon+aL^2}{2\epsilon L^2\sqrt{\mu_2}}
		\end{smallmatrix}
	\end{bmatrix},
\end{align}
where $c_1(x) = \bar{c}_2 \exp(\sqrt{\epsilon}(\bar{c}_1 - \bar{c}_2)x)$, $c_2(x) = -\frac{1}{2\sqrt{\epsilon}}\\ \times\exp(\sqrt{\epsilon}\bar{c}_1 x)$, $c_3(x) = -\bar{c}_1 \exp(-\sqrt{\epsilon}(\bar{c}_1 - \bar{c}_2)x)$, $c_4(x) = \frac{1}{2\sqrt{\epsilon}}\exp(\sqrt{\epsilon}\bar{c}_2 x)$, $c_5(x) = \exp(-\sqrt{\epsilon}\bar{c}_1 x)$, $c_6(x) = \exp(-\sqrt{\epsilon}\bar{c}_2 x)$, $f_{11}(x,y) = \exp(\sqrt{\epsilon}\bar{c}_1(x-y))$, $f_{12}(x,y) = \exp(\sqrt{\epsilon}(\bar{c}_1 x - \bar{c}_2 y))$, $f_{21}(x,y) = \exp(\sqrt{\epsilon}(\bar{c}_2 x - \bar{c}_1 y))$ and $f_{22}(x,y) = \exp(\sqrt{\epsilon}\bar{c}_2(x-y))$.

\section{Expression of $\mathcal{F}_{ij}(\xi)$ and  $\mathcal{D}_{ij}$} \label{adx:input_coeff}
\setcounter{equation}{0}
\renewcommand{\theequation}{B.\arabic{equation}}

For $i = 1,2,3$ and $j = 1,2,3,4,5,6$, the functions $\mathcal{F}_{ij}(\xi)$ and $\mathcal{D}_{ij}$ in \eqref{eq_U1n}--\eqref{eq_U3n} and \eqref{eq_U1nob}--\eqref{eq_U3nob} are shown as follows
\begin{align}
	\nonumber
	\mathcal{F}_{i1}(\xi) &= \sqrt{\epsilon}\bar{c}_1 \exp(-\sqrt{\epsilon}\bar{c}_1 \xi)k_{i1,n}(1,\xi)\\
	\nonumber
	&\quad + \sqrt{\epsilon}\bar{c}_2 \exp(-\sqrt{\epsilon}\bar{c}_2 \xi)l_{i1,n}(1,\xi)\\
	\nonumber
	&\quad - \exp(-\sqrt{\epsilon}\bar{c}_1 \xi)k_{i1,n,\xi}(1,\xi)\\
	&\quad - \exp(-\sqrt{\epsilon}\bar{c}_2 \xi)l_{i1,n,\xi}(1,\xi),\\
	\nonumber
	\mathcal{F}_{i2}(\xi) &= \sqrt{\epsilon}\left(\exp(-\sqrt{\epsilon}\bar{c}_1\xi)k_{i1,n}(1,\xi)\right.\\
	&\left.\quad - \exp(-\sqrt{\epsilon}\bar{c}_2\xi)l_{i1,n}(1,\xi)\right),\\
	\mathcal{F}_{i3}(\xi) &= k_{i2,n,\xi}(1,\xi) + l_{i2,n,\xi}(1,\xi),\\
	\mathcal{F}_{i4}(\xi) &= \sqrt{\mu}\left(k_{i2,n}(1,\xi) - l_{i2,n}(1,\xi)\right),\\
	\mathcal{F}_{i5}(\xi) &= k_{i3,n,\xi}(1,\xi) + l_{i3,n,\xi}(1,\xi),\\
	\mathcal{F}_{i6}(\xi) &= \sqrt{\mu}\left(k_{i3,n}(1,\xi) - l_{i3,n}(1,\xi)\right),\\
	\nonumber
	\mathcal{D}_{i1} &= \exp(-\sqrt{\epsilon}\bar{c}_1)k_{i1,n}(1,1)\\
	&\quad + \exp(-\sqrt{\epsilon}\bar{c}_2)l_{i1,n}(1,1),\\
	\mathcal{D}_{i2} &= k_{i1,n}(1,0) + l_{i1,n}(1,0) - \phi_{i1,n}(1),\\
	\mathcal{D}_{i4} &= k_{i2,n}(1,0) + l_{i2,n}(1,0) - \phi_{i2,n}(1),\\
	\mathcal{D}_{i5} &= k_{i3,n}(1,1) + l_{i3,n}(1,1),\\
	\mathcal{D}_{i6} &= k_{i3,n}(1,0) + l_{i3,n}(1,0) - \phi_{i3,n}(1)
\end{align}
with
\begin{align}
	\mathcal{D}_{13} &= k_{12,n}(1,1) + l_{12,n}(1,1),\\
	\mathcal{D}_{23} &= k_{22,n}(1,1) + l_{22,n}(1,1),\\
	\mathcal{D}_{33} &= k_{32,n}(1,1) + l_{32,n}(1,1) + \tfrac{n\pi}{L}.
\end{align}

\section{Proof of Theorem \ref{thm:1}: Well-posedness of the Kernel Equations of $K$ and $L$} \label{adx:wellposed}
\setcounter{equation}{0}
\renewcommand{\theequation}{C.\arabic{equation}}

To prove the well-posedness of the kernel equations, we transform the kernel equations into integral equations and use the method of successive approximations. For $1 \leq i,j \leq 3$, denote
\begin{align}
	&\scalebox{0.95}{$\Lambda^{+} = \Lambda^{-} = \Sigma, \quad Q_0 = C, \quad \Omega(x) = \{\omega_{ij}(x)\}$},\\
	&\scalebox{0.95}{$A = \{a_{ij}\}, \quad D = \{d_{ij}\}, \quad C = \{c_{ij}\}$},\\
	&\scalebox{0.95}{$\Sigma^{-+}(x) = F_{11}(x) - F_{12}(x) = \{\sigma_{ij}^{-+}(x)\}$},\\
	&\scalebox{0.95}{$\Sigma^{--}(x) = F_{11}(x) + F_{12}(x) = \{\sigma_{ij}^{--}(x)\}$},\\
	&\scalebox{0.95}{$\Sigma^{++}(x) = F_{21}(x) - F_{22}(x) = \{\sigma_{ij}^{++}(x)\}$},\\
	&\scalebox{0.95}{$\Sigma^{+-}(x) = F_{21}(x) + F_{22}(x) = \{\sigma_{ij}^{+-}(x)\}$},\\
	&\scalebox{0.95}{$F_{13}(x) = \{\varepsilon_{ij}^{++}(x)\}, \quad F_{23}(x) = \{\varepsilon_{ij}^{+-}(x)\}$},\\
	&\scalebox{0.95}{$F_{14}(x,y) = \{f_{ij}^{++}(x,y)\}, F_{15}(x,y) = \{f_{ij}^{+-}(x,y)\}$},\\
	&\scalebox{0.95}{$F_{24}(x,y) = \{f_{ij}^{-+}(x,y)\}, F_{25}(x,y) = \{f_{ij}^{--}(x,y)\}$}.
\end{align}	
Developing Eqs. \eqref{eq_kernel}--\eqref{eq_ker}, and applying the method described in \cite{coupled}, we embed the ODE into the domain $\Gamma = \{0 \leq \xi \leq x \leq 1\}$ by denoting
\begin{align}
	I_{\{y=0\}}(x,y) = \left\{
		\begin{array}{l}
			1 \quad if \text{ }y=0\\
			0 \quad otherwise.
		\end{array}
		\right.
\end{align}
Defining $\tilde{\phi}_{ij}$ such that $\forall(x,y) \in \Gamma, \tilde{\phi}_{ij}(x,y) = I_{\{y=0\}}(x,y)\phi_{ij}(x)$, for $1 \leq i \leq 3, 1 \leq j \leq 3$, we get the following set of kernel PDEs:	
\begin{align}
	\label{eq_L} \nonumber
	&\scalebox{0.9}{$\lambda_i \partial_x L_{ij}(x,\xi) - \lambda_j \partial_\xi L_{ij}(x,\xi) = \scalebox{0.9}{$\sum_{k=1}^{3}$} \sigma_{kj}^{++}(\xi)L_{ik}(x,\xi)$}\\
	\nonumber
	&\scalebox{0.9}{$ + \scalebox{0.9}{$\sum_{p=1}^{3}$} \sigma_{pj}^{-+}(\xi)K_{ip}(x,\xi) - \scalebox{0.9}{$\sum_{i<p}^{}$} L_{pj}(x,\xi)\omega_{ip}(x)$}\\
	\nonumber
	&\scalebox{0.9}{$- f_{ij}^{+-}(x,\xi) + \smallint\nolimits_{\xi}^{x} \scalebox{0.9}{$\sum_{p=1}^{3}$} f_{pj}^{+-}(s,\xi)K_{ip}(x,s) ds$}\\
	&\scalebox{0.9}{$ + \smallint\nolimits_{\xi}^{x} \scalebox{0.9}{$\sum_{k=1}^{3}$} f_{kj}^{--}(s,\xi)L_{ik}(x,s) ds$},\\
	\label{eq_K} \nonumber
	&\scalebox{0.9}{$\lambda_i \partial_x K_{ij}(x,\xi) + \lambda_j \partial_\xi K_{ij}(x,\xi) = \scalebox{0.9}{$\sum_{k=1}^{3}$} \sigma_{kj}^{--}(\xi)K_{ik}(x,\xi)$}\\
	\nonumber
	&\scalebox{0.9}{$ + \scalebox{0.9}{$\sum_{p=1}^{3}$} \sigma_{pj}^{+-}(\xi)L_{ip}(x,\xi) - \scalebox{0.9}{$\sum_{i<p}^{}$} K_{pj}(x,\xi)\omega_{ip}(x)$}\\
	\nonumber
	&\scalebox{0.9}{$- f_{ij}^{++}(x,\xi) + \smallint\nolimits_{\xi}^{x} \scalebox{0.9}{$\sum_{p=1}^{3}$} f_{pj}^{++}(s,\xi)K_{ip}(x,s) ds$} \\
	&\scalebox{0.9}{$ + \smallint\nolimits_{\xi}^{x} \scalebox{0.9}{$\sum_{k=1}^{3}$} f_{kj}^{-+}(s,\xi)L_{ik}(x,s) ds$},\\
	\label{eq_phi} \nonumber
	&\scalebox{0.9}{$\lambda_i \partial_x \tilde{\phi}_{ij}(x,\xi) = I_{\{\xi=0\}}(x,\xi)\Big[\scalebox{0.9}{$\sum_{k=1}^{3}$} a_{kj}\tilde{\phi}_{ik}(x,\xi)$}\\
	\nonumber
	&\scalebox{0.9}{$+ \scalebox{0.9}{$\sum_{p=1}^{3}$} \lambda_p d_{pj}L_{ip}(x,0) - \scalebox{0.9}{$\sum_{i<p}^{}$} \omega_{ip}(x)\tilde{\phi}_{pj}(x,\xi)$}\\
	\nonumber
	&\scalebox{0.9}{$- \varepsilon_{ij}^{++}(x) + \smallint\nolimits_{0}^{x} \scalebox{0.9}{$\sum_{p=1}^{3}$} \varepsilon_{pj}^{++}(x)K_{ip}(x,\xi) d\xi$}\\
	&\scalebox{0.9}{$+ \smallint\nolimits_{0}^{x} \scalebox{0.9}{$\sum_{k=1}^{3}$} \varepsilon_{kj}^{+-}(x)L_{ik}(x,\xi) d\xi$}\Big],
\end{align}
with the following set of boundary conditions
\begin{align}
	\label{eq_lbc} & \scalebox{0.9}{$L_{ij}(x,x) = -\tfrac{\sigma_{ij}^{-+}(x)}{\lambda_i + \lambda_j} = l_{ij}(x)$},\\
	\label{eq_kbcx} & \scalebox{0.9}{$K_{ij}(x,x) = -\tfrac{\sigma_{ij}^{--}(x)}{\lambda_i - \lambda_j} = k_{ij}(x)\quad(j<i)$},\\
	\label{eq_kbc0} & \scalebox{0.9}{$\lambda_j k_{ij}(x,0) = \scalebox{0.9}{$\sum_{k=1}^{3}$} \lambda_k L_{ik}(x,0) c_{kj} + \lambda_j \tilde{\phi}_{ij}(x,0)$},\\
	\label{eq_phibc} & \scalebox{0.9}{$\tilde{\phi}_{ij}(x,x) = I_{\{\xi=0\}}(x,x)\phi_{ij}(0)$}.
\end{align}
Besides, \eqref{eq_bcomega} imposes
\begin{align}\label{eq_omegaij}
	\scalebox{0.9}{$\forall i \leq j, \text{ } \omega_{ij}(x) = (\lambda_i - \lambda_j)k_{ij}(x,x) + \sigma_{ij}^{--}(x)$}.
\end{align}
By induction, let us consider the following property $P(s)$ defined for all $1 \leq s \leq 3$:
For $\forall 1 \leq j \leq 3$ and $\forall 3 + 1 - s \leq i \leq 3$, the problem \eqref{eq_L}--\eqref{eq_omegaij} has a unique solution $k_{ij}(\cdot,\cdot)$, $l_{ij}(\cdot,\cdot)$, $\tilde{\phi}_{ij}(\cdot,\cdot) \in L^{\infty}(\Gamma)$.
Let us assume that the property $P(s-1)(1 < s \leq 3-1)$ is true. We consequently have that $\forall 3+2-s \leq p \leq 3$, $\forall 1 \leq j \leq 3$, $k_{pj}(\cdot,\cdot),l_{pj}(\cdot,\cdot),$ and $\tilde{\phi}_{pj}(\cdot,\cdot)$ are bounded. The proof follows along the line in \cite{coupled} and \cite{Heterodirectional} and is skipped due to space limitations. In the following we take $i = 3 + 1 - s$, we now show that \eqref{eq_L}--\eqref{eq_omegaij} is well-posed and that $k_{ij}(\cdot,\cdot), l_{ij}(\cdot,\cdot), \tilde{\phi}_{ij}(\cdot,\cdot) \in L^{\infty}(\Gamma)$.

\subsection{Method of Characteristics}
\subsubsection{Characteristics of the $L$ kernels}

For each $1 \leq i,j \leq 3$, and $(x,\xi) \in \Gamma$, we define the following characteristic lines $(x_{ij}(x,\xi;\cdot), \xi_{ij}(x,\xi;\cdot))$ corresponding to \eqref{eq_L}:
\begin{align}
	\label{eq_lijx} & \scalebox{0.9}{$\left\{
		\begin{array}{l}
			\frac{d x_{ij}}{ds}(x,\xi;s) = -\lambda_i, \quad s \in [0, s_{ij}^{F}(x,\xi)]\\
			x_{ij}(x,\xi;0) = x, \text{ } x_{ij}(x,\xi;s_{ij}^{F}(x,\xi)) = x_{ij}^{F}(x,\xi)
		\end{array}
		\right.$}\\
	\label{eq_lijxi} & \scalebox{0.9}{$\left\{
		\begin{array}{l}
			\frac{d \xi_{ij}}{ds}(x,\xi;s) = \lambda_j, \quad s \in [0, s_{ij}^{F}(x,\xi)]\\
			\xi_{ij}(x,\xi;0) = \xi, \text{ } \xi_{ij}(x,\xi;s_{ij}^{F}(x,\xi)) = x_{ij}^{F}(x,\xi)
		\end{array}
		\right.$}
\end{align}	
These lines originate at the point $(x,\xi)$ and terminate at the hypotenuse at the point $(x_{ij}^{F}(x,\xi),x_{ij}^{F}(x,\xi))$. Here, the expressions of $x_{ij}(x,\xi,s)$, $\xi_{ij}(x,\xi,s)$, $x_{ij}^{F}(x,\xi)$ and $s_{ij}^{F}(x,\xi)$, which are straightforward to obtain, are omitted for simplicity.  Integrating \eqref{eq_L} along the characteristic lines and plugging in the boundary condition \eqref{eq_lbc} yields
\begin{align}\label{eq_lij}
	\nonumber
	&\scalebox{0.9}{$L_{ij}(x,\xi) = l_{ij}(x_{ij}^{F}) + \smallint\nolimits_{0}^{s_{ij}^{F}(x,\xi)} \left[\scalebox{0.9}{$\sum_{k=1}^{3}$} \sigma_{kj}^{++}(\xi_{ij}(x,\xi;s))\right.$}\\
	\nonumber
	&\scalebox{0.9}{$\left.\times L_{ik}(x_{ij}(x,\xi;s),\xi_{ij}(x,\xi;s))\right.$}\\
	\nonumber
	&\scalebox{0.9}{$\left. + \scalebox{0.9}{$\sum_{p=1}^{3}$} \sigma_{pj}^{-+}(\xi_{ij}(x,\xi;s))K_{ip}(x_{ij}(x,\xi;s),\xi_{ij}(x,\xi;s))\right.$}\\
	\nonumber
	&\scalebox{0.9}{$\left. - \scalebox{0.9}{$\sum_{i<p}^{}$} L_{pj}(x_{ij}(x,\xi;s),\xi_{ij}(x,\xi;s)) ((\lambda_i - \lambda_p) \right.$}\\
	\nonumber
	&\scalebox{0.9}{$\left. \times K_{ip}(x_{ij}(x,\xi;s),\xi_{ij}(x,\xi;s)) + \sigma_{ip}^{--}(x_{ij}(x,\xi;s)))\right.$}\\
	\nonumber
	&\scalebox{0.9}{$\left. - f_{ij}^{+-}(x_{ij}(x,\xi;s),\xi_{ij}(x,\xi;s))\right.$}\\
	\nonumber
	&\scalebox{0.9}{$\left. + \smallint\nolimits_{\xi_{ij}}^{x_{ij}} \scalebox{0.9}{$\sum_{p=1}^{3}$} (f_{pj}^{+-}(\tau,\xi_{ij}(x,\xi;s))K_{ip}(x_{ij}(x,\xi;s),\tau) \right.$}\\
	&\scalebox{0.9}{$\left. + f_{pj}^{--}(\tau,\xi_{ij}(x,\xi;s))L_{ip}(x_{ij}(x,\xi;s),\tau)) d\tau \right] ds$}.
\end{align} 
We can notice that the fourth line of \eqref{eq_lij} uses the expression of $L_{pj}$ for $i <p$. This term is known and bounded (induction assumption).

\subsubsection{Characteristics of the $\tilde{\phi}$ kernels}

For each $1 \leq i,j \leq 3$, and $(x,\xi) \in \Gamma$, we define the following characteristic lines $(\kappa_{ij}(x,\xi;\cdot), \iota_{ij}(x,\xi;\cdot))$ corresponding to \eqref{eq_phi}:
\begin{align}
	\label{eq_phiijx} &\scalebox{0.9}{$\left\{
		\begin{array}{l}
			\frac{d \kappa_{ij}}{d\eta}(x,\xi;\eta) = -\lambda_i, \quad \eta \in [0, \eta_{ij}^{F}(x,\xi)]\\
			\kappa_{ij}(x,\xi;0) = x, \quad \kappa_{ij}(x,\xi;\eta_{ij}^{F}(x,\xi)) = \xi
		\end{array}
		\right.$}\\
	\label{eq_phiijxi} &\scalebox{0.9}{$\left\{
	\begin{array}{l}
		\frac{d \iota_{ij}}{d\eta}(x,\xi;\eta) = 0, \quad \eta \in [0, \eta_{ij}^{F}(x,\xi)]\\
		\iota_{ij}(x,\xi;0) = \xi, \quad \iota_{ij}(x,\xi;\eta_{ij}^{F}(x,\xi)) = \xi
	\end{array}
	\right.$}
\end{align}
These lines originate at the point $(x,\xi)$ and terminate at the hypotenuse at the point $(\xi,\xi)$. Integrating \eqref{eq_phi} along the characteristic lines and plugging in the boundary condition \eqref{eq_phibc} yields
\begin{align}\label{eq_phiij}
	\nonumber
	&\scalebox{0.9}{$\tilde{\phi}_{ij}(x,\xi) = I_{\{\xi=0\}}(\xi,\xi)\phi_{ij}(0)$}\\
	\nonumber
	&\scalebox{0.9}{$ + \smallint\nolimits_{0}^{\eta_{ij}^{F}(x,\xi)}I_{\{\xi=0\}}(\kappa_{ij}(x,\xi;\eta),\iota_{ij}(x,\xi;\eta))$}\\
	\nonumber
	&\scalebox{0.9}{$ \left[\scalebox{0.9}{$\sum_{k=1}^{3}$} a_{kj}\tilde{\phi}_{ik}(\kappa_{ij}(x,\xi;\eta),\iota_{ij}(x,\xi;\eta))\right.$}\\
	\nonumber
	&\scalebox{0.9}{$ \left. + \scalebox{0.9}{$\sum_{p=1}^{3}$} \lambda_p d_{pj} L_{ip}(\kappa_{ij}(x,\xi;\eta),0)\right.$}\\
	\nonumber
	&\scalebox{0.9}{$ \left. - \scalebox{0.9}{$\sum_{i<p}^{}$}\tilde{\phi}_{pj}(\kappa_{ij}(x,\xi;\eta),\iota_{ij}(x,\xi;\eta))\right.$}\\
	\nonumber
	&\scalebox{0.9}{$ \left. \times ((\lambda_i - \lambda_p)K_{ip}(\kappa_{ij}(x,\xi;\eta),\iota_{ij}(x,\xi;\eta))\right.$}\\
	\nonumber
	&\scalebox{0.9}{$ \left. + \sigma_{ip}^{--}(\kappa_{ij}(x,\xi;\eta))) - \varepsilon_{ij}^{++}(\kappa_{ij}(x,\xi;\eta))\right.$}\\
	\nonumber
	&\scalebox{0.9}{$ \left. + \smallint\nolimits_{0}^{\kappa_{ij}} \scalebox{0.9}{$\sum_{p=1}^{3}$} (\varepsilon_{pj}^{++}(\kappa_{ij}(x,\xi;\eta))K_{ip}(\kappa_{ij}(x,\xi;\eta),\tau)\right.$}\\
	&\scalebox{0.9}{$ \left. +  \varepsilon_{pj}^{+-}(\kappa_{ij}(x,\xi;\eta))L_{ip}(\kappa_{ij}(x,\xi;\eta),\tau)) d\tau \right] d\eta$}.
\end{align}
The expressions of $\kappa_{ij}(x,\xi,s), \iota_{ij}(x,\xi,s)$ and $\eta_{ij}^{F}(x,\xi)$, which are straightforward to obtain, are omitted here for simplicity. We can also notice that the fifth line of \eqref{eq_phiij} uses the expression of $\tilde{\phi}_{pj}$ for $i <p$. This term is known and bounded (induction assumption).

\subsubsection{Characteristics of the $K$ kernels}

For each $1 \leq i,j \leq 3$, and $(x,\xi) \in \Gamma$, we define the following characteristic lines $(\chi_{ij}(x,\xi;\cdot), \zeta_{ij}(x,\xi;\cdot))$ corresponding to \eqref{eq_K}:
\begin{align}
	\label{eq_kijx} &\scalebox{0.9}{$\left\{
	\begin{array}{l}
		\frac{d \chi_{ij}}{d\nu}(x,\xi;\nu) = -\lambda_i, \quad \nu \in [0, \nu_{ij}^{F}(x,\xi)]\\
		\chi_{ij}(x,\xi;0) = x,  \chi_{ij}(x,\xi;\nu_{ij}^{F}(x,\xi)) = \chi_{ij}^{F}(x,\xi)
	\end{array}
	\right.$}\\
	\label{eq_kijxi} &\scalebox{0.9}{$\left\{
	\begin{array}{l}
		\frac{d \zeta_{ij}}{d\nu}(x,\xi;\nu) = -\lambda_j, \quad \nu \in [0, \nu_{ij}^{F}(x,\xi)]\\
		\zeta_{ij}(x,\xi;0) = \xi,  \zeta_{ij}(x,\xi;\nu_{ij}^{F}(x,\xi)) = \zeta_{ij}^{F}(x,\xi)
	\end{array}
	\right.$}
\end{align}
These lines all originate from $(x,\xi)$ and terminate either at  the point $(\chi_{ij}^{F}(x,\xi),\chi_{ij}^{F}(x,\xi))$ or at the point $(\chi_{ij}^{F}(x,\xi),0)$. They are three distinct cases $i<j$, $i=j$ and $i>j$. The detailed expressions of $\chi_{ij}(x,\xi,\nu), \zeta_{ij}(x,\xi,\nu), \chi_{ij}^{F}(x,\xi), \zeta_{ij}^{F}(x,\xi)$ and $\nu_{ij}^{F}(x,\xi)$ are, again, omitted here because of space constrains. Integrating \eqref{eq_K} along these characteristic lines, plugging in the boundary conditions \eqref{eq_kbcx}, \eqref{eq_kbc0} and \eqref{eq_lij}, \eqref{eq_phiij} evaluated at $(\chi_{ij}^{F}(x,\xi),0)$ yields
\begin{align}\label{eq_kij}
	\nonumber
	&\scalebox{0.9}{$K_{ij}(x,\xi) = -\delta_{ij}\tfrac{\sigma_{ij}^{--}(\chi_{ij}^{F})}{\lambda_i-\lambda_j} + (1-\delta_{ij})\tfrac{1}{\lambda_j}\scalebox{0.9}{$\sum_{k=1}^{3}$} \lambda_k c_{kj} l_{ik}(x_{ij}^{F})$}\\
	\nonumber
	&\scalebox{0.9}{$ + (1-\delta_{ij})\tfrac{1}{\lambda_j}\scalebox{0.9}{$\sum_{r=1}^{3}$} \lambda_r c_{rj}\smallint\nolimits_{0}^{s_{ir}^{F}(\chi_{ij}^{F}(x,\xi),0)}$}\\
	\nonumber
	&\scalebox{0.9}{$ \times \left[\scalebox{0.9}{$\sum_{k=1}^{3}$}\sigma_{kr}^{++}(\xi_{ir}(\chi_{ij}^{F}(x,\xi),0;s))\right.$}\\
	\nonumber
	&\scalebox{0.9}{$\left. \times L_{ik}(x_{ir}(\chi_{ij}^{F}(x,\xi),0;s),\xi_{ir}(\chi_{ij}^{F}(x,\xi),0;s))\right.$}\\
	\nonumber
	&\scalebox{0.9}{$\left. + \scalebox{0.9}{$\sum_{k=1}^{3}$}\sigma_{kr}^{+-}(\xi_{ir}(\chi_{ij}^{F}(x,\xi),0;s))\right.$}\\
	\nonumber
	&\scalebox{0.9}{$\left. \times K_{ik}(x_{ir}(\chi_{ij}^{F}(x,\xi),0;s),\xi_{ir}(\chi_{ij}^{F}(x,\xi),0;s))\right.$}\\
	\nonumber
	&\scalebox{0.9}{$\left. - \scalebox{0.9}{$\sum_{i<p}^{}$} L_{pr}(\chi_{ij}^{F}(x,\xi),0;s),\xi_{ir}(\chi_{ij}^{F}(x,\xi),0;s))\right.$}\\
	\nonumber
	&\scalebox{0.9}{$\left. \times \left((\lambda_i-\lambda_p)K_{ip}(x_{ir}(\chi_{ij}^{F}(x,\xi),0;s),x_{ir}(\chi_{ij}^{F}(x,\xi),0;s))\right.\right.$}\\
	\nonumber
	&\scalebox{0.9}{$\left.\left. + \sigma_{ip}^{--}(x_{ir}(\chi_{ij}^{F}(x,\xi),0;s))\right)\right.$}\\
	\nonumber
	&\scalebox{0.9}{$\left. - f_{ir}^{+-}(x_{ir}(\chi_{ij}^{F}(x,\xi),0;s),\xi_{ir}(\chi_{ij}^{F}(x,\xi),0;s))\right.$}\\
	\nonumber
	&\scalebox{0.9}{$\left. + \smallint\nolimits_{\xi_{ir}}^{x_{ir}} \scalebox{0.9}{$\sum_{p=1}^{3}$} f_{pr}^{+-}(\tau,\xi_{ir}(\chi_{ij}^{F}(x,\xi),0;s))\right.$}\\
	\nonumber
	&\scalebox{0.9}{$\left. \times K_{ip}(x_{ir}(\chi_{ij}^{F}(x,\xi),0;s),\tau)d\tau\right.$}\\
	\nonumber
	&\scalebox{0.9}{$\left. + \smallint\nolimits_{\xi_{ir}}^{x_{ir}} \scalebox{0.9}{$\sum_{k=1}^{3}$} f_{kr}^{--}(\tau,\xi_{ir}(\chi_{ij}^{F}(x,\xi),0;s))\right.$}\\
	\nonumber
	&\scalebox{0.9}{$\left. \times L_{ik}(x_{ir}(\chi_{ij}^{F}(x,\xi),0;s),\tau)d\tau \right] ds$}\\
	\nonumber
	&\scalebox{0.9}{$+ (1-\delta_{ij})\phi_{ij}(0) + (1-\delta_{ij})\smallint\nolimits_{0}^{\eta_{ij}^{F}(\chi_{ij}^{F}(x,\xi),0)}$} \\
	\nonumber
	&\scalebox{0.9}{$ \times \left[\scalebox{0.9}{$\sum_{k=1}^{3}$} a_{kj}\tilde{\phi}_{ik}(\kappa_{ij}(\chi_{ij}^{F}(x,\xi),0;\eta),\iota_{ij}(\chi_{ij}^{F}(x,\xi),0;\eta))\right.$}\\
	\nonumber
	&\scalebox{0.9}{$\left. + \scalebox{0.9}{$\sum_{p=1}^{3}$} \lambda_p d_{pj} L_{ip}(\kappa_{ij}(\chi_{ij}^{F}(x,\xi),0;\eta),0)\right.$}\\
	\nonumber
	&\scalebox{0.9}{$\left. - \scalebox{0.9}{$\sum_{i<p}^{}$}\tilde{\phi}_{pj}(\kappa_{ij}(\chi_{ij}^{F}(x,\xi),0;\eta),\iota_{ij}(\chi_{ij}^{F}(x,\xi),0;\eta))\right.$}\\
	\nonumber
	&\scalebox{0.9}{$\left. \times ((\lambda_i - \lambda_p)K_{ip}(\kappa_{ij}(\chi_{ij}^{F}(x,\xi),0;\eta),\iota_{ij}(\chi_{ij}^{F}(x,\xi),0;\eta))\right.$}\\
	\nonumber
	&\scalebox{0.9}{$\left. + \sigma_{ip}^{--}(\kappa_{ij}(\chi_{ij}^{F}(x,\xi),0;\eta))) - \varepsilon_{ij}^{++}(\kappa_{ij}(\chi_{ij}^{F}(x,\xi),0;\eta))\right.$}\\
	\nonumber
	&\scalebox{0.9}{$\left.  + \smallint\nolimits_{0}^{\kappa_{ij}} \scalebox{0.9}{$\sum_{p=1}^{3}$} \varepsilon_{pj}^{++}(\kappa_{ij}(\chi_{ij}^{F}(x,\xi),0;\eta))\right.$}\\
	\nonumber
	&\scalebox{0.9}{$\left.\times K_{ip}(\kappa_{ij}(\chi_{ij}^{F}(x,\xi),0;\eta),\tau) d\tau\right.$}\\
	\nonumber
	&\scalebox{0.9}{$\left. + \smallint\nolimits_{0}^{\kappa_{ij}} \scalebox{0.9}{$\sum_{k=1}^{3}$} \varepsilon_{kj}^{+-}(\kappa_{ij}(\chi_{ij}^{F}(x,\xi),0;\eta))\right.$}\\
	\nonumber
	&\scalebox{0.9}{$\left. \times L_{ik}(\kappa_{ij}(\chi_{ij}^{F}(x,\xi),0;\eta),\tau) d\tau \right] d\eta$}\\
	\nonumber
	&\scalebox{0.9}{$ + \smallint\nolimits_{0}^{\nu_{ij}^{F}(x,\xi)}\left[\scalebox{0.9}{$\sum_{k=1}^{3}$} \sigma_{kj}^{--}(\zeta_{ij}(x,\xi;\nu))\right.$}\\
	\nonumber
	&\scalebox{0.9}{$\left.\times K_{ik}(\chi_{ij}(x,\xi;\nu),\zeta_{ij}(x,\xi;\nu))\right.$}\\
	\nonumber
	&\scalebox{0.9}{$\left. + \scalebox{0.9}{$\sum_{p=1}^{3}$} \sigma_{pj}^{+-}(\zeta_{ij}(x,\xi;\nu))L_{ip}(\chi_{ij}(x,\xi;\nu),\zeta_{ij}(x,\xi;\nu))\right.$}\\
	\nonumber
	&\scalebox{0.9}{$\left. - \scalebox{0.9}{$\sum_{i<p}^{}$} K_{pj}(\chi_{ij}(x,\xi;\nu),\zeta_{ij}(x,\xi;\nu))\right.$}\\
	\nonumber
	&\scalebox{0.9}{$\left. \times ((\lambda_i - \lambda_p)K_{ip}(\chi_{ij}(x,\xi;\nu),\chi_{ij}(x,\xi;\nu))\right.$}\\
	\nonumber
	&\scalebox{0.9}{$\left. + \sigma_{ip}^{--}(\chi_{ij}(x,\xi;\nu))) - f_{ij}^{++}(\chi_{ij}(x,\xi;\nu),\zeta_{ij}(x,\xi;\nu))\right.$}\\
	\nonumber
	&\scalebox{0.9}{$\left. + \smallint\nolimits_{\zeta_{ij}}^{\chi_{ij}} \scalebox{0.9}{$\sum_{p=1}^{3}$} (f_{pj}^{++}(\tau,\zeta_{ij}(x,\xi;\nu)) K_{ip}(\chi_{ij}(x,\xi;\nu),\tau)\right.$}\\
	&\scalebox{0.9}{$\left. + f_{pj}^{-+}(\tau,\zeta_{ij}(x,\xi;\nu)) L_{ip}(\chi_{ij}(x,\xi;\nu),\tau)) d\tau\right] d\nu$},
\end{align}
where the coefficient $\delta_{ij}(x,\xi)$ is defined by
\begin{align}
	\delta_{ij}(x,\xi) = \left\{
		\begin{array}{l}
			1 \quad if\text{ } j<i \text{ } and \text{ } \lambda_i\xi - \lambda_j x \geq 0\\
			0 \quad else.
		\end{array}
		\right.	
\end{align}
This coefficient reflects the fact that, as mentioned above, some characteristics terminate on the hypotenuse and others on the axis $\xi=0$. Notice that \eqref{eq_kij} uses the expression of $K_{pj}$ for $i <p$. These terms are also known and bounded (induction assumption).

\subsection{Method of successive approximations}
We now use the method of successive approximations to solve \eqref{eq_lij}, \eqref{eq_phiij} and \eqref{eq_kij}. For $1 \leq i,j \leq 3$, define
\begin{align}
	\nonumber
	&\scalebox{0.9}{$\varphi_{ij}^{1}(x,\xi) = l_{ij}(x_{ij}^{F}) + \smallint\nolimits_{0}^{s_{ij}^{F}}\left[-\scalebox{0.9}{$\sum_{i<p}^{}$} L_{pj}(x_{ij}(x,\xi;s),\xi_{ij}(x,\xi;s))\right.$}\\
	&\scalebox{0.9}{$\left.\quad\times\sigma_{ip}^{--}(x_{ij}(x,\xi;s)) - f_{ij}^{+-}(x,\xi)\right] ds$},\\
	\nonumber
	&\scalebox{0.9}{$\varphi_{ij}^{2}(x,\xi) = -\delta_{ij}\tfrac{\sigma_{ij}^{--}(\chi_{ij}^{F})}{\lambda_i-\lambda_j} + (1-\delta_{ij})\tfrac{1}{\lambda_j}\scalebox{0.9}{$\sum_{k=1}^{3}$} \lambda_k c_{kj} l_{ik}(x_{ij}^{F})$}\\
	\nonumber
	&\scalebox{0.9}{$\quad + (1-\delta_{ij})\tfrac{1}{\lambda_j}\scalebox{0.9}{$\sum_{r=1}^{3}$} \lambda_r c_{rj}\smallint\nolimits_{0}^{s_{ir}^{F}(\chi_{ij}^{F}(x,\xi),0)}$}\\
	\nonumber
	&\scalebox{0.9}{$\quad \times \left[-\scalebox{0.9}{$\sum_{i<p}^{}$} L_{pr}(\chi_{ij}^{F}(x,\xi),0;s),\xi_{ir}(\chi_{ij}^{F}(x,\xi),0;s))\right.$}\\
	\nonumber
	&\scalebox{0.9}{$\left.\quad\times \sigma_{ip}^{--}(x_{ir}(\chi_{ij}^{F}(x,\xi),0;s))\right.$}\\
	\nonumber
	&\scalebox{0.9}{$\left.\quad - f_{ir}^{+-}(x_{ir}(\chi_{ij}^{F}(x,\xi),0;s),\xi_{ir}(\chi_{ij}^{F}(x,\xi),0;s))\right] ds$}\\
	\nonumber
	&\scalebox{0.9}{$\quad + (1-\delta_{ij})\phi_{ij}(0) + (1-\delta_{ij})\smallint\nolimits_{0}^{\eta_{ij}^{F}(\chi_{ij}^{F}(x,\xi),0)}$}\\
	\nonumber
	&\scalebox{0.9}{$\quad \times [ - \scalebox{0.9}{$\sum_{i<p}^{}$} \tilde{\phi}_{pj}(\kappa_{ij}(\chi_{ij}^{F}(x,\xi),0;\eta),\iota_{ij}(\chi_{ij}^{F}(x,\xi),0;\eta))$}\\
	\nonumber
	&\scalebox{0.9}{$\quad \times \sigma_{ip}^{--}(\kappa_{ij}(\chi_{ij}^{F}(x,\xi),0;\eta))$}\\
	\nonumber
	&\scalebox{0.9}{$\quad - \varepsilon_{ij}^{++}(\kappa_{ij}(\chi_{ij}^{F}(x,\xi),0;\eta))] d\eta + \smallint\nolimits_{0}^{\nu_{ij}^{F}(x,\xi)}\left[\right.$}\\
	\nonumber
	&\scalebox{0.9}{$\left.\quad  - \scalebox{0.8}{$\sum_{i<p}^{}$} K_{pj}(\chi_{ij}(x,\xi;\nu),\zeta_{ij}(x,\xi;\nu))\sigma_{ip}^{--}(\chi_{ij}(x,\xi;\nu))\right.$}\\
	&\scalebox{0.9}{$\left.\quad  - f_{ij}^{++}(\chi_{ij}(x,\xi;\nu),\zeta_{ij}(x,\xi;\nu))\right] d\nu$},\\
	\nonumber
	&\scalebox{0.9}{$\varphi_{ij}^{3}(x,\xi) = I_{\{\xi=0\}}(\xi,\xi)\phi_{ij}(0)$}\\
	\nonumber
	&\scalebox{0.9}{$\quad - \smallint\nolimits_{0}^{\eta_{ij}^{F}(x,\xi)}  I_{\{\xi=0\}}(\kappa_{ij}(x,\xi;\eta),\iota_{ij}(x,\xi;\eta)) \left[\right.$}\\
	\nonumber
	&\scalebox{0.9}{$\left.\quad \scalebox{0.9}{$\sum_{i<p}^{}$}\tilde{\phi}_{pj}(\kappa_{ij}(x,\xi;\eta),\iota_{ij}(x,\xi;\eta)) \sigma_{ip}^{--}(\kappa_{ij}(x,\xi;\eta))\right.$}\\
	&\scalebox{0.9}{$\left.\quad + \varepsilon_{ij}^{++}(\kappa_{ij}(x,\xi;\eta)) \right] d\eta$}.
\end{align}
Besides, we define $\vec{H}$ as the vector containing all the kernels as follows:
\begin{align}
    \vec{H} &= [
	L_{11} \cdots L_{33}\quad K_{11} \cdots K_{33}\quad \Phi_{11} \cdots \Phi_{33}
	]^T,\\ 
    \vec{\varphi} &= [
	\varphi_{11}^{1} \cdots \varphi_{33}^{1}\quad \varphi_{11}^{2} \cdots \varphi_{33}^{2}\quad \varphi_{11}^{3} \cdots \varphi_{33}^{3}
	]^T.
\end{align}
 For $\forall 1 \leq i,j \leq 3$, consider the following linear operators acting on $\vec{H}$:
\begin{align}
	\label{eq_Psi1} \nonumber
	&\scalebox{0.9}{$\Psi_{ij}^{1}[\vec{H}](x,\xi) = \smallint\nolimits_{0}^{s_{ij}^{F}(x,\xi)} [\scalebox{0.9}{$\sum_{k=1}^{3}$} \sigma_{kj}^{++}(\xi_{ij}(x,\xi;s))$}\\
	\nonumber
	&\scalebox{0.9}{$\times L_{ik}(x_{ij}(x,\xi;s),\xi_{ij}(x,\xi;s))$} \\
	\nonumber
	&\scalebox{0.9}{$ + \scalebox{0.9}{$\sum_{p=1}^{3}$} \sigma_{pj}^{-+}(\xi_{ij}(x,\xi;s))K_{ip}(x_{ij}(x,\xi;s),\xi_{ij}(x,\xi;s))$}\\
	\nonumber
	&\scalebox{0.9}{$ - \scalebox{0.9}{$\sum_{i<p}^{}$} L_{pj}(x_{ij}(x,\xi;s),\xi_{ij}(x,\xi;s))$}\\
	\nonumber
	&\scalebox{0.9}{$ \times (\lambda_i - \lambda_p)K_{ip}(x_{ij}(x,\xi;s),\xi_{ij}(x,\xi;s))$}\\
	\nonumber
	&\scalebox{0.9}{$ + \smallint\nolimits_{\xi_{ij}}^{x_{ij}} \scalebox{0.9}{$\sum_{p=1}^{3}$} (f_{pj}^{+-}(\tau,\xi_{ij}(x,\xi;s))K_{ip}(x_{ij}(x,\xi;s),\tau)$}\\
	&\scalebox{0.9}{$ + f_{pj}^{--}(\tau,\xi_{ij}(x,\xi;s))L_{ip}(x_{ij}(x,\xi;s),\tau)) d\tau ] ds$},\\
	\label{eq_Psi2} \nonumber
	&\scalebox{0.9}{$\Psi_{ij}^{2}[\vec{H}](x,\xi) = (1-\delta_{ij})\tfrac{1}{\lambda_j}\scalebox{0.9}{$\sum_{r=1}^{3}$} \lambda_r c_{rj}\smallint\nolimits_{0}^{s_{ir}^{F}(\chi_{ij}^{F}(x,\xi),0)}$}\\
	\nonumber
	&\scalebox{0.9}{$ \times \left[\scalebox{0.9}{$\sum_{k=1}^{3}$}\sigma_{kr}^{++}(\xi_{ir}(\chi_{ij}^{F}(x,\xi),0;s))\right.$}\\
	\nonumber
	&\scalebox{0.9}{$\left. \times L_{ik}(x_{ir}(\chi_{ij}^{F}(x,\xi),0;s),\xi_{ir}(\chi_{ij}^{F}(x,\xi),0;s))\right.$}\\
	\nonumber
	&\scalebox{0.9}{$\left. + \scalebox{0.9}{$\sum_{k=1}^{3}$}\sigma_{kr}^{+-}(\xi_{ir}(\chi_{ij}^{F}(x,\xi),0;s))\right.$}\\
	\nonumber
	&\scalebox{0.9}{$\left. \times K_{ik}(x_{ir}(\chi_{ij}^{F}(x,\xi),0;s),\xi_{ir}(\chi_{ij}^{F}(x,\xi),0;s))\right.$}\\
	\nonumber
	&\scalebox{0.9}{$\left. - \scalebox{0.9}{$\sum_{i<p}^{}$} L_{pr}(\chi_{ij}^{F}(x,\xi),0;s),\xi_{ir}(\chi_{ij}^{F}(x,\xi),0;s))\right.$}\\
	\nonumber
	&\scalebox{0.9}{$\left. \times (\lambda_i-\lambda_p)K_{ip}(x_{ir}(\chi_{ij}^{F}(x,\xi),0;s),x_{ir}(\chi_{ij}^{F}(x,\xi),0;s))\right.$}\\
	\nonumber
	&\scalebox{0.9}{$\left. + \smallint\nolimits_{\xi_{ir}}^{x_{ir}} \scalebox{0.9}{$\sum_{p=1}^{3}$} f_{pr}^{+-}(\tau,\xi_{ir}(\chi_{ij}^{F}(x,\xi),0;s))\right.$}\\
	\nonumber
	&\scalebox{0.9}{$\left. \times K_{ip}(x_{ir}(\chi_{ij}^{F}(x,\xi),0;s),\tau)d\tau\right.$}\\
	\nonumber
	&\scalebox{0.9}{$\left. + \smallint\nolimits_{\xi_{ir}}^{x_{ir}} \scalebox{0.9}{$\sum_{k=1}^{3}$} f_{kr}^{--}(\tau,\xi_{ir}(\chi_{ij}^{F}(x,\xi),0;s))\right.$}\\
	\nonumber
	&\scalebox{0.9}{$\left. \times L_{ik}(x_{ir}(\chi_{ij}^{F}(x,\xi),0;s),\tau)d\tau \right] ds$}\\
	\nonumber
	&\scalebox{0.9}{$ + (1-\delta_{ij})\smallint\nolimits_{0}^{\eta_{ij}^{F}(\chi_{ij}^{F}(x,\xi),0)}$}\\
	\nonumber
	&\scalebox{0.9}{$\times \left[\scalebox{0.9}{$\sum_{k=1}^{3}$} a_{kj}\tilde{\phi}_{ik}(\kappa_{ij}(\chi_{ij}^{F}(x,\xi),0;\eta),\iota_{ij}(\chi_{ij}^{F}(x,\xi),0;\eta))\right.$}\\
	\nonumber
	&\scalebox{0.9}{$\left. + \scalebox{0.9}{$\sum_{p=1}^{3}$} \lambda_p d_{pj} L_{ip}(\kappa_{ij}(\chi_{ij}^{F}(x,\xi),0;\eta),0)\right.$}\\
	\nonumber
	&\scalebox{0.9}{$\left. - \scalebox{0.9}{$\sum_{i<p}^{}$} \tilde{\phi}_{pj}(\kappa_{ij}(\chi_{ij}^{F}(x,\xi),0;\eta),\iota_{ij}(\chi_{ij}^{F}(x,\xi),0;\eta))\right.$}\\
	\nonumber
	&\scalebox{0.9}{$\left. \times (\lambda_i - \lambda_p)K_{ip}(\kappa_{ij}(\chi_{ij}^{F}(x,\xi),0;\eta),\iota_{ij}(\chi_{ij}^{F}(x,\xi),0;\eta))\right.$}\\
	\nonumber
	&\scalebox{0.9}{$\left.  + \smallint\nolimits_{0}^{\kappa_{ij}} \scalebox{0.9}{$\sum_{p=1}^{3}$} \varepsilon_{pj}^{++}(\kappa_{ij}(\chi_{ij}^{F}(x,\xi),0;\eta))\right.$}\\
	\nonumber
	&\scalebox{0.9}{$\left. \times K_{ip}(\kappa_{ij}(\chi_{ij}^{F}(x,\xi),0;\eta),\tau) d\tau\right.$}\\
	\nonumber
	&\scalebox{0.9}{$\left. + \smallint\nolimits_{0}^{\kappa_{ij}} \scalebox{0.9}{$\sum_{k=1}^{3}$} \varepsilon_{kj}^{+-}(\kappa_{ij}(\chi_{ij}^{F}(x,\xi),0;\eta))\right.$}\\
	\nonumber
	&\scalebox{0.9}{$\left. \times L_{ik}(\kappa_{ij}(\chi_{ij}^{F}(x,\xi),0;\eta),\tau) d\tau \right] d\eta + \smallint\nolimits_{0}^{\nu_{ij}^{F}(x,\xi)}$}\\
	\nonumber
	&\scalebox{0.9}{$[\scalebox{0.9}{$\sum_{k=1}^{3}$} \sigma_{kj}^{--}(\zeta_{ij}(x,\xi;\nu)) K_{ik}(\chi_{ij}(x,\xi;\nu),\zeta_{ij}(x,\xi;\nu))$}\\
	\nonumber
	&\scalebox{0.9}{$ + \scalebox{0.9}{$\sum_{p=1}^{3}$} \sigma_{pj}^{+-}(\zeta_{ij}(x,\xi;\nu))L_{ip}(\chi_{ij}(x,\xi;\nu),\zeta_{ij}(x,\xi;\nu))$}\\
	\nonumber
	&\scalebox{0.9}{$ - \scalebox{0.9}{$\sum_{i<p}^{}$} K_{pj}(\chi_{ij}(x,\xi;\nu),\zeta_{ij}(x,\xi;\nu))$}\\
	\nonumber
	&\scalebox{0.9}{$\times (\lambda_i - \lambda_p)K_{ip}(\chi_{ij}(x,\xi;\nu),\chi_{ij}(x,\xi;\nu))$}\\
	\nonumber
	&\scalebox{0.9}{$ + \smallint\nolimits_{\zeta_{ij}}^{\chi_{ij}} (\scalebox{0.9}{$\sum_{p=1}^{3}$} f_{pj}^{++}(\tau,\zeta_{ij}(x,\xi;\nu)) K_{ip}(\chi_{ij}(x,\xi;\nu),\tau) $}\\
	&\scalebox{0.9}{$+ f_{pj}^{-+}(\tau,\zeta_{ij}(x,\xi;\nu)) L_{ip}(\chi_{ij}(x,\xi;\nu),\tau)) d\tau] d\nu$} ,\\
	\label{eq_Psi3} \nonumber
	&\scalebox{0.9}{$\Psi_{ij}^{3}[\vec{H}](x,\xi) = \smallint\nolimits_{0}^{\eta_{ij}^{F}(x,\xi)}I_{\{\xi=0\}}(\kappa_{ij}(x,\xi;\eta),\iota_{ij}(x,\xi;\eta))$}\\
	\nonumber
	&\scalebox{0.9}{$ [\scalebox{0.9}{$\sum_{k=1}^{3}$} a_{kj}\tilde{\phi}_{ik}(\kappa_{ij}(x,\xi;\eta),\iota_{ij}(x,\xi;\eta))$}\\
	\nonumber
	&\scalebox{0.9}{$ + \scalebox{0.9}{$\sum_{p=1}^{3}$} \lambda_p d_{pj} L_{ip}(\kappa_{ij}(x,\xi;\eta),0)$}\\
	\nonumber
	&\scalebox{0.9}{$- \scalebox{0.9}{$\sum_{i<p}^{}$} \tilde{\phi}_{pj}(\kappa_{ij}(x,\xi;\eta),\iota_{ij}(x,\xi;\eta))$} \\
	\nonumber
	&\scalebox{0.9}{$ \times (\lambda_i - \lambda_p)K_{ip}(\kappa_{ij}(x,\xi;\eta),\iota_{ij}(x,\xi;\eta))$}\\
	\nonumber
	&\scalebox{0.9}{$ + \smallint\nolimits_{0}^{\kappa_{ij}} \scalebox{0.9}{$\sum_{p=1}^{3}$} (\varepsilon_{pj}^{++}(\kappa_{ij}(x,\xi;\eta))K_{ip}(\kappa_{ij}(x,\xi;\eta),\tau)$}\\
	&\scalebox{0.9}{$+ \varepsilon_{pj}^{+-}(\kappa_{ij}(x,\xi;\eta))L_{ip}(\kappa_{ij}(x,\xi;\eta),\tau)) d\tau ] d\eta$} .
\end{align}
We set 
\begin{align}
    \scalebox{0.95}{$\vec{\Psi}[\vec{H}](x,\xi) = [\Psi_{11}^{1} \cdots \Psi_{33}^{1}\text{ } \Psi_{11}^{2} \cdots \Psi_{33}^{2}\text{ } \Psi_{11}^{3} \cdots \Psi_{33}^{3}]^\top$},
\end{align}
and define the following sequence:
\begin{align}
	&\vec{H}^{0}(x,\xi) = 0,\quad \vec{H}^{q}(x,\xi) = \vec{\varphi} + \vec{\Psi}[\vec{H}^{q-1}](x,\xi).\label{eq_Hq}
\end{align}
One should notice that if the limit exists, then $\vec{H} = \lim_{q\to +\infty} \vec{H}^q(x,\xi)$ is a solution of the integral equations, and thus solves the original hyperbolic system. Besides, define for $q \geq 1$ the increment $\Delta\vec{H}^q = \vec{H}^q - \vec{H}^{q-1}$, with $\Delta\vec{H}^0 = \vec{\varphi}$ by definition. Since the functional $\vec{\Psi}$ is linear, the following equation $\Delta\vec{H}^q(x,\xi) = \vec{\Psi}[\vec{H}^{q-1}](x,\xi)$ holds. Using the definition of $\Delta\vec{H}^q$, it follows that if the sum $\sum_{q=0}^{+\infty}\Delta\vec{H}^q(x,\xi)$ is finite, then one can obtain
\begin{align}\label{eq_Hconverge}
	\vec{H}(x,\xi) = \scalebox{0.9}{$\sum_{q=0}^{+\infty}$}\Delta\vec{H}^q(x,\xi).
\end{align}

\subsection{Convergence of the successive approximation series}
To prove convergence of the series, we look for a recursive upper bound, similarly to \cite{n+1},\cite{Heterodirectional},\cite{minimum}. First, we define
\begin{align*}
	&\scalebox{0.95}{$\overline{\lambda} = \lambda_3, \underline{\lambda} = \tfrac{1}{\lambda_1}, \bar{\mu} = \mathop{\max}\limits_{i,j} \{|\lambda_i - \lambda_p|\}$},\\
	&\scalebox{0.95}{$\bar{a} = \mathop{\max}\limits_{i,j}\{a_{ij}\}, \bar{d} = \mathop{\max}\limits_{i,j}\{d_{ij}\}, \bar{c} = \mathop{\max}\limits_{i,j}\{c_{ij}\}$},\\
	&\scalebox{0.95}{$\bar{f} = \mathop{\max}\limits_{i,j}\mathop{\max}\limits_{(x,\xi) \in \Gamma}\{f_{ij}^{++},f_{ij}^{+-},f_{ij}^{-+}, f_{ij}^{--}\}$},\\
	&\scalebox{0.95}{$\bar{\sigma} = \mathop{\max}\limits_{i,j}\mathop{\max}\limits_{x \in [0,1]}\{\sigma_{ij}^{++},\sigma_{ij}^{+-},\sigma_{ij}^{-+},\sigma_{ij}^{--}\}$},\\ 
	&\scalebox{0.95}{$\bar{\varepsilon} = \mathop{\max}\limits_{i,j}\mathop{\max}\limits_{x \in [0,1]}\{\varepsilon_{ij}^{++}(x),\varepsilon_{ij}^{+-}(x)\}$},~\scalebox{0.95}{$M_\lambda = \mathop{\max}\limits_{i=1,2,3}\{\frac{1}{\lambda_i}\} = \frac{1}{\lambda_1}$},
\end{align*}
and \scalebox{0.9}{$\bar{S} = \mathop{\max}\limits_{p>i,1\leq j \leq n}\{\Vert K_{pj} \Vert, \Vert L_{pj} \Vert, \Vert \tilde{\phi}_{pj} \Vert \}$}, which is well-defined according to the hypothesis $P(s-1)$. Moreover we set
$M = \bar{m}\bar{\varphi}M_\lambda$,
where $\bar{m} = 3\bar{\lambda}\underline{\lambda}\bar{c}(6\bar{\sigma}+3\bar{\mu}\bar{S}+3\bar{f})+3\bar{a}+3\bar{\lambda}\bar{d}+6\bar{\mu}\bar{S}+6\bar{\varepsilon}+6\bar{\sigma}+6\bar{f}$.
\begin{my_claim}\label{clm:1}
	For $q \in \mathbb{N}$, $(x,\xi) \in \Gamma$ and $s_{ij}^{F}$, $\eta_{ij}^{F}$, $\nu_{ij}^{F}$, $x_{ij}(x,\xi;\cdot)$, $\xi_{ij}(x,\xi;\cdot)$, $\kappa_{ij}(x,\xi;\cdot)$, $\iota_{ij}(x,\xi;\cdot)$, $\chi_{ij}(x,\xi;\cdot)$, $\zeta_{ij}(x,\xi;\cdot)$ defined as in \eqref{eq_lijx}, \eqref{eq_lijxi}, \eqref{eq_phiijx}, \eqref{eq_phiijxi}, \eqref{eq_kijx}, \eqref{eq_kijxi}, for $\forall 1 \leq i,j \leq 3$, the following inequalities holds:
	\begin{align}
			& \scalebox{0.95}{$\int_{0}^{s_{ij}^{F}(x,\xi)} x_{ij}^{q}(x,\xi;s) ds \leq M_\lambda \frac{x^{q+1}}{q+1}$},\label{x_ijq}\\
			& \scalebox{0.95}{$\int_{0}^{\eta_{ij}^{F}(x,\xi)} \kappa_{ij}^{q}(x,\xi;\eta) ds \leq M_\lambda \frac{x^{q+1}}{q+1}$},\label{kappa_ijq}\\
			& \scalebox{0.95}{$\int_{0}^{\nu_{ij}^{F}(x,\xi)} \chi_{ij}^{q}(x,\xi;\nu) ds \leq M_\lambda \frac{x^{q+1}}{q+1}$},\label{chi_ijq}
	\end{align}
	where $M_\lambda = \mathop{max}\limits_{i=1,2,3}\Big\{\frac{1}{\lambda_i}\Big\} = \frac{1}{\lambda_1}$.
\end{my_claim}
\begin{proof}\normalfont
	We first proof \eqref{x_ijq}. Consider the following change of integration variable, noting \eqref{eq_lijx} and \eqref{eq_lijxi}:
	\begin{align}
		\tau &= x_{ij}(x,\xi;s),\\
		d\tau &= \frac{d}{ds}x_{ij}(x,\xi;s) ds = -\lambda_i ds.
	\end{align}
	Thus, the left-hand-side of \eqref{x_ijq} rewrites as:
	\begin{align}
		\nonumber
		&\smallint\nolimits_{0}^{s_{ij}^{F}(x,\xi)} x_{ij}^{q}(x,\xi;s) ds= \smallint\nolimits_{x}^{x_{ij}^{F}(x,\xi)} \frac{-\tau^q}{\lambda_i} d\tau\\
		&  = \frac{x^{q+1} - (x_{ij}^{F})^{q+1}}{\lambda_i(q+1)}
		\leq M_\lambda\frac{x^{q+1}}{q+1}.
	\end{align}
	Inequality \eqref{x_ijq} is obtained. Inequalities \eqref{kappa_ijq} and \eqref{chi_ijq} are proved the same way using change of integration variables $\tau = \kappa_{ij}(x,\xi;\eta)$ and $\tau = \chi_{ij}(x,\xi;\nu)$.
\end{proof}
\begin{my_claim}\label{clm:2}
    \normalfont
	For $q \geq 1$, assume that, for $\forall (x,\xi)\in \Gamma$, $\forall i=1,\ldots,27$,
	\begin{align}\label{eq_deltaHq}
		|\Delta\vec{H}_i^q(x,\xi)| \leq \bar{\varphi} \frac{M^q x^q}{q!},
	\end{align}
	where $\Delta\vec{H}_i^q(x,\xi)$ denotes the $i-$th $(i = 1,\ldots,27)$ component of $\Delta\vec{H}^q(x,\xi)$, it follows that $\forall (x,\xi) \in \Gamma, \forall i = 1,2,3, \forall j = 1,2,3$:
	\begin{align}
		& |\Psi_{ij}^{1}[\Delta\vec{H}^q(x,\xi)]| \leq \bar{\varphi}\frac{M^{q+1}x^{q+1}}{(q+1)!},\\
		& |\Psi_{ij}^{2}[\Delta\vec{H}^q(x,\xi)]| \leq \bar{\varphi}\frac{M^{q+1}x^{q+1}}{(q+1)!},\\
		& |\Psi_{ij}^{3}[\Delta\vec{H}^q(x,\xi)]| \leq \bar{\varphi}\frac{M^{q+1}x^{q+1}}{(q+1)!},
	\end{align}
	where $M = [3\bar{\lambda}\underline{\lambda}\bar{c}(6\bar{\sigma}+3\bar{\mu}\bar{S}+3\bar{f})+3\bar{a}+3\bar{\lambda}\bar{d}+6\bar{\mu}\bar{S}+6\bar{\varepsilon}+6\bar{\sigma}+6\bar{f}]\bar{\varphi}M_\lambda$.
\end{my_claim}
\begin{proof}\normalfont
	Considering \eqref{eq_deltaHq}, using the expression of $\Psi_{ij}^1[\vec{H}](x,\xi)$ given by \eqref{eq_Psi1}, for all $i=1$,$\ldots,27$ and $(x,\xi) \in \Gamma$, one obtain
	\begin{align}
		\nonumber
		&|\Psi_{ij}^1[\Delta\vec{H}^q](x,\xi)| \leq \smallint\nolimits_{0}^{s_{ij}^{F}(x,\xi)} \Big|\scalebox{0.9}{$\sum_{k=1}^{3}$} \sigma_{kj}^{++}(\xi_{ij}(x,\xi;s))\\
		\nonumber
		& \times \Delta L_{ik}^{q}(x_{ij}(x,\xi;s),\xi_{ij}(x,\xi;s))\\
		\nonumber
		& + \scalebox{0.9}{$\sum_{p=1}^{3}$} \sigma_{pj}^{-+}(\xi_{ij}(x,\xi;s)) \Delta K_{ip}^{q}(x_{ij}(x,\xi;s),\xi_{ij}(x,\xi;s)).\\
		\nonumber
		& - \scalebox{0.9}{$\sum_{i<p}^{}$} L_{pj}(x_{ij}(x,\xi;s),\xi_{ij}(x,\xi;s))(\lambda_i - \lambda_p)\\
		\nonumber
		& \times \Delta K_{ip}^{q}(x_{ij}(x,\xi;s),\xi_{ij}(x,\xi;s))\\
		\nonumber
		& + \smallint\nolimits_{\xi_{ij}}^{x_{ij}} \scalebox{0.9}{$\sum_{p=1}^{3}$} (f_{pj}^{+-}(\tau,\xi_{ij}(x,\xi;s)) \Delta K_{ip}^{q}(x_{ij}(x,\xi;s),\tau)\\
		& + f_{pj}^{--}(\tau,\xi_{ij}(x,\xi;s))\Delta L_{ip}^{q}(x_{ij}(x,\xi;s),\tau)) d\tau \Big| ds,
	\end{align}
	using \eqref{x_ijq} and \eqref{eq_deltaHq}, which yields
	\begin{align}
		\nonumber
		&|\Psi_{ij}^1[\Delta\vec{H}^q](x,\xi)|\\
		\nonumber
		& \leq (6\bar{\sigma} + 3\bar{\mu}\bar{S}) \smallint\nolimits_{0}^{s_{ij}^{F}(x,\xi)} \bar{\varphi} \tfrac{M^q {x_{ij}(x,\xi;s)}^q}{q!} ds\\
		\nonumber
		&\quad + 6\bar{f} \smallint\nolimits_{0}^{s_{ij}^{F}(x,\xi)} \smallint\nolimits_{\xi_{ij}}^{x_{ij}} \bar{\varphi} \tfrac{M^q {x_{ij}(x,\xi;s)}^q}{q!} d\tau ds\\
		\nonumber
		& = (6\bar{\sigma} + 3\bar{\mu}\bar{S} + 6\bar{f}) \smallint\nolimits_{0}^{s_{ij}^{F}(x,\xi)} \bar{\varphi} \tfrac{M^q {x_{ij}(x,\xi;s)}^q}{q!} ds\\
		\nonumber
		&\leq (6\bar{\sigma} + 3\bar{\mu}\bar{S} + 6\bar{f})\tfrac{\bar{\varphi} M_\lambda}{q!} \tfrac{M^q x^{q+1}}{q+1}\\
		& \leq \bar{\varphi} \tfrac{M^{q+1} x^{q+1}}{(q+1)!}.
	\end{align}
	Similarly, for $1 \leq i,j \leq 3, (x,\xi) \in \Gamma$, one gets
	\begin{align}
		\nonumber
		&|\Psi_{ij}^2[\Delta\vec{H}^q](x,\xi)|\\
		\nonumber
		&\leq [3\bar{\lambda}\underline{\lambda}\bar{c}(6\bar{\sigma} + 3\bar{\mu}\bar{S}) + 3(\bar{a} + \bar{\lambda}\bar{d} + \bar{\mu}\bar{S}) + 3\bar{\lambda}\underline{\lambda}\bar{c}\cdot 6\bar{f}\\
		\nonumber
		&\quad + 6\bar{\varepsilon}]\cdot\bar{\varphi}M_\lambda\tfrac{M^q(\chi_{ij}^{F}(x,\xi))^{q+1}}{(q+1)!}\\
		&\quad + (6\bar{\sigma} + 3\bar{\mu}\bar{S} + 6\bar{f})\bar{\varphi}M_\lambda\tfrac{M^q x^{q+1}}{(q+1)!}
	\end{align}
	Using the fact that $\chi_{ij}^{F}(x,\xi) \leq x$, this yields
	\begin{align}
		\nonumber
		&|\Psi_{ij}^2[\Delta\vec{H}^q](x,\xi)|\\
		\nonumber
		&\leq [3\bar{\lambda}\underline{\lambda}\bar{c}(6\bar{\sigma} + 3\bar{\mu}\bar{S}) + 3(\bar{a} + \bar{\lambda}\bar{d} + \bar{\mu}\bar{S}) + 3\bar{\lambda}\underline{\lambda}\bar{c}\cdot 6\bar{f}\\
		\nonumber
		&\quad + 6\bar{\varepsilon} + 6\bar{\sigma} + 3\bar{\mu}\bar{S} + 6\bar{f}] \cdot\bar{\varphi}M_\lambda\tfrac{M^q x^{q+1}}{(q+1)!}\\
		\nonumber
		&= [3\bar{\lambda}\underline{\lambda}\bar{c}(6\bar{\sigma} + 3\bar{\mu}\bar{S} + 6\bar{f}) + 3\bar{a} + 3\bar{\lambda}\bar{d} + 6\bar{\mu}\bar{S} + 6\bar{\varepsilon}\\
		\nonumber
		&\quad + 6\bar{\sigma} + 6\bar{f} ] \cdot\bar{\varphi}M_\lambda\tfrac{M^q x^{q+1}}{(q+1)!}\\
		&\leq \bar{\varphi} \tfrac{M^{q+1} x^{q+1}}{(q+1)!}.
	\end{align}
	In the same way, for $1 \leq i,j \leq 3, (x,\xi) \in \Gamma$, using the fact that $\kappa_{ij}(x,\xi;\eta) \leq x$, one obtain
	\begin{align}
		\nonumber
		&|\Psi_{ij}^3[\Delta\vec{H}^q](x,\xi)|\\
		\nonumber
		&\leq (3\bar{a} + 3\bar{\lambda}\bar{d} + 3\bar{\mu}\bar{S} + 6\bar{\varepsilon}) \bar{\varphi}M_\lambda \tfrac{M^q x^{q+1}}{(q+1)!}\\
		&\leq \bar{\varphi} \tfrac{M^{q+1} x^{q+1}}{(q+1)!},
	\end{align}
	which concludes the proof.
\end{proof}	
Claim \ref{clm:2} directly leads to Theorem \ref{thm:1}, and one has that the following series normally converges on $\Gamma$ and we have the upper bound
\begin{align}
	|\vec{H}(x,\xi)| = \Big|\scalebox{0.9}{$\sum_{q=0}^{q=+\infty}$} \Delta\vec{H}^q(x,\xi)\Big| \leq \bar{\varphi} e^M.
\end{align}

\section{Proof of Lemma \ref{lem:2}: Well-posedness of the Kernel Equations of $M$ and $N$} \label{adx:observer}
\setcounter{equation}{0}
\renewcommand{\theequation}{D.\arabic{equation}}

Developing Eqs.\eqref{eq_obkernel}--\eqref{eq_P}, $for\text{ } 1 \leq i,j \leq 3$, we get the following set of kernel PDEs
\begin{align}
	\nonumber
	&\scalebox{0.95}{$\lambda_i \partial_x M_{ij}(x,\xi) - \lambda_j \partial_\xi M_{ij}(x,\xi) = \scalebox{0.9}{$\sum_{k=1}^{3}$} \sigma_{ik}^{++}(x) M_{kj}(x,\xi)$}\\
	\nonumber
	&\scalebox{0.95}{$ + \scalebox{0.9}{$\sum_{p=1}^{3}$} \sigma_{ip}^{+-}(x) N_{pj}(x,\xi) - \scalebox{0.9}{$\sum_{i<p}^{}$} M_{ip}(x,\xi)\omega_{pj}(\xi)$}\\
	\nonumber
	&\scalebox{0.95}{$ + f_{ij}^{-+}(x,\xi) + \smallint\nolimits_{\xi}^{x} \scalebox{0.9}{$\sum_{p=1}^{3}$} f_{ip}^{-+}(x,s)N_{pj}(s,\xi) ds$}\\
	&\scalebox{0.9}{$ + \smallint\nolimits_{\xi}^{x} \scalebox{0.9}{$\sum_{k=1}^{3}$} f_{ik}^{--}(x,s)M_{kj}(s,\xi) ds$},\\
	\nonumber
	&\scalebox{0.95}{$\lambda_i \partial_x N_{ij}(x,\xi) + \lambda_j \partial_\xi N_{ij}(x,\xi) = -\scalebox{0.9}{$\sum_{k=1}^{3}$} \sigma_{ik}^{--}(x) N_{kj}(x,\xi)$}\\
	\nonumber
	&\scalebox{0.95}{$ - \scalebox{0.9}{$\sum_{p=1}^{3}$} \sigma_{ip}^{-+}(x) M_{pj}(x,\xi) + \scalebox{0.9}{$\sum_{i<p}^{}$} N_{ip}(x,\xi)\omega_{pj}(\xi)$}\\
	\nonumber
	&\scalebox{0.95}{$ - f_{ij}^{++}(x,\xi) - \smallint\nolimits_{\xi}^{x} \scalebox{0.9}{$\sum_{p=1}^{3}$} f_{ip}^{++}(x,s)N_{pj}(s,\xi) ds$}\\
	&\scalebox{0.95}{$ - \smallint\nolimits_{\xi}^{x} \scalebox{0.9}{$\sum_{k=1}^{3}$} f_{ik}^{+-}(x,s)M_{kj}(s,\xi) ds$},
\end{align}
with the following set of boundary conditions for $\forall 1 \leq i,j \leq 3$:
\begin{align}
	&M_{ij}(x,x) = \tfrac{\sigma_{ij}^{+-}}{\lambda_i + \lambda_j},\\
	&N_{ij}(x,x) = -\tfrac{\sigma_{ij}^{--}}{\lambda_i - \lambda_j},\quad (j<i)\\
	&\omega_{ij}(x) = (\lambda_i - \lambda_j)N_{ij}(x,x) + \sigma_{ij}^{--}(x)\quad (j<i),
\end{align}
where the definition of $\lambda_i$, $\sigma_{ij}^{++}$, $\sigma_{ij}^{+-}$, $\sigma_{ij}^{-+}$, $\sigma_{ij}^{--}$, $\omega_{ij}$, $f_{ij}^{++}$, $f_{ij}^{+-}$, $f_{ij}^{-+}$ and  $f_{ij}^{--}$ are shown in Appendix-\ref{adx:wellposed}. Evaluating \eqref{eq_tibacz}, \eqref{eq_tibacy} at $x = 1$ yields
\begin{align}
	\forall 1 \leq i,j \leq 3,\quad N_{ij}(1,x) = r_i M_{ij}(1,x),
\end{align}
where $r_i$ denotes the $(i,i)$-th (diagonal) entry of matrix $R$, and $d_{ij}^{+}, d_{ij}^{-}$ are given by
\begin{align}
	\nonumber
	&d_{ij}^{+}(x,\xi) = -\scalebox{0.9}{$\sum_{k=1}^{3}$} M_{ik}(x,\xi)\sigma_{kj}^{-+}(y)\\
	& - \smallint\nolimits_{\xi}^{x} \scalebox{0.9}{$\sum_{k=1}^{3}$} M_{ik}(x,s) d_{kj}^{-}(s,\xi) ds + f_{ij}^{--}(x,\xi),\\
	\nonumber
	&d_{ij}^{-}(x,\xi) = -\scalebox{0.9}{$\sum_{k=1}^{3}$} N_{ik}(x,\xi)\sigma_{kj}^{-+}(y)\\
	& - \smallint\nolimits_{\xi}^{x} \scalebox{0.9}{$\sum_{k=1}^{3}$} N_{ik}(x,s) d_{kj}^{-}(s,\xi) ds + f_{ij}^{+-}(x,\xi),
\end{align}
provided the $M$ and $N$ kernels are well-defined. Finally, the observer gains are given by
\begin{align}
	p_{ij}^{+}(x) &= \lambda_j m_{ij}(x,0),\quad p_{ij}^{-}(x) = \lambda_j n_{ij}(x,0).
\end{align}
Indeed, considering the following alternate variables
\begin{align}
	&\bar{M}_{ij}(\chi,y) = M_{ij}(1-y,1-\chi) = M_{ij}(x,\xi),\\
	&\bar{N}_{ij}(\chi,y) = N_{ij}(1-y,1-\chi) = N_{ij}(x,\xi),\\
	&\bar{\omega}_{ij}(y) = \omega_{ij}(\xi),\\
	&\bar{f}_{ij}^{++}(\chi,y) = f_{ij}^{++}(1-y,1-\chi) = f_{ij}^{++}(x,\xi),\\
	&\bar{f}_{ij}^{+-}(\chi,y) = f_{ij}^{+-}(1-y,1-\chi) = f_{ij}^{+-}(x,\xi),\\
	&\bar{f}_{ij}^{-+}(\chi,y) = f_{ij}^{-+}(1-y,1-\chi) = f_{ij}^{-+}(x,\xi),\\
	&\bar{f}_{ij}^{--}(\chi,y) = f_{ij}^{--}(1-y,1-\chi) = f_{ij}^{--}(x,\xi),\\
	&\bar{\sigma}_{ij}^{++}(\chi) = \sigma_{ij}^{++}(x),\quad
	\bar{\sigma}_{ij}^{+-}(\chi) = \sigma_{ij}^{+-}(x),\\
	&\bar{\sigma}_{ij}^{-+}(\chi) = \sigma_{ij}^{-+}(x),\quad
	\bar{\sigma}_{ij}^{--}(\chi) = \sigma_{ij}^{--}(x),
\end{align}
one can prove that this system has the same structure as the controller kernel system in Appendix-\ref{adx:wellposed}, but differs in the absence of ODE, which does not affect the overall proof process. Using a similar proof, we can obtain its well-posedness.

\section{Proof of Theorem \ref{thm:2}: Lyapunov-Based Stability Analysis} \label{adx:lyapunov}
\setcounter{equation}{0}
\renewcommand{\theequation}{E.\arabic{equation}}

In this section, we use a Lyapunov function for the stability analysis of the target system for each mode $n$, to show exponential stability of the origin with a tunable convergence rate. For convenience, we will omit the subscript $n$ in the proof. Define
\begin{align}\label{eq_V1}
	\nonumber
	V &= \zeta_1 X^\top X + \zeta_2\smallint\nolimits_{0}^{1} e^{\delta x} \sigma^\top(t,x)\Sigma^{-1}\sigma(t,x) dx\\
	&\quad + \smallint\nolimits_{0}^{1} e^{-\delta x}\psi^\top(t,x)\Sigma^{-1}\psi(t,x) dx.
\end{align}
Differentiating \eqref{eq_V1} with respect to $t$, we get
\begin{align}\label{eq_dotV}
	\nonumber
	\dot{V} &= 2\zeta_1 X^\top X_t + 2\zeta_2\smallint\nolimits_{0}^{1} e^{\delta x}\sigma^\top(t,x)\Sigma^{-1}\sigma_t(t,x) dx\\
	&\quad + 2\smallint\nolimits_{0}^{1} e^{-\delta x}\psi^\top(t,x)\Sigma^{-1}\psi_t(t,x) dx.
\end{align}
Recalling that $E_1 = \Sigma\Phi(0) + A$ and \eqref{eq_ker} with $\delta_1, \delta_2, \delta_3 > 0$, $E_1$ is a diagonal matrix with entries $-\delta_1, -\delta_2$ and $-\delta_3$. Choosing parameter $c = \min\left\{\delta_1, \delta_2,\delta_3\right\}$, we have $X^\top E_1 X \leq -cX^\top X$.
Substituting \eqref{eq_target}--\eqref{eq_bc} into \eqref{eq_dotV}, we have
\begin{align}\label{eq_dotV1}
	\nonumber
	&\dot{V} \leq -2\zeta_1 c X^\top X + 2\zeta_1 X^\top \Sigma \sigma(t,0) - \zeta_2 \sigma^\top(t,0)\sigma(t,0)\\
	\nonumber
	& - \zeta_2\smallint\nolimits_{0}^{1} e^{\delta x} \sigma^\top(t,x)(\delta I - 2\Sigma^{-1}\Omega(x))\sigma(t,x) dx\\
	\nonumber
	& - \smallint\nolimits_{0}^{1} e^{-\delta x} \psi^\top(t,x)(\delta I - 2\Sigma^{-1}(F_{21}(x) - F_{22}(x)))\psi(t,x) dx\\
	\nonumber
	& + 2\smallint\nolimits_{0}^{1} e^{-\delta x}\psi^\top(t,x)\Sigma^{-1}\smallint\nolimits_{0}^{x} \Xi_2(x,y)\sigma(t,y) dydx\\
	\nonumber
	& +2\smallint\nolimits_{0}^{1} e^{-\delta x}\psi^\top(t,x)\Sigma^{-1}\smallint\nolimits_{0}^{x} \Xi_3(x,y)\psi(t,y) dydx\\
	\nonumber
	& + 2\smallint\nolimits_{0}^{1} e^{-\delta x}\psi^\top(t,x)\Sigma^{-1}((F_{21}(x) + F_{22}(x))\sigma(t,x)\\
	& + \Xi_1(x)X) dx + \psi^\top(t,0)\psi(t,0).
\end{align}
Regarding the last term of \eqref{eq_dotV1}, using $\psi(t,0) = E_2X + C\sigma(t,0)$, we have $\psi^\top(t,0)\psi(t,0) = X^\top{E_2}^\top E_2X + 2X^\top{E_2}^\top C\sigma(t,0)+\sigma^\top(t,0)C^\top C\sigma(t,0)$.Then, the first line and last term of \eqref{eq_dotV1} become
\begin{align}
	\nonumber
	& -X^\top(2\zeta_1 cI - {E_2}^\top E_2)X + 2X^\top(\zeta_1 \Sigma + {E_2}^\top C)\sigma(t,0)\\
	\nonumber
	& - \sigma^\top(t,0)(\zeta_2 I - C^\top C)\sigma(t,0)\\
	&\leq -(2\zeta_1 c - M_1)X^\top X - (\zeta_2 - M_2)\sigma^\top(t,0)\sigma(t,0),
\end{align}
with $M_1 = {\Vert E_2 \Vert}^2 + 1, M_2 = {\Vert \zeta_1 \Sigma + {E_2}^\top	C \Vert}^2 + {\Vert C \Vert}^2$.
The sixth line of \eqref{eq_dotV1} is bounded as follows:
\begin{align}
	\nonumber
	&2\smallint\nolimits_{0}^{1} e^{-\delta x}\psi^\top(t,x)\Sigma^{-1}((F_{21}(x) + F_{22}(x))\sigma(t,x)\\
	\nonumber
	& + \Xi_1(x)X) dx\\
	\nonumber
	& \leq (1 + M_4)\smallint\nolimits_{0}^{1} e^{-\delta x}\psi^\top(t,x)\psi(t,x) dx\\
	& + M_3\smallint\nolimits_{0}^{1} e^{\delta x} \sigma^\top(t,x)\sigma(t,x) dx + X^\top X,
\end{align}
with $M_3 = \mathop{\max}\limits_{x \in [0,1]}{\Vert \Sigma^{-1}(F_{21}(x) + F_{22}(x)) \Vert}^2$ and $M_4 = \mathop{\max}\limits_{x \in [0,1]}{\Vert \left(\Sigma^{-1}\Xi_1(x)\right)^\top \Vert}^2$. The fourth line of \eqref{eq_dotV1} is bounded as follows:
\begin{align}
	\nonumber
	&2\smallint\nolimits_{0}^{1} e^{-\delta x}\psi^\top(t,x)\Sigma^{-1}\smallint\nolimits_{0}^{x} \Xi_2(x,y)\sigma(t,y) dydx\\
	\nonumber
	&\leq \smallint\nolimits_{0}^{1} e^{-\delta x}\psi^\top(t,x)\psi(t,x) dx\\
	& + M_5\smallint\nolimits_{0}^{1} e^{\delta x}\sigma^\top(t,x)\sigma(t,x) dx,
\end{align}
with $M_5 = \mathop{\max}\limits_{x,y \in [0,1]}{\Vert \Sigma^{-1}\Xi_2(x,y) \Vert}^2$. The fifth line of \eqref{eq_dotV1} is also bounded:
\begin{align}
	\nonumber
	&2\smallint\nolimits_{0}^{1} e^{-\delta x}\psi^\top(t,x)\Sigma^{-1}\smallint\nolimits_{0}^{x} \Xi_3(x,y)\psi(t,y) dydx\\
	&\leq M_6\smallint\nolimits_{0}^{1} e^{-\delta x}\psi^\top(t,x)\psi(t,x) dx,
\end{align}
with $M_6 = 1 + \mathop{\max}\limits_{x,y \in [0,1]}{\Vert \Sigma^{-1}\Xi_3(x,y) \Vert}^2$. Thus,
\begin{align}
	\nonumber
	\dot{V}_1 & \leq -(2\zeta_1 c - M_1 - 1)X^\top X - (\zeta_2 - M_2)\sigma^\top(t,0)\sigma(t,0)\\
	\nonumber
	&\quad - \smallint\nolimits_{0}^{1} e^{\delta x}\sigma^\top(t,x)(\zeta_2\delta I - 2\zeta_2\Sigma^{-1}\Omega(x)\\
	\nonumber
	&\quad - M_3I - M_5I)\sigma(t,x) dx\\
	\nonumber
	&\quad + \smallint\nolimits_{0}^{1} e^{-\delta x}\psi^\top(t,x)(M_6I + 2\Sigma^{-1}(F_{21}(x) - F_{22}(x))\\
	&\quad + 2I + M_4I -\delta I)\psi(t,x) dx.
\end{align}
Choosing \scalebox{0.9}{$c = \tfrac{c^{\prime} + 1}{2}$}, \scalebox{0.9}{$\zeta_1 = M_1 + 1$}, \scalebox{0.9}{$\zeta_2 > \max\left\{M_3 + M_5, M_2\right\}$}, \scalebox{0.9}{$\delta > \max\{(2\mathop{\max}\limits_{x \in [0,1]}\Vert F_{21}(x) - F_{22}(x) \Vert + c^{\prime})\Vert \Sigma^{-1} \Vert + M_6 + M_4$}\\\scalebox{0.9}{$ + 2, (c^{\prime} + 2\mathop{\max}\limits_{x \in [0,1]}\Vert \Omega(x) \Vert)\Vert \Sigma^{-1} \Vert + 1\}$} with $c^{\prime} = 2\min\{\delta_1, \delta_2,$\\$\delta_3\} - 1 > 0$, one then obtain
\begin{align}
	\dot{V} \leq -c^{\prime}V.
\end{align}
\noindent
Recalling both direct transformation \eqref{eq_bactr},\eqref{eq_bactr2} and inverse transformation \eqref{eq_invtr},\eqref{eq_invtr2}, applying Cauchy-Schwarz inequality,
we have
\begin{align}
	\nonumber
	&{\Vert p(t,\cdot) \Vert}_{L^2}^{2} + {\Vert q(t,\cdot) \Vert}_{L^2}^{2} + {\Vert r(t,\cdot) \Vert}_{L^2}^{2} + {\Vert s(t,\cdot) \Vert}_{L^2}^{2}\\
	\nonumber
	& + {\Vert u(t,\cdot) \Vert}_{L^2}^{2} + {\Vert v(t,\cdot) \Vert}_{L^2}^{2} + {x_{1}^{2}}(t) + {x_{2}^{2}}(t) + {x_{3}^{2}}(t)\\
	\nonumber
	&\leq K e^{-c^{\prime}t}({\Vert p_0 \Vert}_{L^2}^{2} + {\Vert q_0 \Vert}_{L^2}^{2} + {\Vert r_0 \Vert}_{L^2}^{2} + {\Vert s_0 \Vert}_{L^2}^{2}\\
	& + {\Vert u_0 \Vert}_{L^2}^{2} + {\Vert v_0 \Vert}_{L^2}^{2} + {x_{1}^{2}}(0) + {x_{2}^{2}}(0) + {x_{3}^{2}}(0))
\end{align}
for some positive $K$. Considering
\begin{align}\label{convert}
	\left\{
	\begin{array}{l}
		\alpha_n(t,x) = \tfrac{1}{2}\left[\smallint\nolimits_{0}^{x} (r_n(t,y) + s_n(t,y)) dy + 2x_2(t)\right],\\
		\beta_n(t,x) = \tfrac{1}{2}\left[\smallint\nolimits_{0}^{x} (u_n(t,y) + v_n(t,y)) dy + 2x_3(t)\right],\\
		w_n(t,x) = \smallint\nolimits_{0}^{x} \tfrac{k_2(y)p_n(t,y) + k_1(y)q_n(t,y)}{k_1(y)\cdot k_2(y)} dy + x_1(t),\\
		\alpha_{n,t}(t,x) = \tfrac{r_n(t,x) - s_n(t,x)}{2\sqrt{\mu_1}},
		\beta_{n,t}(t,x) = \tfrac{u_n(t,x) - v_n(t,x)}{2\sqrt{\mu_2}},\\
		w_{n,t}(t,x) = \tfrac{k_2(x)p_n(t,x)-k_1(x)q_n(t,x)}{2\sqrt{\epsilon}(k_1(x)\cdot k_2(x))},\\
		\alpha_{n,x}(t,x) = \tfrac{r_n(t,x) + s_n(t,x)}{2}, \beta_{n,x}(t,x) = \tfrac{u_n(t,x) + v_n(t,x)}{2},\\
		w_{n,x}(t,x) = \tfrac{k_2(x)p_n(t,x) + k_1(x)q_n(t,x)}{2k_1(x)\cdot k_2(x)},
	\end{array}
	\right.
\end{align}
which are obtained from \eqref{eq_riemann}--\eqref{eq_rie}, where functions $k_1(x) = \exp(\sqrt{\epsilon}\bar{c}_1x)$, $k_2(x) = \exp(\sqrt{\epsilon}\bar{c}_2x)$, we thus obtain Theorem \ref{thm:2}.

\end{document}